\documentclass{article}
\usepackage[a4paper, total={6.5in, 10.5in}]{geometry}
\usepackage{graphicx} % Required for inserting images

\usepackage{amsmath}
\usepackage{amssymb}
\usepackage{amsthm}
\usepackage{thm-restate}
\usepackage{physics}
\usepackage{xcolor}
\usepackage{mathtools}
\usepackage{graphicx}
\usepackage{subcaption}
\usepackage{hyperref}
\usepackage[
backend=biber,
style=alphabetic,
]{biblatex}
\usepackage[capitalize]{cleveref}
\hypersetup{
  colorlinks = true,
  linkcolor = blue,
  anchorcolor = blue,
  citecolor = blue,
  urlcolor = blue
  }
\usepackage{comment}
\usepackage{authblk}
\usepackage{mathtools}
\newtheorem{theorem}{Theorem}
\newtheorem{lemma}{Lemma}

\newtheorem{definition}{Definition}

\DeclareMathOperator{\id}{\mathbf{1}}
\usepackage{csquotes}
\usepackage{epigraph} 
\usepackage{multicol}
\usepackage[normalem]{ulem}     %% only for \sout{} for lining out sentences
\usepackage{algorithm}
\usepackage{algpseudocode}
\usepackage{lmodern,tikz,caption}
\usetikzlibrary{arrows.meta}

\DeclareMathOperator{\lpi}{\mathbb{L}_{\textnormal{PI}}}
\DeclareMathOperator{\mqmn}{\mathbb{M}_{\textnormal{QMN}}}

\usepackage{tcolorbox}
\definecolor{teal}{RGB}{50, 136, 140}
\newtcolorbox[auto counter]{pabox}[2][]{fonttitle=\bfseries,
title=Software Box~\thetcbcounter: #2,#1,colframe=gray}
\newtcolorbox[use counter from=pabox]{implementationbox}[2][]{
floatplacement=h,float,
colback=teal!5!white,colframe=teal!75!black,title= Implementation~\thetcbcounter: #2,#1}

\usepackage{blindtext}

\usepackage{enumitem}
\newlist{todolist}{itemize}{2}
\setlist[todolist]{label=$\square$}
\usepackage{pifont}
\algnewcommand{\INPUT}{\item[\textbf{Input:}]}
\algnewcommand{\SET}{\item[\textbf{Set:}]}
\algnewcommand{\OUTPUT}{\item[\textbf{Output:}]}

\title{Local Relaxation Hierarchies for Quantum Ground State Energies: Convergence Guarantees and Message Passing Algorithms}

\author[1]{Sheng-Ku Lin\thanks{These authors contributed equally to this work.}}
\author[2]{Ricardo Rivera Cardoso${}^*$}
\author[1]{Roberto Bondesan}
\affil[1]{Department of Computing, Imperial College London}
\affil[2]{RCQI, Institute of Physics, Slovak Academy of Sciences, Bratislava, Slovakia}

\begin{document}

\maketitle

\begin{abstract}
Convex relaxation hierarchies provide lower bounds to the ground state energy of quantum many-body systems that can be computed in polynomial time on a classical computer, at any fixed hierarchy level. However, scaling these methods to large systems and accurate approximations remains challenging due to the computational cost of traditional solvers and the scarcity of
efficiency guarantees. In this work, we develop local relaxation hierarchies and efficient, highly parallelisable message passing algorithms for estimating the relaxed ground state energies. We show that the first level of the hierarchy---based on local consistency of pairwise reduced density matrices---is exact for commuting Hamiltonians on trees. We further establish that another hierarchy, based on consistent intervals, converges exponentially fast in the interval size to the ground state energy for weak perturbations of separable Hamiltonians on a chain, thereby providing an efficient classical algorithm for these systems. Then, we introduce two variants of message passing algorithms that run in $\mathcal{O}(n/\epsilon^2)$ and $\mathcal{O}(n/\epsilon)$ time for any fixed level of the local hierarchy on bounded-degree graphs, where $\epsilon$ is the precision for the relaxed ground state energy per site. This assumes that the optimal messages have $\mathcal{O}(1)$ norm---a condition we observe in practical settings in our experiments. These algorithms are based on the subgradient method and the Nesterov-type accelerated gradient descent method applied to an entropy-smoothed objective. Finally, we benchmark the message passing algorithms across different quantum Hamiltonians, lattice geometries, and relaxation levels, validating the theoretical predictions and their potential to surpass standard convex optimisation solvers for this problem. We release the resulting library at \url{github.com/rick1924/gse-message-passing}.
\end{abstract}

% \clearpage

% \tableofcontents

% \clearpage

\section{Introduction}

Estimating the ground-state energy of a quantum many-body system is a central task in condensed matter physics and quantum chemistry. This quantity determines, depending on the system, its stability and structure, its low-temperature thermodynamic properties, or chemical reaction rates. Formally, it is the smallest eigenvalue of the system's Hamiltonian $H$, or equivalently the optimal value of the minimisation problem $E_0 = \min_{\rho} \Tr(\rho H)$, taken over the set of all density matrices $\rho$. Although this is a convex optimisation problem, it is hard to solve in general due to the exponential growth of the Hilbert space with the number of particles. In fact, even for local Hamiltonians, deciding their ground state energy up to inverse-polynomial precision is \textsf{QMA}-complete \cite{kitaev2002classical}, meaning that under standard complexity-theoretic assumptions, no classical or quantum algorithm can solve \emph{all} instances of the problem efficiently.

That being said, for several physically relevant Hamiltonians, numerical approximation methods perform surprisingly well. A particular class of algorithms, including tensor networks \cite{DMRG,MERA}, variational quantum Monte Carlo \cite{MonteCarlo}, and the Variational Quantum Eigensolver \cite{VQE}, produce estimates of the ground state energy by optimising over particular subsets of quantum states with efficient representations (or efficiently preparable in the case of VQE), rather than over the set of all quantum states. Since the actual ground state may lie outside the chosen subset, these methods produce upper bounds to the ground state energy. Proving lower bounds---although often overlooked (in comparison with the vast amount of literature on upper bounds)---is equally important if one is to set rigorous error bounds to numerical estimations of the ground state energy. 

In contrast, lower bounds are achieved by \emph{relaxing} the feasible region and outer-bounding the set of quantum states. The result being an optimisation program with an efficient solution.  Most methods in the literature that we are aware of fall within three categories.\footnote{\cite{dualvqe} present the only method we are aware of that can provide a lower bound that is not included in any of the three categories mentioned. Here, they first consider the dual program to ground state energy estimation problem, and then restrict the allowed set of states to those that can be parametrised by a low-depth quantum circuit. It does not result in a hierarchy like the cases we discuss here.} The first method replaces optimisation over a global state with optimisation over locally consistent reduced density matrices $\{\rho_{S_i}\}$, each supported on a subset $S_i \subseteq [n]$ of size $m \leq n$, such that 
the union of all subsets is the whole lattice \cite{poulin2011markov,eisert2026lower,huber2024second,kull2024lower,fawzi2024entropy}. These can be referred to as \emph{generalized Anderson bounds}, as they generalize the Anderson bound \cite{andersonbound} obtained by diagonalizing each term of the Hamiltonian and summing their smallest eigenvalues. The second, rooted in non-commutative polynomial optimisation \cite{NPA}, optimizes over global pseudo-states whose first $k$ moments match those from actual quantum states \cite{nakata2001variational, FirstCLDMrelaxation, fukuda2007large, BarthelFermionicRelaxations, baumgratz2012lower, wang2024certifying,eisert2026lower,kull2024lower,huber2024second}. Depending on the field, this relaxation is known under different names, such as the sum-of-squares hierarchy, quantum Lasserre hierarchy, or the reduced density matrix method. The last method, roughly sits between these two, accurately modelling local descriptions of the system tied together by a global constraint (weaker than the one used in Lasserre's hierarchy) \cite{lin2022variational, khoo2024scalable, li2023quantum}. These receive the name of \emph{variational embeddings}.

Each method can be systematically improved, yielding a \emph{hierarchy} of increasingly tighter relaxations to the set of quantum states, and hence better approximations to the ground state energy. For the generalized Anderson bounds and variational embeddings, this is achieved by increasing the size of the subsets on which the marginals are supported, whereas for the quantum Lasserre hierarchy it is achieved by considering higher moments. The exact solution can be recovered when $m = n$, or $k \to \infty$. Importantly, the memory cost increases exponentially with each level of the hierarchy, yet at a fixed level, each method results in a polynomial size semidefinite program (SDP), which can be solved in polynomial time. While indeed efficient and highly versatile, these programs are typically a black-box algorithm not specialized to the problem at hand, thus leaving room for computational improvements. Although the quantum Lasserre hierarchy can generally be thought of as stronger than variational embeddings, and these in turn as stronger than local marginal consistency (owing to the strength of their constraints), it is not yet well understood how higher levels of the different hierarchies relate to each other. Hence, an exceptionally fast specialised algorithm for any one of these relaxations could unlock higher levels of its hierarchy, potentially yielding a tighter approximation than the other methods can achieve, as they may be computationally prohibited from reaching the levels needed to match it.

Here, our focus is on the local consistency of marginals relaxation, for which we demonstrate convergence guarantees for restricted families of Hamiltonians, as well as an efficient, highly parallelisable ``message-passing'' algorithm for any fixed level of the hierarchy.

A related aspect of our work concerns belief propagation algorithms. The classical sum-product and max-product (also known as min-sum) algorithms compute, respectively, marginals and modes of probabilistic graphical models by iteratively exchanging messages between nodes \cite{JW}. Although originally conceived from the distributive law of commutative semi-rings, which makes them immediately exact on tree graphs, these algorithms have proven to be empirically successful in a variety of settings, including decoding LPDC codes and solving SAT problems \cite{mezard2009information}. This success, together with their parallelisable nature, which allows for highly optimised implementations on dedicated architectures such as GPUs, is the reason for their popularity. Their effectiveness on tree-like graphs and their connection to statistical physics (see \cite{mezard2009information}) are best explained from a variational perspective. Indeed, the sum-product fixed points correspond to stationary points of the Bethe free energy. Importantly, this perspective also helps identify the max-product as the ``zero-temperature limit'' of the sum-product algorithm.

Leifer and Poulin \cite{leifer2008quantum} generalized the sum-product algorithm to the quantum setting, this time to compute quantum marginals, i.e.\ reduced density matrices.\footnote{Technically, the first mention of a quantum belief propagation algorithm was by Hastings \cite{hastingsQBP}, in the context of an algorithm to compute the partition function of 1D systems.} This quantum belief propagation algorithm can be recovered from the dual of a convex relaxation used to estimate the Gibbs free energy $F(T) = \min_{\rho} [ E(\rho) - T S(\rho) ]$, where $T$ is temperature, and $E(\rho), S(\rho)$ are, respectively, the energy and entropy of state $\rho$ \cite{poulin2011markov}.
In this relaxation, the entropy is replaced by a convex approximation, and recently, \cite{scalet2025classicalestimationfreeenergy} has shown that this leads to an efficient classical algorithm for quantum Gibbs sampling at any fixed temperature in one dimension. To go beyond finite temperatures, \cite{CGBP} introduced a modified framework that combines entanglement renormalization techniques with quantum belief propagation; however, as $T \to 0$ most steps of that algorithm reduce to MERA. Furthermore, taking the limit $T \to 0$ in the finite temperature quantum belief propagation equations does not straightforwardly yield a new set of quantum belief propagation equations.
This leaves open a natural question: is there message-passing algorithm that operates directly on the ground state, and what guarantees does it have?

\subsection{Contributions and related work} \label{subsection:contributions}

In this work, we study ground state relaxation based on the local consistency of marginals
and contribute on three accounts.

\paragraph{Exactness on trees for commuting Hamiltonians.} Our first result is related to an open question in the theory asking \emph{``how tight is the relaxation at the $t$-th level of the hierarchy?''}. In classical statistical inference, it is a well-known result that the first level of the relaxation (2-site marginals) is tight when the interaction graph is a tree.\footnote{There, the graph is a graphical model where the nodes represent random variables, and the edges encode the dependencies between them.} While this can not be generally true in quantum theory, we show, by making a connection to the theory of quantum Markov networks and the Structure lemma for commuting hamiltonians \cite{BV}, that this relaxation is tight for 2-local commuting Hamiltonians when the interaction graph is a tree and can thus retrieve the exact ground state energy in polynomial time. Aharonov, Arad, and Irani \cite{aharonov2010efficient} had already presented an efficient algorithm for such Hamiltonians, but our result is an alternative proof of this statement, now from the perspective of semidefinite relaxations. Moreover, our method is simpler, and although it does not result in an efficient representation of the global state, ours still obtains the local marginals, which suffice for computing expectations of local observables. Recent work by Faisal, Natarajan, and Poremba \cite{faisal2026rounding} shows that any 2-local qubit Hamiltonian $H = \sum_{i=1}^m h_i$ with approximately commuting terms ($||[h_i,h_j]|| \leq \epsilon \; \forall i,j$) can be rounded in linear time to an exactly commuting 2-local Hamiltonian $\hat{H}$ such that $| \lambda_{\mathrm{min}}(H)- \lambda_{\mathrm{min}}(\hat{H})| = \mathcal{O}(m \epsilon^{1/6})$. Crucially, the reduction preserves both locality and the interaction graph. Together with our result, this shows that when the graph is a tree and the system is qubits, we can compute $\lambda_{\mathrm{min}}(\hat{H})$ exactly, yielding an $\mathcal{O}(m \epsilon^{1/6})$-approximation to the ground state energy of $H$.

\paragraph{Exponential convergence for perturbations of gapped separable Hamiltonians.}
Our second contribution concerns weakly-interacting chains that are short-range perturbations of non-interacting (i.e.~separable) Hamiltonians with a unique and gapped ground state. In this setting, we  show that the level-$t$ hierarchy defined by considering locally consistent density matrices on intervals of size $t$ converges exponentially fast (in $t$) to the exact ground state energy. This establishes a classical polynomial time and space algorithm to compute the ground state energy in this setting. We remark, however, that our result should not be seen as 
the first efficient classical algorithm for this regime, rather as a guarantee for the validity of local relaxation hierarchies. In fact, \cite{Bravyi_2008} already derived a polynomial time algorithm for the ground state energy of perturbations of gapped separable Hamiltonians in any dimensions using the cluster expansion. Despite our result being specialised to chains only, we expect this to be a limitation of our proof strategy rather than of the algorithm itself, since, 
as we have discussed above, local relaxations are exact for generic commuting Hamiltonians on trees, while the validity of the cluster expansion used in \cite{Bravyi_2008} is limited to perturbations of separable Hamiltonians.
Hastings
\cite{hastings2024perturbationtheorysumsquares,hastings2023fieldtheorysumofsquaresquantum} also discuss the relationships between perturbation theory and ground state relaxations. The main difference is that we restrict to relaxations with a constraint on the support locality, while in those works the author considers the most standard restriction to operators of a fixed degree, as in the Lasserre hierarchy.

\paragraph{Message passing algorithms.} Our last contributions are two distributed (and hence highly parallelisable) algorithms based on message-passing for solving local consistency relaxations, derived by solving the dual optimisation programs. A simple one based on subgradients, and a more sophisticated one based on Nesterov-type accelerated gradient descent method applied to an entropy-smoothed objective. For a particular, physically motivated hierarchy of local consistent marginals, these algorithms run in $\mathcal{O}(n/\epsilon^2)$ and $\mathcal{O}(n/\epsilon)$ time, respectively, assuming the interaction graph of the relevant Hamiltonian has bounded vertex degree and that the optimal messages have $\mathcal{O}(1)$ norm---a condition we verify numerically.
The runtime dependency in $n$ is better with respect to unstructured interior point methods used by standard SDP solvers like CVXPY \cite{SDPs}, but worse in $1/\epsilon$. Moreover, the parallel nature of the algorithm also allows for further runtime improvements. Relating to the last open question of the previous section, our algorithms (beside the fact that they are message-passing) do not possess the other distinctive features of ``belief propagation'' algorithms: they lack the semi-ring structure, and notion of ``beliefs'', and possess different update rules. They instead generalise the classical message passing algorithms of \cite{jojic2010accelerated}. We remark that the local structure of the relaxations considered here is critical in deriving message passing algorithms, and our methods do not apply to Lasserre and variational embeddings, which instead use global constraints. We also note that, as recently pointed out in \cite{li2023quantum}, a cheap and fast algorithm can be useful to identify how to best set up the clusters for subsequent more comprehensive and expensive quantum variational embeddings.

\subsection{Organisation}

\cref{section:framework} discusses marginal sets, including relevant restrictions and relaxations. We demonstrate the equality between two such sets, which is used in the proof of our first result in the following section. Subsequently, in \cref{section:gse}, we formulate the ground state energy problem using the locally consistent relaxations, and argue that for arbitrary Hamiltonians, the amount of resources necessary to produce a good estimate to the ground state energy grows exponentially with the desired precision. We prove that this is not the case for commuting Hamiltonians on a tree, as the first level of the relaxation is already exact, as well as show that this is also not the case for particular weakly-interacting families of Hamiltonians. In \cref{section:qmp}, we derive the subgradient and smooth message-passing algorithms for both 2-local and higher order marginals. Finally, in \cref{section:numerics}, we compile several numerical results supporting arguments scattered throughout the text, as well as compare the performance between our message-passing algorithms and CVXPY.

We note that most theorem proofs are omitted in the main text and are instead presented in \cref{appendix:thmproofs}.

\section{Global and locally consistent marginals} \label{section:framework}

In this section, we introduce the framework needed to describe the marginal relaxations. Although our goal is to consider relaxations to the marginal set, for reasons that will become apparent in the next section, we also consider restrictions to this set.

\subsection{Marginal sets}

To begin, let us briefly introduce some notation. It is convenient to describe a quantum system in terms of an undirected graph $G = (V,E)$. For a vertex $v \in V$, we associate a local Hilbert space $\mathcal{H}_v$ of arbitrary finite dimension. We often refer to this as a \emph{site}. In this section, the edges $E$, encode some sense of locality between sites, and in the next sections, will represent the interactions of 2-local spin Hamiltonians. We denote $\mu_V$ as the quantum state acting on the Hilbert space $\mathcal{H} = \bigotimes_{v \in V} \mathcal{H}_v$, and we let $n = |V|$. Additionally, for a set $S$, we let $D(S) := \{f \, | \, f \subset S\}$ be set of \emph{descendants} of $S$. Finally, throughout the text, we interchangeably refer to reduced density matrices as quantum marginals (or simply marginals whenever it is clear from context that refer to the quantum ones). 

A crucial object in variational relaxations related to local consistency is the set of globally consistent marginals.

\begin{definition}[Quantum marginal set] \label{defn:M}
     Let $G = (V,E)$ be a graph on $n$ vertices, and let $\mu_V \in D(\mathcal{H})$ be a quantum state supported on the vertices of the graph. The \emph{quantum marginal set} associated with $G$ is
     \begin{equation}
         \mathbb{M}(G) = \left\{ (\mu_e, \mu_v)_{e\in E, v \in V} \middle| \;\exists \, \mu_V \text{ s.t. for all } e \in E: \Tr_{V \setminus e}(\mu_V) = \mu_e \text{ and } \Tr_{V \setminus v} (\rho_V) = \mu_v  \; \forall v\in V \right\}.
     \end{equation}
     \noindent
     Here, $(\mu_e, \mu_v)_{e \in E, v \in V}$ is shorthand for the tuple consisting of the marginals $(\mu_{e_1}, \ldots, \mu_{e_{|E|}}, \mu_{v_1}, \ldots, \mu_{v_{|V|}})$.
\end{definition}

Note that in this definition, the positivity and normalization of the density matrices $\mu_{e}$ and $\mu_v$ within the set is implicit from the existence of a global state that is consistent with them.

Unfortunately, providing an efficient and explicit description of $\mathbb{M}(G)$ for an arbitrary graph $G$ is likely not possible. Indeed, the problem of deciding whether a given set of marginals are consistent with a global state is known as the \emph{quantum marginal problem} or \emph{consistency of local density matrices problem}, and is known to be $\mathsf{QMA}$-complete \cite{cldmliu,cldmgrilo}.

As an alternative, a set of marginals that can be described efficiently are those where instead of demanding global consistency, one demands \emph{local consistency} between the marginals.

\begin{definition}[Locally consistent marginal set]
    Let $G = (V,E)$ be a graph on $n$ vertices, The \emph{locally consistent set of marginals} is 

    \begin{equation}
        \mathbb{L}(G) := \left\{ (\mu_e, \mu_v)_{e \in E, v \in V} \;\middle|\; 
        \mu_e \succeq 0, \Tr(\mu_e) = 1 \; \forall e \in E \text{ and }  \mu_v = \Tr_{e \setminus v} (\mu_e) \; \forall v \in V \right\}.
    \end{equation}
    \noindent
\end{definition}

Both $\mathbb{M}(G)$ and $\mathbb{L}(G)$ are convex sets, as convex combinations of its elements result in another element of the set.

In classical statistical inference, it can be shown that for a tree graph $T$, $\mathbb{M}_c(T) = \mathbb{L}_c(T)$ (the classical analogues of these sets defined for marginals diagonal in a fixed product basis) \cite{JW}. In the quantum setting, this is no longer true. To see this, simply consider a path in the tree with three sites and let $\mu_2 = \id /2$ and $\mu_{1,2} = \mu_{2,3} = \ket{\Psi^-} \! \bra{\Psi^-}$ with $\ket{\Psi^-} = 1/\sqrt{2}(\ket{01}-\ket{10})$ be the reduced density matrices supported on the two edges. Clearly, $(\mu_{1,2}, \mu_{2,3},\mu_2, \ldots) \in \mathbb{L}(T)$ since they are positive, normalized, and $\Tr_1(\mu_{1,2}) = \Tr_3(\mu_{2,3}) = \mu_2$. Monogamy of entanglement, however, shows that there is no quantum state in any Hilbert space that is consistent with these marginals, concluding that $\mathbb{M}(T) \neq \mathbb{L}(T)$.

That being said, there are restrictions of these sets that are equal on trees; and although our main interest is on relaxations of the marginal set $\mathbb{M}$, this is relevant for its applications to commuting Hamiltonians. These sets are related to \emph{quantum Markov networks}, whose definition we now recall.

Let us first recall their definition.

\begin{definition}[Quantum Markov networks; \cite{brown2012quantum}]
    Given a graph $G = (V,E)$, $(\rho, G)$ is a \emph{quantum Markov network} if and only if for any three disjoint subsets of particles $A,B,C \subseteq V$ such that $B$ \emph{shields} $A$ from $C$, the conditional mutual information between $A$ and $C$ given $B$ is zero, i.e.\ $I(A:C|B) = 0$. 
\end{definition}

A set $B$ shields $A$ from $C$ if any paths on $G$ from $A$ to $C$ pass through vertices in $B$.
The conditional mutual information, defined as $I(A:C|B) := S(AB) + S(BC) - S(B) - S(ABC)$ quantifies the amount of classical and quantum correlations between systems $A$ and $C$, given that system $B$ has been measured. It satisfies $I(A:C|B) \geq 0$ and is equal to zero iff all correlations between $A$ and $C$ (if any) are mediated through $B$. In a sense, it can be said that $A$ and $C$ are independent given $B$, thus mimicking the classical interpretation of Markov networks.

Then, the restriction of the marginal set to quantum Markov networks is

\begin{definition}
    Given a graph $G$, $\mqmn (G)$ is the set $\mathbb{M} (G)$ except $(\rho_V,G)$ is a quantum Markov network.
\end{definition}
\noindent
It is evident that $\mqmn (G) \subseteq \mathbb{M}(G)$ since $\rho_V$ is from a restricted family of quantum states. We also remark that since the mixture of two quantum Markov states is not generally another such state,  this set is not convex.

For the restriction to the set of locally consistent marginals, consider the following set, which we call \emph{Petz iterable}.

\begin{definition}[$\mathbb{L}_{\mathrm{PI}}(G)$]
    Consider a graph $G = (V,E)$, and assume that the local Hilbert space of each particle $i$ in the bulk decomposes as $\mathcal{H}_i = \bigoplus_{\alpha_i} \bigotimes_{k=1}^{s_i}
    \mathcal{H}_{i \to j_k}^{\alpha_i}$. Here, $\mathcal{H}_{i \to j_k}^{\alpha_i}$ can be understood as the Hilbert space of the subparticle on the edge $(i,j_k) \in E$ of sector $\alpha_i$, and $j_k \in N(i)$ is a neighbour of site $i$ and $s_i = |N(i)|$. $\mathbb{L}_{\mathrm{PI}}(G) := \{(\mu_e, \mu_v)_{e \in E, v \in V}\}$ is the set of locally consistent marginals of a specific form, defined on the edges and nodes of this graph. Because the decomposition applies only to bulk vertices and not leafs, the edge marginals take two different forms depending on whether an edge meets a leaf or joins two bulk vertices.
    
    The leaf marginals are of the form 

    \begin{equation} \label{eqn:leaf_marginals}
        \mu_{\ell,j} = \bigoplus_{\alpha_j} p_{\alpha_j} \mu_{\ell,j \to \ell}^{\alpha_j} \otimes \mu_{\mathrm{comp}(j \to \ell)}^{\alpha_j}
    \end{equation}
    \noindent
    where $\ell$ denotes a leaf vertex, $p_{\alpha_j}$ is a probability distribution, $\mu_{\ell,j \to \ell}^{\alpha_j}$ is the operator supported on $\mathcal{H}_\ell \otimes \mathcal{H}_{j \to \ell}^{\alpha_j}$, and $\mu_{\mathrm{comp}(j \to \ell)}^{\alpha_j}$ is supported on all subspaces of sector $\alpha_j$ except for $\mathcal{H}_{j \to \ell}^{\alpha_j}$.
    
    The bulk marginals, i.e.\ the marginals for edges $(i,j) \in E$ where $i$ and $j$ are bulk vertices, are of the form

    \begin{equation} \label{eqn:bulk_marginals}
        \mu_{i,j} = \bigoplus_{\alpha_i,\alpha_j} t_{\alpha_i,\alpha_j} \mu_{\mathrm{comp}(i \to j)}^{\alpha_i} \otimes \mu_{i \to j, j \to i}^{\alpha_i,\alpha_j} \otimes \mu_{\mathrm{comp}(j \to i)}^{\alpha_j},
    \end{equation}
    \noindent
    where $\mu_{i \to j, j \to i}^{\alpha_i,\alpha_j}$ has support only on $\mathcal{H}_{i \to j}^{\alpha_i} \otimes \mathcal{H}_{j \to i}^{\alpha_j}$. Here, $t_{\alpha_i, \alpha_j}$ is a joint probability distribution.

    The single-site node marginals are defined via the partial trace of these edge marginals.
\end{definition}

Our first main result, proved in \cref{appendix:thmproofs}, is about the equality of these last two sets. 

\begin{restatable}{theorem}{LequalsT} \label{thm:lpi}
    Let $T = (V,E)$ be a tree. Then $\lpi (T) = \mqmn (T)$.
\end{restatable}

The key idea behind this result is the decomposition and reconstruction properties of states with zero conditional mutual information, applied to quantum Markov networks. More specifically, \cite{QMN} shows that for a state $\rho_{ABC}$ where $A,B,C \subseteq V$ and $B$ shields $A$ from $C$, the Hilbert space of $B$ decomposes into $\mathcal{H}_B = \bigoplus_{\alpha} \mathcal{H}_{B \to A}^{\alpha} \otimes \mathcal{H}_{B \to C}^{\alpha}$ and the state in these three systems can be expressed as $\rho_{ABC} = \bigoplus_{\alpha} p_{\alpha} \rho_{A, B \to A}^{\alpha} \otimes \rho_{B \to C,C}^{\alpha}$, where $p_{\alpha}$ is a probability distribution. At the same time, they show that given the bipartite marginals, the \emph{Petz recovery map} \cite{petzmap1,petzmap2} exactly recovers the tripartite state. Showing that $\mqmn(T) \subseteq \lpi(T)$ then is nothing but accounting for all of the decompositions at the ``edge-level'' of a quantum Markov network. The other direction requires showing that given the two-site marginals of \cref{eqn:leaf_marginals,eqn:bulk_marginals}, the Petz map can be iterated through all edges of the tree to yield a global quantum Markov network state.

\subsection{Outer-bounds}

Earlier, we showed that $\mathbb{M}(G) \subseteq \mathbb{L}(G)$; however, $\mathbb{L}(G)$ generally results in a loose outer bound. It is possible to systematically define tighter approximations to the marginal set by adding more variable constraints. There are many ways of defining these sets, some which are more practically motivated than others, and may also result in slightly different hierarchies. Here we choose one which is best for our distributed algorithm of \cref{section:qmp}. 

\begin{definition} [Higher-order local consistency sets]
    For all $t \geq 2$, the $t$-th local consistency marginal set is 
    
    \begin{equation} 
        \mathbb{L}_{t} := \left\{ (\mu_S, \mu_f)_{S \subseteq V, f \in D(S)} \; \middle| \;  |S| = t, \mu_S \succeq 0, \Tr(\mu_S) = 1 \; \forall S \subseteq V  \text{ and } \mu_f = \Tr_{S \setminus f}(\mu_S) \; \forall f \in D(S) \right\},
    \end{equation}
    \noindent
    noting that the shared descendants of two subsets $S$ and $S'$ are assumed to correspond to the same variable.
\end{definition}

As special cases, we also let $\mathbb{L}_1(G) := \mathbb{L}(G)$ so the ``natural'' local consistency set is the 1st level of the hierarchy, and let $\mathbb{L}_0(G)$ be the set of edge and node marginals on graph $G$ where local consistency is not required. As we soon demonstrate, these definitions illustrate why this hierarchy (in the context of ground state energy estimation) is known as \emph{generalized Anderson bounds}.

Note that for a given $t \geq 2$, each element within $\mathbb{L}_t$ is a tuple of size $\sum_{k =1}^t \binom{n}{k} \in \mathcal{O}(n^t)$, and hence also has a similar number of constraints. For any constant $t$, these sets can then be described efficiently, yet, in practice, this is not a favourable scaling for large $n$. We return to this issue in \cref{subsection:clustering}. 

Finally, to compare these sets at different orders, it is necessary to project them unto the same ``coordinates'' since currently, the elements of $\mathbb{L}_t(G)$ and $\mathbb{L}_{t'}(G)$ are of different size for $t \neq t'$. We let $\mathbb{L}_t(G)$ be the sets whose coordinates are projected to the edges and nodes specified by some graph $G$. Then, $\mathbb{M}(G) = \mathbb{L}_n(G) \subseteq \ldots \subseteq \mathbb{L}_0(G)$.

\section{Relaxations to the ground state energy problem} \label{section:gse}

In this section we introduce the hierarchy of semidefinite relaxations based on the locally consistent marginal sets. We discuss the tightness of the approximations and show that the 1st level of the hierarchy is tight for commuting Hamiltonians on a tree, which as mentioned in \cref{subsection:contributions} can be can be used to provide $\epsilon$-precision estimates to approximately commuting Hamiltonians, based on the results of \cite{faisal2026rounding}.

\subsection{A hierarchy of relaxations}

Consider a system of $n$ $d$-dimensional particles and let the interactions between them be given by a $2$-local Hamiltonian $H = \sum_{e\in E} h_e + \sum_{v \in V} h_v$. From a variational perspective, the ground state energy of the Hamiltonian is given by

\begin{equation}
    \begin{aligned}
        E_0 & = \min_{\rho \in D(\mathcal{H})} \langle H, \rho \rangle \\
        & = \min_{\mu \in \mathbb{M}(G)} \sum_{e \in E} \; \langle h_e, \mu_e \rangle + \sum_{v \in V} \langle h_v, \mu_v \rangle
    \end{aligned}
\end{equation}
\noindent
where $D(\mathcal{H})$ is the set of quantum states on the Hilbert space $\mathcal{H} = (\mathbb{C}^d)^{\otimes n}$, $\mu$ is shorthand for the tuple $(\mu_e, \mu_v)_{e \in E, v \in V}$ defined earlier, and for two operators $A,B$, we have let $\langle A,B\rangle=\Tr(A^\dagger B)$. In the second line, we observed that since the hamiltonian terms $h_e$ and $h_v$ are local operators the energy can also be estimated locally if using marginals $(\mu_e, \mu_v)$ that are consistent with a global state.

We can obtain a hierarchy of approximations to this quantity by relaxing the global consistency constraint and require local marginal consistency only. Then, 

\begin{equation}
    f_t^* = \min_{\mu \in \mathbb{L}_t (G)} \sum_{e \in E} \; \langle h_e , \mu_e \rangle + \sum_{v \in V} \langle h_v, \mu_v \rangle
\end{equation}
\noindent
is a lower bound to the ground state energy, and from the relationship of sets $\mathbb{L}_t(G)$, it is clear that

\begin{equation} \label{eqn:gs_hierarchy}
    E_0 = f_n^* \geq \ldots \geq f_0^*.
\end{equation}

Furthermore, observe that the approximation achieved via the Anderson bound, $f_{\mathrm{Anderson}}:= \sum_i \lambda_{\mathrm{min}}(h_i)$, is equivalent to $f_0^*$ since

\begin{equation}
    \begin{aligned}
    f_{\mathrm{Anderson}} &= \sum_{e \in E} \lambda_{\mathrm{min}}(h_e) + \sum_{v \in V} \lambda_{\mathrm{min}}(h_v)\\
    &= \sum_{e \in E} \min_{\substack{\mu_e: \\ \mu_e \succeq 0 \\ \Tr[\mu_e] = 1}} \langle \mu_e, h_i \rangle + \sum_{v \in V} \min_{\substack{\mu_v: \\ \mu_v \succeq 0 \\ \Tr[\mu_v] = 1}} \langle \mu_v, h_v \rangle  \\
    &= \min_{\mu \in \mathbb{L}_0(G)} \sum_{e \in E} \langle \mu_e, h_e \rangle + \sum_{v \in V} \langle \mu_v, h_v \rangle  \\
    &= f_0^*
    \end{aligned}
\end{equation}
\noindent
where the second equality follows from rewriting the minimum eigenvalue as a minimization problem, and the next equality holds for independent marginals, which is the case for marginals in $\mathbb{L}_0$. This then shows why the hierarchy of \cref{eqn:gs_hierarchy} is also known as generalized Anderson bounds. However, as we discuss below, this is not the only such hierarchy.

Since for any constant $t$, the optimisation program is convex, has a polynomial number of constraints, and the variables are matrices of size $d^t \times d^t$, a semidefinite optimisation program can produce an $\epsilon$-approximation to $f_t^*$ in $\mathrm{poly}(n, \log(1/\epsilon))$ time \cite{SDPs}.

Each level of the hierarchy is efficiently computable, yet, an open question is how well the $t$-th level of the hierarchy approximates the actual ground state energy. We now elaborate on cases where the first few levels of the relaxation can be shown to be exact for some restricted families of Hamiltonians. Later, we will discuss the error incurred by truncating the hierarchy to level $t$, and how $t$ needs to scale to achieve a desired estimation error.

\subsubsection{Exact relaxations}

Somewhat trivially, the first example where the relaxation is tight is for frustration-free Hamiltonians. A Hamiltonian $H = \sum_{e \in E} h_e$ is \emph{frustration-free} if the ground state simultaneously minimizes each of the $h_e$ terms independently. Consequently, the energy of this ground state is $E_0 = \sum_{e \in E} \lambda_{\text{min}}(h_e) \equiv f_{\mathrm{Anderson}}$ obtaining the smallest energy possible. Consequently, for any $t \geq 1$, the relaxation is tight. 

The second example is where the Hamiltonian is diagonal in a fixed product basis and its interaction graph is a tree. The ground state energy is then given by

\begin{equation}
    \min_{(\mu_e) \in \mathbb{M}_c(T)} \sum_{e \in E} \langle h_e, \mu_e \rangle = \min_{(\mu_e) \in \mathbb{L}_c(T)} \sum_{e \in E} \langle h_e, \mu_e \rangle
\end{equation}
\noindent
where the equality follows since $\mathbb{M}_c(T) = \mathbb{L}_c(T)$ \cite{JW}.

The third example, and the more interesting case where the relaxation is tight is the case of commuting Hamiltonians on a tree. Concretely, we show:

\begin{restatable}{theorem}{commutingtree} \label{thm:commuting}
    Let $T = (V,E)$ be a tree, and let $H = \sum_{(i,j) \in E} h_{i,j}$ be a commuting Hamiltonian on $n = |V|$ qudits of local dimension $d$, one on each vertex of $T$, and the $h_{i,j}$ are hermitian operators supported on the edges of the tree. The optimization program

    \begin{equation} \label{eqn:minimisation_commuting}
        f^*_1 = \min_{\mu \in \mathbb{L}(T)} \sum_{(i,j) \in E} \; \langle \mu_{i,j}, h_{i,j} \rangle \, ,
    \end{equation}
    \noindent
    yields the exact ground state energy, and the reduced density matrices at optimum $\{\mu^*_{i,j}\}$ are consistent with a global quantum Markov network ground state.
\end{restatable}

Roughly, this statement holds due to the specific objective being minimized, which results in the $\mu^*_{i,j}$ having the structure as the marginals described in the set $\lpi$, which per \cref{thm:lpi}, necessarily correspond to the marginals of a global quantum Markov network state. In other words, the relaxation always results in a valid global quantum state and the optimization program is hence exact. As pointed out in the proof of this result in \cref{appendix:thmproofs}, the marginals obtain this decomposition due to the Structure lemma of commuting local Hamiltonians \cite{BV}. This lemma, based on the theory of representations of $C^*$-algebras, explains the decomposition of the local Hilbert spaces into direct sums and tensor products of smaller invariant subspaces from the perspective of the interaction algebra, rather than as a characterization of states as in \cite{QMN}.

This theorem also implies that the ground state of a commuting Hamiltonian on a tree is a quantum Markov network. Among other results connecting commuting Hamiltonians to quantum Markov networks, Brown and Poulin \cite{brown2012quantum} also demonstrated a more general version of this fact for higher-order local Hamiltonians and in arbitrary graphs. Yet, attempting to take advantage of this fact directly to obtain the ground state energy through variational optimisation does not work since $\mqmn (T)$ is generally not a convex set. Instead, one has to optimise over the set of all locally consistent matrices on the tree and the optimisation naturally results in marginals from the set $\lpi$.

\cref{thm:commuting} also provides an alternative proof to Theorem 2 of \cite{aharonov2010efficient}, demonstrating an algorithm to efficiently construct the ground state of a 1D commuting Hamiltonian. However, our result provides significant advantages: our algorithm is simpler and does not require computing the local subspace decompositions, and while their algorithm does provide the whole state (as a matrix product state), ours returns the local marginals of the ground state. For computing expected values of local observables, having the reduced density matrices suffices. If one still desires the density matrix of the whole state, it can be reconstructed, albeit not efficiently, using the Petz recovery map. This is partly demonstrated in the proof of \cref{thm:lpi}.

To discuss other cases in which the relaxation is not exact, but where one can nevertheless provide guarantees on the approximation error to the ground-state energy, we introduce a new class of hierarchies in which marginals are defined only for \emph{clusters} of sites that are in close proximity. As we discuss in more detail below, these hierarchies offer several practical advantages. In particular, they are physically well motivated by the fact that most Hamiltonians in condensed matter theory are geometrically local and their correlations are also often localized in space. Consequently, correlations can be captured using local marginals, while still yielding accurate approximations to the ground state energy compared with optimisation over the full marginal set ($\mathbb{L}_t$), which includes marginals between widely separated sites. Restricting to this \emph{geometrically local hierarchy} also substantially reduces the number of constraints, leading to more efficient optimisation.

\subsection{Geometrically local hierarchies}

We begin by first introducing the concepts of clusters and cluster families. A \emph{cluster family} is a set $C = \{c_1, \ldots, c_m\}$ where each $c_i \subseteq V$ is a \emph{cluster} and the family generally satisfies the following two conditions: (1) $C$ is a cover of $V$, i.e.\ $\bigcup_{i = 1}^m c_i = V$, and (2) each cluster $c_i$ contains at least one edge of $G$. It is often convenient to think of the clusters $c_i$ as hyperedges of a hypergraph $K = (V,C)$. Any hypergraph can be associated to a \emph{poset} diagram where nodes are hyperedges and the partial ordering is inclusion, so we draw an arrow $e\to f$ if $f\subset e$; see \cref{fig:poset chain 3} for an example.
Recall also the definition of descendants
$D(e) := \{f \in C\, | \, f \subset e\}$ and ancestors  $\mathcal{A}(e) = \{ f \in C\,|\, f\supset e\}$.

Given a cluster family, the optimisation program then considers marginals supported on the clusters of this family, requiring local consistency on their overlaps. Concretely, this set of marginals is defined as follows. 

\begin{definition}
    Given a cluster family $C = \{c_1, \ldots, c_m\}$ such that conditions (1) and (2) above are true, the set of locally consistent marginals defined over these clusters is 

    \begin{equation}
        \mathbb{L}^C = \left\{(\mu_c, \mu_f)_{c \in C, f \in D(c)} \; \middle| \;  \mu_c \succeq 0, \Tr(\mu_c) = 1, \text{ and } \Tr_{c \setminus f}(\mu_c) = \mu_f \; \forall f \in D(c) \right\} \, .
    \end{equation}
\end{definition}

As before, $\mathbb{L}^C(G)$, denotes its projection to the graph $G$. The optimisation program that is then solved in practice is:

\begin{equation} \label{eqn:cluster_relaxation_energy}
    f^*_C = \min_{\mu \in \mathbb{L}^C(G)} \sum_{e \in E} \langle \mu_e, h_e \rangle + \sum_{v \in V} \langle \mu_v, h_v \rangle \, .
\end{equation}

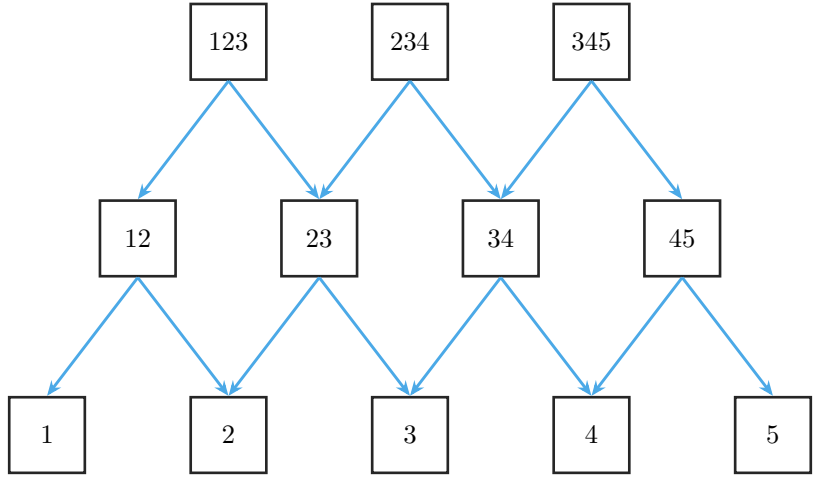
\begin{figure}
    \centering
\begin{tikzpicture}[
x=2.4cm,y=2.6cm,
vertex/.style={draw=black!85,line width=1pt,rectangle,minimum size=1cm,inner sep=0pt,font=\sffamily},
inclusion/.style={draw={rgb,255:red,75;green,169;blue,231},line width=1.1pt,-{Stealth[length=2mm,width=1.6mm]}}]
\foreach \i in {1,...,5}
\node[vertex] (v\i) at ({\i-1},0) {$\i$};
\node[vertex] (e12) at (0.5,1) {$12$};
\node[vertex] (e23) at (1.5,1) {$23$};
\node[vertex] (e34) at (2.5,1) {$34$};
\node[vertex] (e45) at (3.5,1) {$45$};
\node[vertex] (e123) at (1,2) {$123$};
\node[vertex] (e234) at (2,2) {$234$};
\node[vertex] (e345) at (3,2) {$345$};
\foreach \parent/\child in {e123/e12,e123/e23,e234/e23,e234/e34,e345/e34,e345/e45}
\draw[inclusion] (\parent.south) -- (\child.north);
\foreach \parent/\child in {e12/v1,e12/v2,e23/v2,e23/v3,e34/v3,e34/v4,e45/v4,e45/v5}
\draw[inclusion] (\parent.south) -- (\child.north);
\end{tikzpicture}
    \caption{Poset diagram for a hypergraph with vertices $V = \{1,2,3,4,5\}$ and hyperedges $E=\{(123),(234),(345),(12),(23),(34),(45),(1),(2),(3),(4),(5) \}$.}
    \label{fig:poset chain 3}
\end{figure}

Now it becomes clear why conditions (1) and (2) are necessary: without them, the objective above would be impossible to evaluate. Additionally, adding the condition $\max_{i} |c_i| \leq t$, we can define $\mathbb{L}_t^C$, and then the projection of this set to marginals on edges and nodes of graph $G$, can be written as $\mathbb{L}_t^C(G)$. 
If $C = \{S \subseteq V: |S| = t\}$ is the cluster family containing all clusters of size $t$, then $\mathbb{L}^C = \mathbb{L}_t$. Hence, this notation is quite general.

Note that while it is true that $\mathbb{L}^C_t(G) \subseteq \mathbb{L}_{t}(G)$, it is not true that $\mathbb{L}^C_t(G) \subseteq \mathbb{L}^C_{t'}(G)$ for $t \geq t'$. Hence, these sets do not form a hierarchy. 
Note that while $\mathbb{L}^C(G) \subseteq \mathbb{L}_{t}(G)$ whenever $\max_i|c_i| \leq t$, it is not true in general that $\mathbb{L}^{C'}(G) \subseteq \mathbb{L}^C(G)$ for $\max_i |c_i'| \geq \max_j |c_j|$. Consequently, $f^*_{C'}  \not \geq f^*_C$ merely because $C'$ has larger clusters. Intuitively, large clusters spanning distant sites can perform worse than smaller clusters of nearby sites, since the latter's marginals more faithfully capture ground-state correlations, which are typically local. 

However, it is still possible to establish a hierarchy if one more methodically increases the cluster support. 
The general idea is to use a cluster family $C$ such that each cluster is a connected subgraph of $G$. This defines the base level of the hierarchy. For the next level one may choose to increase all or a subset of the clusters by one adjacent site. In general then, at level $t$, one increases all or a subset of the clusters by one, from the previous state of the clusters. For these clusters, it is then trivial that $\mathbb{L}^{C_{t}} \subseteq \mathbb{L}^{C_{t'}}$ for $t' \geq t$, and hence 

\begin{equation}
    E_0 = f^*_{C_{+m}} \geq \ldots \geq f_{C_{+1}}\geq f^*_C \, ,
\end{equation}
\noindent
where $m = n - \max_i |c_i|$. In other words, $C_{+m}$ is the cluster family where one or more of its hyperedges are of size $n$. 

One example of such a hierarchy that we use next, is based on intervals. (We define it for a chain but it can be extended to lattices of higher dimension). For $2\le t\le n$, denote by $\mathcal{I}_t$ the set of intervals of size $t$ in a chain of $n$ sites, together with their subsets. We define the level-$t$ hierarchy consisting of density matrices $\mu_I$ supported on $I\in\mathcal{I}_t$ satisfying the local consistency conditions

\begin{align} \label{eqn:interval_set_defn}
    \mathbb{Y}_t
    &:=
    \left\{
    (\rho_I)_{I\in\mathcal{I}_t}\,\middle|\, 
    \rho_I \succeq 0\,,\quad
    \Tr(\rho_I)=1\,,\quad 
    \Tr_{I\setminus J}(\rho_I) 
    =\rho_J \,,\text{ for all }J\subseteq I
    \right\} .
\end{align}

With these definitions at hand, we not discuss some cases where convergence guarantees associated with these geometrically local relaxations exist.

\subsection{Other bounds: exponential resource cost} \label{sec:other bounds}

In general, we expect the computational resources of this hierarchy to grow exponentially with the desired inverse error to the ground state energy per site. We present a simple argument for this, and then discuss relevant literature showing clever tricks to ease or sidestep this cost.

To see this exponential cost explicitly, it is enough to consider a chain:

\begin{align}
    H = \sum_{i=1}^{n-1}h_{i,i+1}\,.
\end{align}
Denote the optimiser of the level-$t$ relaxation by $\mu_S^*$, where set $S$ is a cluster consisting of at most $t$ sites. Then, assume that $n = tJ$ for some $J$ and partition the chain into $J$ disjoint and adjacent clusters of $t$ sites: $\{ (1:t), (t+1:2t), \dots \}$, where $(i:j)\equiv (i,i+1,\dots, j)$. Then we define the rounding procedure to a global state by taking the product of the optimum density matrices on these disjoint sets:
\begin{align}
    \sigma = 
    \mu^*_{1:t}\otimes \mu^*_{t+1:2t}
    \otimes \cdots \otimes 
    \mu^*_{(J-1)t+1:Jt}
    \,.
\end{align}
By construction, $\sigma_{i,i+1}=\mu^*_{i,i+1}$
if $(i,i+1)$ is within one of these disjoint clusters and 
$\sigma_{i,i+1}=\mu^*_i\otimes \mu_{i+1}^*$ if $i,i+1$ is across two neighbouring partitions. We denote the set of such cuts by $D:=\{(t,t+1), (2t,2t+1), \dots\}$.
Since $\sigma$ is a valid quantum state, we have $E_0 \le \Tr(H\sigma)$.
Thus, subtracting $f_t^*$ from both sides of this inequality:
\begin{align}
    E_0 - f^*_t \le
    \sum_{i=1}^{n-1}
    \left(\Tr(h_{i,i+1}\sigma_{i,i+1})
    -\Tr(h_{i,i+1}\mu^*_{i,i+1})\right)
    =
    \sum_{i\in D}
    \left(\Tr(h_{i,i+1}\sigma_{i,i+1})
    -\Tr(h_{i,i+1}\mu^*_{i,i+1})\right) .
\end{align}
Now note that $
\lambda_{\text{min}}(h)\le 
\Tr(h\rho)\le \lambda_{\text{max}}(h)$
and that $\lambda_{\text{max}}(h)-\lambda_{\text{min}}(h)\le
|\lambda_{\text{max}}(h)|+|\lambda_{\text{min}}(h)|\le 
2\| h \|$, and assume that 
$2\|h\|\le 1$.
So we obtain for the approximation to the energy per site
\begin{align}
    \frac{E_0 - f^*_t}{n} \le
    \frac{1}{n}
    \sum_{i\in D}
(\lambda_{\text{max}}(h_{i,i+1})-\lambda_{\text{min}}(h_{i,i+1}))
\le
\frac{J}{n}
=
\frac{1}{t}\,.
\end{align}

If we want the approximation to this quantity to be $\epsilon$ close, we need $t=\lceil 1/\epsilon \rceil$, and since the time and memory complexity of the classical simulation algorithm grow polynomially in $n$ but exponentially in $t$, they will grow exponentially in $1/\epsilon$ as well.

The recent works of \cite{eisert2026lower} and \cite{kull2024lower} also observe this exponential growth of resources in slightly different settings and develop clever tricks to practically reduce this cost. We now briefly elaborate on the work of \cite{eisert2026lower} which introduces a different geometrically local hierarchy, reinforcing our claim of their importance.  \cite{kull2024lower} is briefly discussed in \cref{section:outro}.

In \cite{eisert2026lower}, Eisert strengthens the Anderson bound by considering local marginal consistency in translationally invariant systems. For a translationally invariant geometrically local Hamiltonian on a $D$-dimensional
cubic lattice, the standard Anderson bound partitions the lattice into finite patches of $m^D$ sites. It lower-bounds the ground state energy by minimizing the Hamiltonian terms supported on each patch ($h_m$) independently and summing the results. As shown in \cite{eisert2026lower}, the error of this bound in the energy per site decreases only polynomially with the patch size, being governed by the terms crossing the patch boundaries. This is the same boundary effect as in our argument above. Hence, to reach precision $\epsilon$ the patch size $m$ must grow polynomially in $1/\epsilon$. The Hilbert space dimension of a patch, $d^{m^D}$, and with it the cost of direct diagonalization, therefore grows exponentially in $1/\epsilon$.

The key observation of \cite{eisert2026lower} is that the Anderson bound neglects consistency between neighbouring patches. Each patch is optimized independently, even though in a global state their reduced density matrices must be mutually compatible. Eisert exploits this consistency to obtain the bound of a larger patch without diagonalizing it directly. For simplicity, we restrict the following discussion to a chain ($D=1$). To approximate $\lambda_{\mathrm{min}}(h_{2m})$, he optimizes over two disjoint marginals of size $m$, supported on $(1:m)$ and $(m+1:2m)$. He models the correlations between these two patches with an additional marginal of variable size $2s$, supported on $(m-s+1:m+s)$ and required to be consistent with both. The resulting optimal values satisfy $f_{m,s} \leq f_{m,t} \leq f_{m,m} =
\lambda_{\mathrm{min}}(h_{2m})$ for all $t = s, \ldots, m$. The largest variable has dimension $d^{\max(m,2s)}$, so for $2s \le m$ one obtains a lower bound on the $2m$-site Anderson bound using only $m$-site marginals.

For any fixed $m$, this construction defines a new geometrically local hierarchy. Here $s$ indicates the hierarchy level, and the cluster family is $C = \{(1:m), (m-s+1:m+s), (m+1:2m)\}$. Note, however, that increasing $s$ converges to the Anderson bound of a $2m$-site patch rather than to the ground state energy itself. The accuracy is therefore still controlled by $m$: the trick reduces the size of the largest marginal from $2m$ to $m$ sites, but reaching precision $\epsilon$ still requires $m$ to grow polynomially in $1/\epsilon$, and hence a cost exponential in $1/\epsilon$.

Next, we will make assumptions on the Hamiltonian to improve on this result and guarantee exponential convergence of the local hierarchy.

\subsection{Exponential convergence of the local hierarchy for weakly-interacting chains}

We consider the ground state problem for the following Hamiltonian on a chain of 
$n$ qudits:
\begin{align}
    \label{eq:H0tV}
    H = H_0 + \lambda V\,,
    \quad 
    H_0 = \sum_{i=1}^n h_i
    \,,\quad 
    V = \sum_{i=1}^{n-1}v_{i,i+1}
\end{align}
Since $H_0$ is separable, its
ground state is a product state and we assume it to be non-degenerate and of energy zero.
We also assume that there is a gap of at least one with the first excited state of $H_0$, the bounds on the operator norms 
\begin{align}
    \|h_i\|\le h_* \,,\quad \|v_{i,i+1}\|\le 1\,,
\end{align}
and take $\lambda>0$ since the sign of $\lambda$ can be absorbed in the definition of $V$.

Note that the current setup encompasses general geometrically-local perturbations of separable  Hamiltonians. 
Indeed, the choice of nearest-neighbour interaction is without loss of generality for short range Hamiltonians 
since we can always coarse grain the lattice so that interactions are among neighbours at the expense of an increased dimensionality of the qudits. Further, if the on-site gap is $\Delta$ and if $\|v_{i,i+1}\|\le J$, we can simply divide $H$ by $\Delta$ and reabsorb $J$ in $\lambda$ to bring the Hamiltonian to the form assumed here. 

We prove exponential convergence of the local hierarchy  $\mathbb{Y}_t$ for the weakly-interacting setting.

\begin{theorem}\label{thm:exp_conv_1d_weak}
Given the Hamiltonian $H$ satisfying the assumptions above, 
denote 
by $E_0$
its ground state energy and 
by $f_t^*$
the energy of the relaxation where feasible density matrices are in $\mathbb{Y}_t$.
Then, there exists a system-size independent $\lambda_*>0$ such that for all $0<\lambda<\lambda_*$,  
\begin{align}
    0\le 
    \frac{E_0-f^*_t}{n}
    \le 
    c(\lambda) e^{-t/8}
    \,,
\end{align}
where $c(\lambda)$ is a
positive function of $\lambda$ independent of the system size.
\end{theorem}

Thus, we need to take $t = \mathcal{O}(\log(1/\epsilon))$ to achieve $\epsilon$ error in estimating the ground state energy per site. This implies a time and space complexity of the classical algorithm that solves the relaxation which is polynomial in $1/\epsilon$, in contrast to the general case, where the complexity is exponential in $1/\epsilon$ as remarked in Section \ref{sec:other bounds}.

The proof of this theorem is detailed in Appendix \ref{app:proof exp_conv_1d_weak} and we here sketch the main ideas.
The derivation is based on a block diagonalisation result of \cite{froehlich2019lieschwingerblockdiagonalizationgappedquantum}. In this work, the authors construct a unitary $U$ as a product of gates, one per interval  of the chain, such that $UHU^\dagger$ is block-diagonal with respect to the ground state of $H_0$ for sufficiently small $\lambda$. In the proof of \cref{thm:exp_conv_1d_weak}, we first show how this result can be interpreted as a sum-of-squares decomposition, $H-E_0\id=\sum_\gamma P_\gamma$, $P_\gamma\succeq 0$, where the sum is over intervals but $P_\gamma$ is supported over the whole chain. Exploiting the locality structure of $U$ and Lieb-Robinson bounds \cite{hastings2010localityquantumsystems,Nachtergaele_2019}, we show that the $P_\gamma$'s quasi-local, exponentially decaying around intervals. Note, however, that the expectation value functional $\omega_t(A)=\Tr(\mu_IA)$ is defined only for operators $A$ supported on intervals $I$ of size at most $t$, and sums thereof. To overcome this, we then introduce an extension of $\omega_t$ on the space of all operators, defined by applying $\omega_t$ on truncations of operators to intervals of length at most $t$. We finally apply this extension to the sum-of-squares decomposition of $H-E_0\id$ and show that, thanks to the quasi-locality of $P_\gamma$, $E_0 -\omega_t(H)$ decays exponentially with $t$, thus establishing the theorem.

\section{Distributed algorithm} \label{section:qmp}

In this section, we derive the distributed algorithm, first, for a simpler version where marginals are supported on the edges and nodes of the graph, and in \cref{subsection:clustering} we generalise the algorithm for marginals supported on larger sites.

In the discussions below, we restrict ourselves to geometrically local Hamiltonians acting on at most two sites, but we remark that it can be readily extended to short-range Hamiltonians acting on more than two sites by grouping sites together at the expense of a bigger on-site dimensionality.

\subsection{Two-local marginals}

Consider a graph $G=(V,E)$ and a 2-local qudit Hamiltonian $H = \sum_{e\in E}h_e + \sum_{i\in V} h_i$ with interactions supported on the edges and nodes of the graph. The first-order relaxation to the ground state energy is then

\begin{align} \label{eqn:gse_2_local}
    f_1^* = \min_{\mu\in\mathbb{L}(G)}
    \sum_{e\in E} \; \langle\mu_e,h_e\rangle+
    \sum_{i\in V} \; \langle \mu_i, h_i\rangle\,,
\end{align}
which can be efficiently solved by a semidefinite program to obtain a lower bound to the exact ground state energy $E_0$. 

We derive the distributed algorithm by exploiting \emph{duality}, which we now describe. For a constrained optimization program $\min_x f(x)$ with constraints $c_i(x) \leq b_i$ $(i = 1, \ldots, m)$, the \emph{Lagrangian} $L(x, \nu)$  is a function that incorporates both the original objective function and constraints into a single new objective, with new variables $\vec \nu = (v_1, \ldots, v_m)$ called \emph{Lagrange multipliers}. These can be thought of as penalty terms that negatively affect the value of the objective whenever the constraints are not satisfied. Solving the original (primal) problem is then equivalent to $\min_x \max_\nu L(x,\nu)$, and the \emph{dual program} swaps the order of the min and max: $\max_\nu \min_x L(x,\nu)$. As a result, the dual program generally results in a lower bound to the primal problem, yet, is a powerful transformation since the \emph{dual function} $Q(\nu) = \min_x  L(x,\nu)$ is always concave. Fortunately, under the mild condition that a strictly feasible point exists---a point $x$ satisfying $c_i(x) < b_i$ for all non-affine $c_i$ (and $c_i(x) \leq b_i$ for any affine $c_i$) for all $i \in [m]$---the optimal value of the primal and dual coincide and the bound is tight. This condition is known as \emph{Slater's condition}, and if met, the problem is said to satisfy \emph{strong duality}.

The minimisation program of \cref{eqn:gse_2_local} satisfies strong duality since $\mu_e = \id /d^2$ and $\mu_v = \id/d$ (for all $e \in E$ and $v \in V$) is a strictly feasible point.\footnote{All equality constraints imposed by $\mathbb{L}(G)$ are affine and can be written as inequalities of the required form, so they need only be satisfied, not necessarily strictly. The positivity constraint $-\mu \preceq 0$ does require strict feasibility (as it is not affine in the polyhedral sense, but with respect to the PSD cone), and $\mu_e = \id/d^2$, $\mu_v = \id/d$ satisfy this, being positive definite.} Instead of incorporating all constraints into the Lagrangian, we consider a \emph{partial} dualisation incorporating only the compatibility constraints. The ones for positivity and normalization of the marginals will be considered manually within the algorithm. The Lagrange multipliers lie in the same space as the constraints, and so for each constraint $\Tr_k(\mu_{k, \ell}) = \mu_\ell$ we introduce the Lagrange multiplier $\nu_{k \to \ell} \in \mathrm{Herm}(\mathbb{C}^{d \times d})$. By strong duality, our dual program is exact and yields $f_1^*$ .

To be concrete, our partial Lagrangian is then

\begin{align*}
    L(\mu,\nu)
    &=
    \sum_{e\in E}\langle\mu_e,h_e\rangle+\sum_{i\in V}\langle\mu_i, h_i\rangle + \sum_{(i,j) \in E} \langle v_{i \to j}, \Tr_i(\mu_{i,j}) - \mu_j\rangle + \langle v_{j \to i}, \Tr_j(\mu_{i,j}) - \mu_i \rangle \\
    &=
    \sum_{e\in E}\langle\mu_e,h_e\rangle+\sum_{i\in V}\langle\mu_i, h_i\rangle +\sum_{(i,j)\in E}
    \langle \mu_i,\nu_{j\to i}\rangle - 
    \langle\mu_{i,j}, \nu_{j\to i}\otimes \id\rangle
    +
    \langle \mu_j,\nu_{i\to j}\rangle - 
    \langle \mu_{i,j} ,\id\otimes \nu_{i\to j}\rangle \\
    &=
    \sum_{(i,j)\in E}\langle
    \mu_{i,j},
    h_{i,j}-\nu_{j\to i}\otimes \id
    -\id\otimes \nu_{i\to j}\rangle 
    +
    \sum_{i\in V}
    \left\langle\mu_i ,
    h_i
    +\sum_{j\in N(i)}\nu_{j\to i}
    \right\rangle \, ,
\end{align*}
\noindent
where $N(i)$ is the set of neighbours of $i$. The dual function $Q(\nu)=\min_{\mu} L(\mu, \nu)$ can be readily computed since $L$ is separable in $\mu$'s:

\begin{align}
\label{eqn:dual_function_non-smooth}
    Q(\nu)
    =
    \sum_{(i,j)\in E}\lambda_{\text{min}}
    (
    h_{i,j}-\nu_{j\to i}\otimes \id
    -\id\otimes \nu_{i\to j}) 
    +
    \sum_{i\in V}
    \lambda_{\text{min}}
    \left(
    h_i
    +\textstyle\sum_{j\in N(i)}\nu_{j\to i}
    \right)
    \,.
\end{align}
\noindent
where $\lambda_{\text{min}}(\cdot)$ refers to the minimum eigenvalue of the matrix within. This function is concave but non-smooth in general due to the possible multiplicity of the minimum eigenvalues. We solve this minimization problem in two different ways, each with their own trade-offs, but both resulting in distributed ``message-passing'' algorithms.

\subsubsection{Subgradient method}

Since the function is non-smooth, considering solving the equation via gradient descent instead results in subgradient descent. As a consequence, the objective does not decrease with every iteration, but generally does trend downward, and can be shown to converge for specific choices of hyperparameters. We first describe the algorithm and then discuss its convergence. 

Recall that the subgradient of a function $f$ at $x$ is any vector $g$ such that $f(y)\ge f(x) + \langle g, x-y\rangle$ for all $y$. We denote the set of subgradients by $\partial f(x)$. The simplest subgradient method to minimise $f$ updates $x$ as: $x\leftarrow x - \alpha g$ with $\alpha$ being a small step size.

Now, we apply the subgradient method to maximise $Q(\nu)$, i.e.~minimise $-Q(\nu)$. The $j\to i$ component of a subgradient of $Q(\nu)$ is

\begin{equation}
    -\Tr_j (\mu^*_{i,j}) + \mu^*_i  \, ,
\end{equation}
\noindent
where
\begin{equation}
    \mu^*_{i,j}\in \mathcal{G}(h_{i,j}-\nu_{j\to i}\otimes \id
    -\id\otimes \nu_{i\to j})
\,,\quad 
\mu_i^*
\in\mathcal{G}\left(
    h_i
    +\sum_{j\in N(i)}\nu_{j\to i}
    \right)\,,
\end{equation}
\noindent
and $\mathcal{G}(h)$ denotes the set of ground states of $h$.
The update of Lagrange multiplier $\nu_{j \to i}$ is hence
\begin{align}
    \nu_{j\to i}
    \leftarrow 
    % \nu_{j\to i}
    % -
    % \alpha 
    % (\Tr_j \mu^*_{(ij)} - \mu^*_i)
    % =
    \nu_{j\to i}
    +
    \alpha 
    \left[ \mu^*_i-\Tr_j ( \mu^*_{i,j}) \right]
    \,.
\end{align}

\begin{algorithm}[t]
\caption{Constant step size subgradient descent}
\label{alg:subgradient}
\begin{algorithmic}[1]
\INPUT A Hamiltonian $H = \sum_{i \in V} h_i + \sum_{(i,j) \in E} h_{i,j}$ supported upon on the nodes and edges of a graph $G = (V, E)$, number of iterations $T \in \mathbb{N}$, and step size $\alpha \in \mathbb{R}$.
\SET $\boldsymbol \nu^{(0)} \coloneqq \bigoplus_{i \in V} \bigoplus_{j \in \mathcal N(i)} \nu^{(0)}_{j \to i} = \mathbf{0}$, $q_{\mathrm{best}} = - \infty$.
\vspace{0.2cm}
\For{$k = 0, 1, \ldots, T-1$}
    \ForAll{site $i \in V$ and edge $(i,j) \in E$} \Comment{This step can be done in parallel}
        \State $\widehat H_i(\boldsymbol \nu^{(k)}) \gets h_i + \sum_{j \in \mathcal N(i)} \nu^{(k)}_{j \to i}$
        \State $\mu^{(k)}_{i} \gets$ $\ket{\psi_i } \bra{\psi_i}$, where $\ket{\psi_i}$ is a sampled zero-energy eigenvector of $\widehat H_i(\boldsymbol \nu^{(k)})$;
        \vspace{0.2cm}
        \State $\widehat H_{i,j}(\boldsymbol \nu^{(k)}) \gets 
        h_{i,j}
        - \nu^{(k)}_{j \to i} \otimes \id_j
        - \id_i \otimes \nu^{(k)}_{i \to j}$
        \State ${\mu}_{i,j}^{(k)} \gets$ $\ket{\psi_{i, j}} \bra{\psi_{i, j}}$, where $\ket{\psi_{i, j}}$ is a sampled zero-energy eigenvector of $\widehat H_{i,j}(\boldsymbol \nu^{(k)})$
    \EndFor

    % ==================================
    % I added the following section as the method really outputs the best multipliers after T iterations  
    \vspace{0.2cm}
    \State $q^{(k)} \gets
    \sum_{(i,j)\in E} \lambda_{\min} \left( \widehat H_{i,j}(\boldsymbol \nu^{(k)} ) \right) + \sum_{i\in V} \lambda_{\min} \left( \widehat H_i(\boldsymbol \nu^{(k)}) \right)$

    \If{$q^{(k)}>q_{\mathrm{best}}$}
    \State $q_{\mathrm{best}}\gets q^{(k)}$
    \State $\boldsymbol\nu_{\mathrm{best}}\gets\boldsymbol\nu^{(k)}$
    \EndIf
    % ==================================
    
    \ForAll{directed edge $j \to i$} \Comment{This step can be done in parallel}
    \vspace{0.1cm}
        \State $\nu_{j \to i}^{(k+1)} \gets \nu_{j \to i}^{(k)} + \alpha \left [\mu_i^{(k)} - \Tr_j(\mu_{ij}^{(k)}) \right]$
    \EndFor
\EndFor

% \State \Return $\boldsymbol \nu^{(T)}$
\State \Return $(\boldsymbol \nu_{\mathrm{best}}, q_{\mathrm{best}})$, where $q_{\mathrm{best}} = \max_{0 \le t < T} Q(\boldsymbol \nu^{(t)})$
\end{algorithmic}
\end{algorithm}

The algorithm iterates this procedure until the number of maximum iterations is reached, or other practically motivated condition is met (discussed in more detail below). The procedure is condensed in \cref{alg:subgradient}.

Observe that each multiplier update can be performed independently from all others using only local information coming from its neighbour. This information can be thought of as a ``message'', and can hence be thought of as a distributed message passing algorithm. Note that even though our algorithm is designed to estimate the ground state energy, it also yields the optimal marginals as byproduct.

% This algorithm can be interpreted as the quantum version of the classical min-sum (or alternatively max-product) algorithm.\footnote{The min-sum algorithm seeks a mode $x$ that maximises (or minimises) a joint probability distribution $p(x)$ \cite{JW}}. 

Regarding the algorithmic complexity, observe that an iteration of the algorithm simply consists of diagonalizing a $d \times d$ matrix for each multiplier and performing basic arithmetic. In a given instance, there are $2 |E|$ Lagrange multipliers, which in the worst-case scales as $\mathcal{O}(n^2)$ with the number of sites. But since we consider mainly graphs with a bounded degree, in most cases we have $|E| = \mathcal O (n)$. The question then is: how many iterations are required for convergence?

By adapting the proof technique of \cite{boyd2003subgradient} to our setting (see \cref{appendix:section_qmp}), we show that:

\begin{restatable}[Convergence guarantee for subgradient descent]{theorem}{convergencesubgradient}\label{thm:convergence_subgradient_complete}
    Given a target error density $\epsilon$, 
    \cref{alg:subgradient} 
    with the input step size $\alpha = c \epsilon$, where $0 < c < 1/\Delta_G$ is a constant independent of both $\epsilon$ and $n$
    and $\Delta_G \coloneqq \max_{i \in V} |N(i)|$, 
    produces an objective value $f_{\mathrm{subgrad}} \coloneqq \max_{0 \le t < T} Q \left ( \nu^{(t)} \right)$ 
    such that $\dfrac{f^*_1 - f_{\mathrm{subgrad}}}{n} \leq \epsilon$ 
    after $T$ iterations, for any integer $T \ge \left \lceil \dfrac{R^2}{2 c \left (1 - c \Delta_G \right) n \epsilon^2} \right \rceil$, 
    where $R \coloneqq \mathrm{dist} \left(\nu^{(0)}, \arg \max_\nu Q \left( \nu \right) \right)$. 
    Each iteration in \cref{alg:subgradient}
    requires $\mathcal O \left (nd^3 + |E|d^6 \right)$ arithmetic operations, where $d$ is the dimension of the local Hilbert space per site. %\footnote{\todo{Mention something like: under certain condition, such as when we know the spectral gap, we can extract the ground-state eigenvector with only $O \left ((d^2)^{2} \right)$}}
\end{restatable}

Notice that the above convergence bound depends on $R \coloneqq \mathrm{dist} \left( \nu^{(0)}, \arg \max_\nu Q \left (\nu \right) \right)$, the distance between the initial Lagrange multipliers used by \cref{alg:subgradient} and the dual optimal set. In general, this is a difficult value to bound, as it requires a priori estimates on the norm of the optimal multipliers. In our numerical experiments (see \cref{fig:multiplier_independence}), however, we have observed that the Frobenius norms of individual optimal multipliers show little dependence on system size $n$ over the instances studied. We therefore assume that the individual optimal multiplier norms to be uniformly $\mathcal O(1)$, and obtain $R = \mathcal O(\sqrt{n})$ for bounded-degree graphs with zero initialisation.

For bounded-degree graphs with fixed on-site Hilbert-space dimension $d$, we also have $\Delta_G = \mathcal O(1)$ and $|E| = \mathcal O(n)$. Each iteration costs $\mathcal O(n)$ as $d$ is treated as a constant. \cref{thm:convergence_subgradient_complete} then gives the following corollary.

\begin{restatable}{corollary}{corollarysubgradient}
    For bounded-degree graphs and $\mathcal O (1)$ optimal message norms, the target error density $\epsilon$ can be achieved in $\mathcal O \left (1 / \epsilon^2 \right)$ subgradient iterations, independently of the system size $n$, and the total arithmetic complexity scales as $\mathcal O \left (n / \epsilon^2 \right)$.
\end{restatable}

\subsubsection{Accelerated smoothing method}

The subgradient descent method described in the previous section incurs the penalty of a $\mathcal O \left(\epsilon^{-2} \right)$ dependency. This is not a unique problem, as standard naive subgradient methods in classical nonsmooth convex optimisation have the same $\mathcal O \left (\epsilon^{-2} \right)$ worst-case iteration bound. To improve this quadratic error scaling, we adopt the two-step approach optimisation method of \cite{jojic2010accelerated} from the classical literature, where we first smooth the dual function so that there exists a unique and well-behaved gradient for each of our local dual objective, before performing gradient descent with a Nesterov-type accelerated variation introduced in \cite{nesterov1983method}. We show that by combining the two approaches and splitting the total error budget between the approximation error as the result of smoothing, which can be shown to be bounded, and the optimisation error, we can achieve an improved error-density dependency of $\mathcal O \left (\epsilon^{-1} \right)$. %, with \todo{moderate caveats which we will discuss later} in subsequent sections.

The subsequent construction follows from the original dual function $Q \left (\nu \right)$ as stated in \cref{eqn:dual_function_non-smooth}. We begin by adding a strongly concave term to each local minimum eigenvalue problem $\lambda_{\mathrm{min}} \left (\cdot \right)$, the von Neumann entropy $S(\rho) = - \Tr [\rho \log \rho]$ scaled by inverse temperature $\beta$. Recall that the von Neumann entropy $S(\rho)$ is $1-$strongly concave with respect to the trace norm, and $2-$strongly concave with respect to the Frobenius norm by the quantum Pinsker's inequality. We define the \emph{soft} local minimum eigenvalue function as \begin{align}
    F_\beta \left (H \right) 
    &\coloneqq \min_{\rho \succeq 0, \Tr[\rho]=1} \left \{ \Tr \left [\rho H \right] - \frac{1}{\beta} S \left (\rho \right) \right\} \\
    &= - \frac{1}{\beta} \log \Tr \left[ e^{- \beta H} \right],
\end{align}
\noindent 
which proves to be equivalent to the equilibrium free energy of local Hamiltonians at inverse temperature $\beta$. Effectively, we are transforming the local ground state energy problem $\lambda_{\min} \left (\cdot \right)$ into its finite-temperature equilibrium free energy approximation $F_\beta \left (\cdot \right)$, where the inverse temperature $\beta$ acts as both the control of the approximation error and the \emph{smoothing parameter} of the local dual objectives that we can tune.

% \todo{Mention its well-behaved gradient and its Lipschitz constant here? $\| \nabla F_\beta(K) - \nabla F_\beta(K') \|_F \le \frac{\beta}{m_F} \| K - K' \|_F.$ To justify the term ``smoothing''.}

The smoothed dual function is then 
\begin{equation}
    Q_\beta \left (\nu \right) = \sum_{(i,j) \in E} F_\beta \left( \widehat H_{i,j} \left (\nu \right) \right) 
    + \sum_{i \in V} F_\beta \left( \widehat H_{i} \left(\nu \right) \right),
\end{equation}
\noindent
where $\widehat H_{i,j}\left (\nu \right) = h_{i,j} - \nu_{j \to i} \otimes \id_j - \id_i \otimes \nu_{i \to j}$, $\widehat H_{i} \left(\nu \right) = h_{i} + \sum_{j \in \mathcal N(i)} \nu_{j \to i}$ are the effective edge and node Hamiltonians, respectively. Replacing the original dual function with its smoothed counterpart introduces approximation error, but we now show that it can be bounded in terms of $\beta$, the interaction structure, and the local Hilbert space dimensions, as stated in the following lemma. 

% \addRicardo{I believe here is a good place to put Lemma 1, what do you think? We can connect it to the previous with something like:} Replacing the original dual for a smoothed one introduces error, yet it can be bounded as follows.

\begin{restatable}[Bound on smoothing error]{lemma}{smootherror}
\label{lem:bound_on_smoothing_error_2_local}
For any feasible $\nu$, 
    \begin{align}
        Q_\beta \left(\nu \right) \leq Q \left (\nu \right) \leq Q_\beta \left (\nu \right) + \frac{\Gamma_G}{\beta},
    \end{align} 
    \noindent
    where $\Gamma_G \coloneqq \sum_{(i,j) \in E} \log (d_i d_j) + \sum_{i \in V} \log (d_i)$ and $d_i = \dim(\mathcal H_i)$ is the dimension of the local Hilbert space at site $i$.
\end{restatable}
\noindent
The proof of this lemma is listed in \cref{appendix:section_qmp}.

With a now concave and smooth dual function $Q_\beta \left(\nu \right)$, the gradients are unique and well-defined. For the soft local minimum eigenvalue function $F_\beta \left (H \right)$, the gradient with respect to $H$ is
\begin{align}
    \nabla_H F_\beta \left (H \right) &= \rho_\beta \left (H \right) = \frac{e^{- \beta H}}{\Tr \left[e^{- \beta H} \right]}.
\end{align} 
\noindent
Therefore, the original local ground state subgradients are replaced by Gibbs states at inverse temperature $\beta$. The $j \to i$ component of the gradient of $Q_\beta \left (\nu \right)$ is thus
\begin{equation}
    \nabla_{\nu_{j \to i}} Q_\beta \left(\nu \right) = {\rho_\beta}_{i}(\nu) - \Tr_j \left[ {\rho_\beta}_{ij} \left (\nu \right) \right],
\end{equation} 
\noindent
where \begin{align}
    {\rho_\beta}_{ij} \left (\nu \right) = \frac{ \exp (- \beta \widehat H_{i,j}(\nu))}{\Tr \left[ \exp(- \beta \widehat H_{i,j}(\nu)) \right]}, \; 
    {\rho_\beta}_{i}\left (\nu \right) = \frac{\exp(- \beta \widehat H_{i}(\nu)) }{\Tr \left[ \exp (- \beta \widehat H_{i}(\nu)) \right]}
\end{align}
\noindent
are the Gibbs states of the corresponding effective edge and node Hamiltonians.

For the second stage of our two-step adaptation of \cite{jojic2010accelerated}'s method, we perform a Nesterov-type accelerated gradient descent method, instead of a naive gradient descent, to maximise our now-smoothed dual function $Q_\beta \left (\nu \right)$. Introduced in \cite{nesterov1983method}, Nesterov's  accelerated gradient descent is a provably-optimal first-order method that can yield a $\epsilon$-solution to a differentiable and convex function $f$ in $\mathcal O (\sqrt{L/\epsilon})$ number of iterations, where $L$ is the Lipschitz constant of $\nabla f$. For our setting, we let $f = - Q_{\beta} (\nu)$ and split the error-density budget $\epsilon$ equally between the smoothing approximation error and the optimisation error. We choose $\beta = 2 \Gamma_G / \left (n \epsilon \right)$ and $\widehat L_\beta = \beta \Lambda_G/2$, where $\Lambda_G \coloneqq \Delta_G + d$ with $\Delta_G \coloneqq \max_{i \in V} | N(i) |$ and $d$ the dimension of the local Hilbert space at each site, and run Nesterov's method for $\mathcal O \left(\sqrt{2 \widehat L_\beta / \left (n \epsilon \right )} \right)$ number of iterations. The method is summarised in \cref{alg:accelerated-smoothed-quantum-dual}. 

The convergence guarantee of the method is stated in the following theorem, which we prove in \cref{appendix:section_qmp}.

\begin{restatable}[Convergence guarantee for accelerated smoothing]{theorem}{convergencesmoothedgradient}\label{thm:convergence_smoothed_gradient_complete}
    Given a target error density $\epsilon$, 
    \cref{alg:accelerated-smoothed-quantum-dual} 
    with inputs 
    $\beta = 2 \Gamma_G/ \left ( n \epsilon \right)$ 
    % $\beta = \dfrac{2 \Gamma_G}{n \epsilon_{\mathrm{den}}}$ 
    and 
    $\widehat L_\beta = \beta \Lambda_G/2$ 
    % $\widehat L_\beta = \dfrac{\beta \Lambda_G}{2}$ 
    produces an objective value $f_{\mathrm{grad}} \coloneqq Q_\beta \left ( \nu^{(T)} \right)$ 
    such that $ \dfrac{f^*_1 - f_{\mathrm{grad}}}{n} \leq \epsilon$ 
    after $T$ iterations, for any integer
    $T \ge \left \lceil \dfrac{2 R_\beta \sqrt{\Gamma_G \left(\Delta_G + d \right) }}{n \epsilon } \right\rceil - 1$, 
    where $R_\beta \coloneqq \mathrm{dist} \left( \nu^{(0)}, \arg \max_\nu Q_\beta \left (\nu \right) \right)$, 
    $\Gamma_G \coloneqq \sum_{(i,j)\in E} \log(d^2) + \sum_{i \in V} \log (d)$,
    $\Delta_G \coloneq \max_{i \in V} | N(i) |$ and $d$ is the uniform on-site dimension.
    Each iteration in \cref{alg:accelerated-smoothed-quantum-dual}
    requires $\mathcal O \left (nd^3 + |E|d^6 \right)$ arithmetic operations.
\end{restatable}

Consider a family of bounded-degree graphs with fixed on-site Hilbert-space dimension $d$. Then $\Delta_G = \mathcal O(1)$, $|E| = \mathcal O(n)$, and $\Gamma_G = \mathcal O(n)$. Each iteration costs $\mathcal O(n)$ as $d$ is treated as a constant. Assume further that, for the same reason as with the subgradient case, the optimal Lagrange multipliers $\arg \max_\nu Q_\beta  \left (\nu \right)$ have components with $\mathcal O(1)$ norm that is independent of the system size $n$ and the smoothing parameter $\beta = 2 \Gamma_G / \left (n \epsilon \right)$ under consideration. Then $R_\beta = \mathcal O(\sqrt n)$ uniformly with zero initialisation. \cref{thm:convergence_smoothed_gradient_complete} then gives the following corollary.

\begin{restatable}{corollary}{corollarygradient}
    For bounded-degree graphs and $\mathcal O(1)$ optimal message norms, the target error density $\epsilon$ can be achieved in $\mathcal O(1 / \epsilon)$ accelerated gradient iterations, independently of the system size $n$, and the total arithmetic complexity scales as $\mathcal O(n / \epsilon)$.
\end{restatable}

\begin{algorithm}[t]
\caption{Accelerated gradient descent for the smoothed quantum dual}
\label{alg:accelerated-smoothed-quantum-dual}
\begin{algorithmic}[1]
\INPUT A Hamiltonian $H = \sum_{i \in V} h_i + \sum_{(i,j) \in E} h_{i,j}$ supported upon on the nodes and edges of a graph $G = (V, E)$, smoothing parameter $\beta > 0$, Lipschitz bound $\widehat L_\beta$, number of iterations $T \in \mathbb N$.
\SET $\boldsymbol \nu^{(0)} \coloneqq \bigoplus_{i \in V} \bigoplus_{j \in \mathcal N(i)} \nu^{(0)}_{j \to i} = \mathbf{0}$, $\boldsymbol \zeta^{(0)} \coloneqq \bigoplus_{i \in V} \bigoplus_{j \in \mathcal N(i)} \zeta^{(0)}_{j \to i} = \mathbf{0}$ and $\theta_0 = 1 \in \mathbb R$.
\vspace{0.2cm}
\For{$k = 0, 1, \ldots, T-1$}
    \State $\boldsymbol \eta^{(k)} \gets (1-\theta_k) \boldsymbol \nu^{(k)} + \theta_k \boldsymbol \zeta^{(k)}$
    \ForAll{site $i \in V$ and edge $(ij) \in E$} \Comment{This step can be done in parallel}
        \State $\widehat H_i(\boldsymbol \eta^{(k)}) \gets h_i + \sum_{j \in \mathcal N(i)} \eta^{(k)}_{j \to i}$
        \vspace{0.2cm}
        \State ${\rho_\beta}_i^{(k)} \gets 
        \dfrac{\exp(- \beta \widehat H_i(\boldsymbol \eta^{(k)}))}
        {\Tr \left[ \exp(- \beta \widehat H_i(\boldsymbol \eta^{(k)})) \right]} \; ;$
        \vspace{0.2cm}
        \State $\widehat H_{ij}(\boldsymbol \eta^{(k)}) \gets 
        h_{ij}
        - \eta^{(k)}_{j \to i} \otimes \id_j
        - \id_i \otimes \eta^{(k)}_{i \to j}$
        \vspace{0.2cm}
        \State ${\rho_\beta}_{ij}^{(k)} \gets 
        \dfrac{\exp(- \beta \widehat H_{i,j}(\boldsymbol \eta^{(k)}))}
        {\Tr \left[ \exp(- \beta \widehat H_{i,j}(\boldsymbol \eta^{(k)})) \right]}$
    \EndFor
    \ForAll{directed edge $j \to i$} \Comment{This step can be done in parallel}
        \State
        $g_{j\to i}^{(k)}
        \gets
        {\rho_\beta}_i^{(k)} - \Tr_j \left[ {\rho_\beta}_{ij}^{(k)} \right]$
    \EndFor
    
    \State $\boldsymbol \zeta^{(k+1)} \gets \boldsymbol \zeta^{(k)} + \dfrac{1}{\theta_k \widehat L_\beta} \boldsymbol g^{(k)}$
    \State $\boldsymbol \nu^{(k+1)} \gets (1-\theta_k) \boldsymbol \nu^{(k)} + \theta_k \boldsymbol \zeta^{(k+1)}$
    \State $\theta_{k+1} \gets \dfrac{\sqrt{\theta_k^4+4\theta_k^2}-\theta_k^2}{2}$
\EndFor
\State \Return $\boldsymbol \nu^{(T)}$
\end{algorithmic}
\end{algorithm}

Thus, under mild assumptions, accelerated gradient descent with objective smoothing improves the iteration bound for achieving a target error density of $\epsilon$ from $\mathcal O(\epsilon^{-2})$ of naive subgradient descent to $\mathcal O(\epsilon^{-1})$, while maintaining the worst-case per-iteration arithmetic cost.

% \todo{Emphasise that all messages will stay traceless if initialised as traceless, which is crucial for obtaining the Lipschitz constant upper bound.}

% \todo{Insert discussion on the cost of the algorithm per iteration. Point out that, perhaps somewhat surprisingly, that per-iteration cost for subgrad and grad is the same, cubic in the dimension of the matric being sampled ground-state eigenvalue/computed Gibbs state.}

% \todo{Stress that parts of the algorithm can be computed in parallel? ``Distributed''}

\subsection{Higher-order marginals} \label{subsection:clustering}

We now discuss our algorithms applied to higher-order relaxations, in particular for any cluster family $C$. These ideas are typically referred to as clustering methods, augmented hypergraph methods, or generalised belief propagation in the classical case \cite{JW}.

The program we are trying to evaluate is the dual of \cref{eqn:cluster_relaxation_energy}, resulting in an approximation to the ground state of $f_C^*$. To construct the partial Lagrangian, for every hyperedge $e$, we introduce a Lagrange multiplier $\nu_{e \to f}$ for all $f \in \mathcal{D}(e)$. Each of these has support on $f$. The Lagrangian is then 

\begin{align*}
    L(\mu,\nu)
    &=
    \sum_{e\in E}\langle \mu_e, h_e\rangle 
    +
    \sum_{e\in E}\sum_{f\in \mathcal{D}(e)}
    \Tr\left[\left(\mu_f-\Tr_{f^c}(\mu_e) \right)\nu_{e\to f}\right] \label{eqn:main_lagrangian}
    \\
    &=
    \sum_{e\in E}\langle \mu_e, h_e\rangle 
    +
    \sum_{e\in E}\sum_{f\in \mathcal{D}(e)}
    [
    \langle \mu_f ,\nu_{e\to f}\rangle
    -
    \langle \mu_e, \id_{f^c}\otimes\nu_{e\to f}\rangle]
    \\
    &=
    \sum_{e\in E}\left\langle \mu_e, h_e
    -
    \sum_{f\in \mathcal{D}(e)}
    \id_{f^c}\otimes\nu_{e\to f}
    \right\rangle 
    +
    \sum_{e\in E}\sum_{f\in \mathcal{D}(e)}
    \langle \mu_f ,\nu_{e\to f}\rangle
    \\
    &=
    \sum_{e\in E}\left\langle \mu_e, h_e
    -
    \sum_{f\in \mathcal{D}(e)}
    \id_{f^c}\otimes\nu_{e\to f}
    \right\rangle 
    +
    \sum_{f\in E}\sum_{e\in \mathcal{A}(f)}
    \langle \mu_f ,\nu_{e\to f}\rangle
    \\
    &=
    \sum_{e\in E}\left\langle \mu_e, h_e
    -
    \sum_{f\in \mathcal{D}(e)}
    \id_{f^c}\otimes\nu_{e\to f}
    \right\rangle 
    +
    \sum_{e\in E}\sum_{g\in \mathcal{A}(e)}
    \langle \mu_e ,\nu_{g\to e}\rangle
    \\
    &=
    \sum_{e\in E}\left\langle \mu_e, h_e
    -
    \sum_{f\in \mathcal{D}(e)}
    \id_{f^c}\otimes\nu_{e\to f}
    +
    \sum_{g\in \mathcal{A}(e)}\nu_{g\to e}
    \right\rangle 
    \,.
\end{align*}
\noindent
where most equalities follow simply from relabelling and rearranging the elements of the sums, and the equality from the third to the fourth line follows since there is a term $\langle \mu_f ,\nu_{e\to f}\rangle$ for each $e\in E$ and $f\in \mathcal{D}(e)$, or equivalently for each $f\in E$ and $e\in \mathcal{A}(f)$.

Defining the local augmented Hamiltonian as
\begin{align}
    \theta_e(\{\nu_{e\to f},
    \nu_{g\to e}\})
    \coloneqq 
    h_e
    -
    \sum_{f\in \mathcal{D}(e)}
    \id_{f^c}\otimes\nu_{e\to f}
    +
    \sum_{g\in \mathcal{A}(e)}\nu_{g\to e}
    \, ,
\end{align}
\noindent
the dual function can be written as
\begin{align}
    Q_C(\nu) 
    =
    \min_{\mu} L(\mu,\nu)
    =
    \sum_{e\in E}\lambda_{\text{min}}
    (
    \theta_e(\{\nu_{e\to f},
    \nu_{g\to e}\})
    ) . 
\end{align}

Before elaborating on how this function can be solved via both subgradient and smooth methods, we remark that in practice, only the constraints between consecutive levels in the poset diagram are needed. To see why, consider the example shown in \cref{fig:poset chain 3} and observe that for edge $(1,2)$, we must enforce the constraints: $\Tr_2 (\mu_{1,2}) = \mu_1$ and $\Tr_1 (\mu_{1,2}) = \mu_2$. Then, constraints like $\Tr_{2,3} (\mu_{1,2,3}) = \mu_1$ are unnecessary as they are implied by both $\Tr_3 (\mu_{1,2,3}) = \mu_{1,2}$ and $\Tr_2 (\mu_{1,2}) = \mu_1$.

\subsubsection{Subgradient method}

For any $e \in E$ where $f \in \mathcal{D}(e)$, the subgradient is
\begin{align}
    \partial_{\nu_{e\to f}} Q_C
    \ni 
    \mu_f^* - \Tr_{f^c}\mu_e^* \, , 
\end{align}
where $\mu_e^*$ is a ground state of $\theta_e(\{\nu_{e\to f},
    \nu_{g\to e}\})$. Thus, the update rule is
\begin{align}
    \nu_{e\to f}
    \leftarrow 
    \nu_{e\to f}
    +
    \alpha 
    (\mu^*_f-\Tr_{f^c} \mu^*_{e})
    \,.
\end{align}

The convergence result is summarised in the following theorem.

\begin{restatable}[Convergence guarantee for higher-order subgradient descent]{theorem}{convergencesubgradienthigherorder}\label{thm:convergence_subgradient_complete_higherorder}
    Given a target error density $\epsilon$, 
    the higher-order variant of \cref{alg:subgradient} 
    with the input step size $\alpha = c \epsilon$, where $0 < c < 1/\kappa_G$ is a constant independent of both $\epsilon$ and $n$
    and $\kappa_G \coloneqq \max_{i \in V} \left | \left\{ e \to f: e \in E, f \in \mathcal D(e), i \in f \right\} \right|$, 
    produces an objective value $f_{\mathrm{subgrad}} \coloneqq \max_{0 \le t < T} Q \left ( \nu^{(t)} \right)$ 
    such that $\dfrac{f^*_1 - f_{\mathrm{subgrad}}}{n} \leq \epsilon$ 
    after $T$ iterations, for any integer $T \ge \left \lceil \dfrac{R_C^2}{2 c \left (1 - c \kappa_G \right) n \epsilon^2} \right \rceil$, 
    where $R_C \coloneqq \mathrm{dist} \left(\nu^{(0)}, \arg \max_\nu Q_C \left( \nu \right) \right)$. 
    Each iteration in the algorithm
    requires $\mathcal O \left ( \sum_{e \in E} D_e^3 \right)$ arithmetic operations, where $D_e = \dim (\mathcal H_e)$ for $e \in E$.
\end{restatable} 
\noindent
Its proof technique follows similarly to that of \cref{thm:convergence_subgradient_complete}.

Consider a family of bounded-degree hypergraphs with uniform on-site Hilbert-space dimension $d$, bounded maximum cluster size $k$, and bounded local incidence. It follows that $\kappa_G = \mathcal O(1)$, $|E| = \mathcal O (n)$, and $\sum_{e \in E} |\mathcal D(e)| = \mathcal O(n)$. Assume further that $R_C = \mathcal O(\sqrt{N})$. \cref{thm:convergence_subgradient_complete_higherorder} then gives the following corollary.

\begin{restatable}{corollary}{corollarysubgradienccccct}
    For bounded-degree hypergraphs, bounded maximum cluster size and $\mathcal O (1)$ optimal message norms, the target error density $\epsilon$ can be achieved in $\mathcal O \left (1 / \epsilon^2 \right)$ subgradient iterations, independently of the system size $n$, and the total arithmetic complexity scales as $\mathcal O \left (n / \epsilon^2 \right)$.
\end{restatable}

In particular, for uniform on-site dimension $d$, maximum cluster
size $k$, and a bounded number of inclusions per cluster, the per-iteration cost of the algorithm becomes
\[
\mathcal O \left( n (d^{k})^3\right) = \mathcal O \left( n d^{3k}\right).
\]

\subsubsection{Accelerated smoothing method}

In this section, we extend the accelerated smoothing method introduced above to higher-order relaxations. As many of the main ideas carry over directly, our treatment here would be brief and focus only on the required modifications.

In order to smoothen our objective, we replace the minimum eigenvalue of the local augmented Hamiltonians  in the dual function $Q_C \left (\nu \right)$ with its free energy approximation, resulting in the now-smoothed dual function \begin{align}
    Q_{C, \beta} (\nu)
    =
    \sum_{e\in E} F_\beta
    \left(
    \theta_e(\{\nu_{e\to f},
    \nu_{g\to e}\})
    \right).
\end{align}
\noindent
The following lemma bounds the approximation error introduced by $Q_{C, \beta} \left (\nu \right)$.

\begin{restatable}[Bound on smoothing error]{lemma}{smootherror_higherorder}
For any feasible $\nu$, 
    \begin{align}
        Q_{C, \beta} \left (\nu \right) \leq Q_C \left (\nu \right) \leq Q_{C, \beta} \left (\nu \right) + \frac{\Gamma_G}{\beta},
    \end{align} 
    \noindent
    where $\Gamma_G \coloneqq \sum_{e \in E} \log D_e$ and $D_e = \dim (\mathcal H_e)$ is the dimension of the local Hilbert space associated with hyperedge $e$.
\end{restatable}
\noindent
The proof technique follows that of \cref{lem:bound_on_smoothing_error_2_local}.

The $e \to f$ component of the gradient of $Q_{C, \beta} (\nu)$ is then \begin{equation}
    \nabla_{\nu_{e \to f}} Q_{C, \beta} \left (\nu \right) = {\rho_\beta}_f \left (\nu \right) - \Tr_{f^c} \left[ {\rho_\beta}_e \left (\nu \right)\right],
\end{equation}
\noindent 
where \begin{align}
    {\rho_\beta}_{e} \left(\nu \right) 
    = \dfrac{\exp \left (- \beta \theta_e(\{\nu_{e\to f}, \nu_{g \to e}\}) \right)}{\Tr \left[ \exp \left (- \beta  \theta_e(\{\nu_{e\to f}, \nu_{g \to e}\}) \right) \right]}
\end{align}
\noindent
is the Gibbs state of the local augmented Hamiltonian supported on hyperedge $e$.

In order to improve the quadratic error scaling of the naive subgradient method, we again adopt Nesterov's accelerated gradient descent to maximise $Q_{C, \beta} \left (\nu \right)$. As before, we split the error-density budget $\epsilon$ equally between the smoothing approximation error and the optimisation error. We choose $\beta = 2 \Gamma_G / \left (n \epsilon \right)$ and $\widehat L_\beta = \beta \Lambda_G/2$, where $\Lambda_G \coloneqq \max_{e \to f} \left\{ m_f + m_e D_e/D_f \right\}$ with $m_e \coloneqq | \mathcal A(e)| + |\mathcal D(e)|$, and run Nesterov's method for $\mathcal O \left(\sqrt{2 \widehat L_\beta / \left (n \epsilon \right ) } \right)$ many iterations.

The convergence result is summarised in the following theorem.

\begin{restatable}[Convergence guarantee for higher-order accelerated smoothing]{theorem}{convergencesmoothedgradienthigherordermarginals}\label{thm:convergence_smoothed_gradient_higher-order_marginals}
    Given a target error density $\epsilon$, the higher-order variant of \cref{alg:accelerated-smoothed-quantum-dual} 
    with inputs 
    $\beta = 2 \Gamma_G/ \left ( n \epsilon \right)$
    % $\beta = \dfrac{ 2 \Gamma_G}{ N \epsilon_{\mathrm{den}} }$ 
    and 
    $\widehat L_\beta = \beta \Lambda_G/2$ 
    % $\widehat L_\beta = \dfrac{\beta \Lambda_G}{2}$ 
    produces an objective value $f_{\mathrm{grad}} \coloneq Q_{C, \beta} \left (\nu^{(T)} \right)$ such that 
    $ \dfrac{f_C^* - f_{\mathrm{grad}}}{n} \leq \epsilon$ after $T$ iterations, for any
    $T \ge \left \lceil \dfrac{2 R_{C, \beta} \sqrt{\Gamma_G \Lambda_G }}{n \epsilon } \right \rceil - 1$, where 
    $R_{C, \beta} \coloneq \mathrm{dist} \left( \nu^{(0)}, \arg \max_\nu Q_{C, \beta} \left (\nu \right) \right)$, 
    $\Gamma_G \coloneqq \sum_{e \in E} \log \left (D_e \right)$ 
    with $D_e = \dim \left(\mathcal H_e \right)$ and 
    $\Lambda_G \coloneqq \max_{e \to f} \left \{ m_f + m_e D_e / D_f \right\}$ with 
    $m_e \coloneqq | \mathcal A(e)| + |\mathcal D(e)|$.
    Each iteration in the algorithm
    requires $\mathcal O \left( \sum_{e \in E } D_e^3 \right)$ arithmetic operations.
\end{restatable}
\noindent
Its proof technique follows similarly to that of \cref{thm:convergence_smoothed_gradient_complete}.

Consider a family of bounded-degree hypergraphs with uniform on-site Hilbert-space dimension $d$, bounded maximum cluster size $k$, and bounded local incidence. It follows that $\Gamma_G = \mathcal O(n)$, $\Lambda_G = \mathcal O(1)$, $|E| = \mathcal O (n)$, and $\sum_{e \in E} |\mathcal D(e)| = \mathcal O(n)$. Assume further that $R_{C, \beta} = \mathcal O(\sqrt{n})$. \cref{thm:convergence_smoothed_gradient_higher-order_marginals} then gives the following corollary.

\begin{restatable}{corollary}{corollarysubgradienccccct}
    For bounded-degree hypergraphs, bounded maximum cluster size and $\mathcal O (1)$ optimal message norms, the target error density $\epsilon$ can be achieved in $\mathcal O \left (1 / \epsilon \right)$ accelerated gradient iterations, independently of the system size $n$, and the total arithmetic complexity scales as $\mathcal O \left (n / \epsilon \right)$.
\end{restatable}

In particular, for uniform on-site dimension $d$, maximum cluster
size $k$, and a bounded number of inclusions per cluster, the per-iteration cost of the algorithm becomes
\[
\mathcal O \left( n (d^{k})^3\right) = \mathcal O \left( n d^{3k}\right).
\]

\subsection{Comparison between the subgradient and accelerated smoothing methods}

Our analysis on the convergence guarantees of the subgradient method and the accelerated smoothing method above shows a convincing advantage for the latter. For a target error density $\epsilon$, the iteration bound for the subgradient method scales as $\mathcal O(1/\epsilon^2)$, whereas the iteration bound for accelerated smoothing only scales as $\mathcal O(1/\epsilon)$. In the higher-order relaxation regime, there is theoretical evidence that accelerated smoothing suffers more when the size of the clusters is large, but the disadvantage is mild under most circumstances, and it can be easily compensated for by the improved $\epsilon$ dependence in any case. In terms of per-iteration cost, both algorithms have the same asymptotical scaling of $\mathcal O(n d^3)$, where $n$ is the system size and $d$ is the dimension of the matrices they are computing. However, in practical settings, there can be cases where the per-iteration cost for the subgradient method is cheaper, but the specific conditions under which these advantages occur remain unclear. We designate this question to future work, and refer our readers to \cref{section:numerics_comparisons} on the relevant numerical experiments we have done so far on the comparison between the two methods.

\section{Numerics} \label{section:numerics}

In this section we provide numerical evidence to some of the things mentioned throughout the text.

\subsection{Full vs.\ geometrically local relaxations}

In \cref{section:gse}, we mentioned that although at each level of the hierarchy, the ground state energy relaxation with marginal set $\mathbb{L}_t(G)$ can be solved in polynomial time, the large amount of constraints needed to implement the relaxation scale as $\mathcal{O}(n^t)$, which in practice can become computationally expensive for large $n$ and moderate $t$. We also mentioned that one instead typically considers optimising over a cluster family, among which we highlighted the hierarchy of consecutive interval clusters of size $t$, and denoted the set of marginals as $\mathbb{Y}_t$; see \cref{eqn:interval_set_defn}. There, we mentioned that although the corresponding optimisation program would result is a worse approximation to the ground state energy, the loss of fidelity compared to the speed gain is overall positive. Here, we present numerical evidence for this fact. To simplify notation for what follows, we equate the optimal value of the relaxation with the set used, e.g.\ $\mathbb{Y}_t$ is the optimal value of the ground state energy relaxation with this set of marginals.

In \cref{fig:relative_convergence}, we plot the energy difference between the approximation found with the 3-local consecutive clusters against the one with $\mathbb{L}_3$, and also $\mathbb{Y}_4$ and $\mathbb{Y}_5$ for a chain of varying sizes with random Hamiltonian terms. The image shows that while $\mathbb{L}_3$ does produce a better approximation than $\mathbb{Y}_3$, it is overcome by the relaxation of $\mathbb{Y}_4$, which from \cref{fig:convergence_time} we see that is also (up to $8 \times$) faster.

We also observed that for Hamiltonians with more symmetries, e.g.\ Heisenberg and Ising, the less difference there was between the relaxations with $\mathbb{L}_t$ and $\mathbb{Y}_t$.

\begin{figure}[t]
    \centering
    \begin{subfigure}{0.48\textwidth}
        \centering
        \includegraphics[width=1\linewidth]{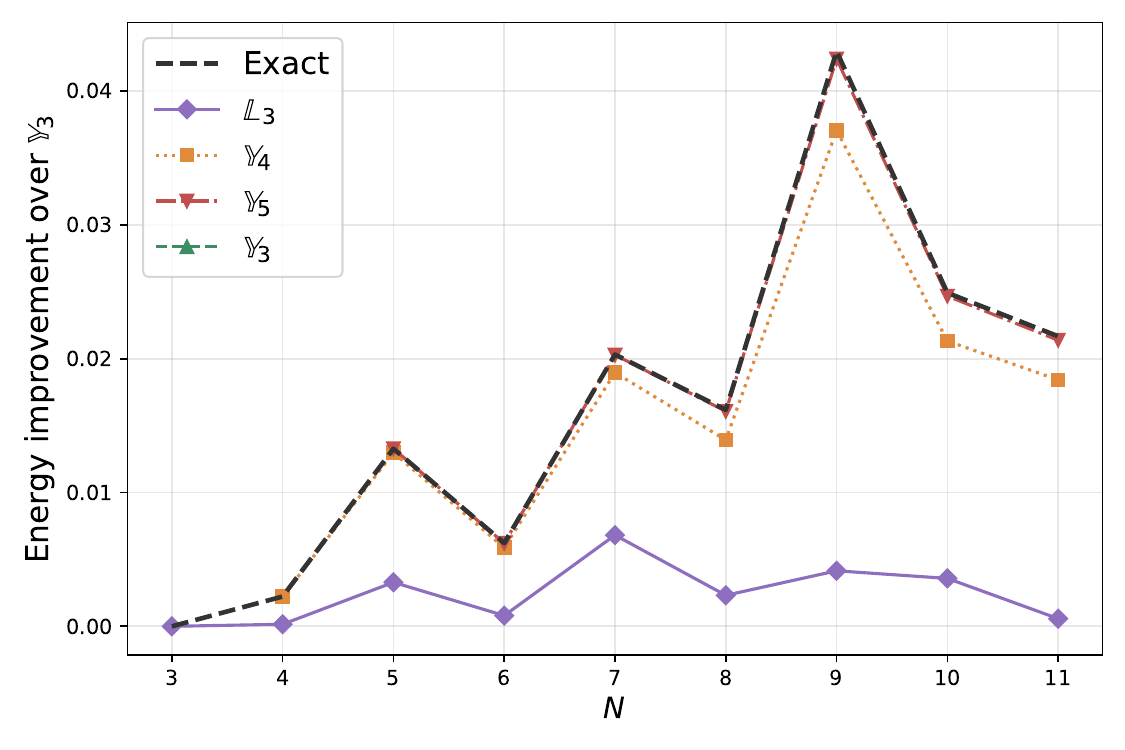}
        \caption{}
        \label{fig:relative_convergence}
    \end{subfigure}
    \begin{subfigure}{0.48\textwidth}
        \centering
        \includegraphics[width=1\linewidth]{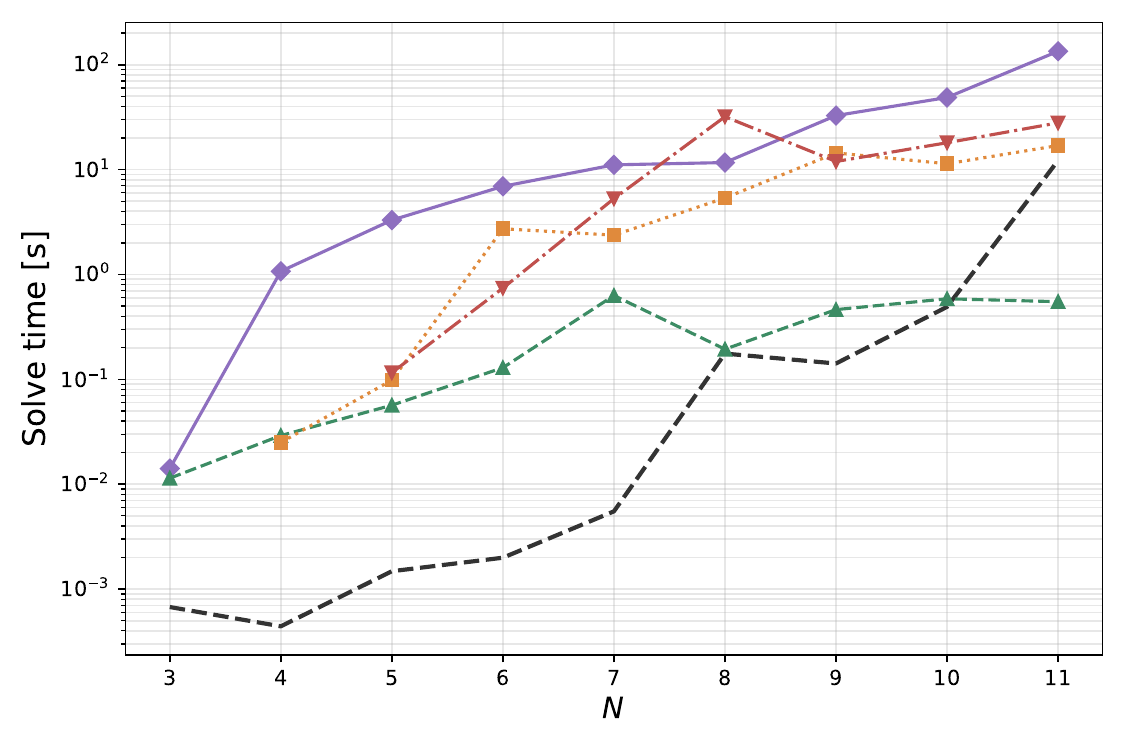}
        \caption{}
        \label{fig:convergence_time}
    \end{subfigure}
    \caption{Comparison of the relaxations $\mathbb{L}_3$, $\mathbb{Y}_3$, $\mathbb{Y}_4$ and $\mathbb{Y}_5$ for a chain of size $N$ with random interactions on each edge. (a) Compares the energy improvement of the different optimisations relative to $\mathbb{Y}_3$. (b) Wall-clock time it took CVXPY to solve these relaxations.}
    \label{fig:full_vs_geometric_hierarchies}
\end{figure}

\subsection{Size independence of the optimal Lagrange multiplier norm}

In \cref{section:qmp}, a relevant part of the of the convergence guarantees for the subgradient and smooth algorithms, depended on the optimal Lagrange multiplier norm. Here, we provide numerical evidence that this quantity is independent of size.

This is shown in \cref{fig:multiplier_independence}, which plots the average optimal multiplier Frobenius norm as a function of size for a 1D chain of qubits for several typical Hamiltonians. The image demonstrates that there is no visible correlation.

\begin{figure}[t]
    \center
    \includegraphics[width = 0.5\linewidth]{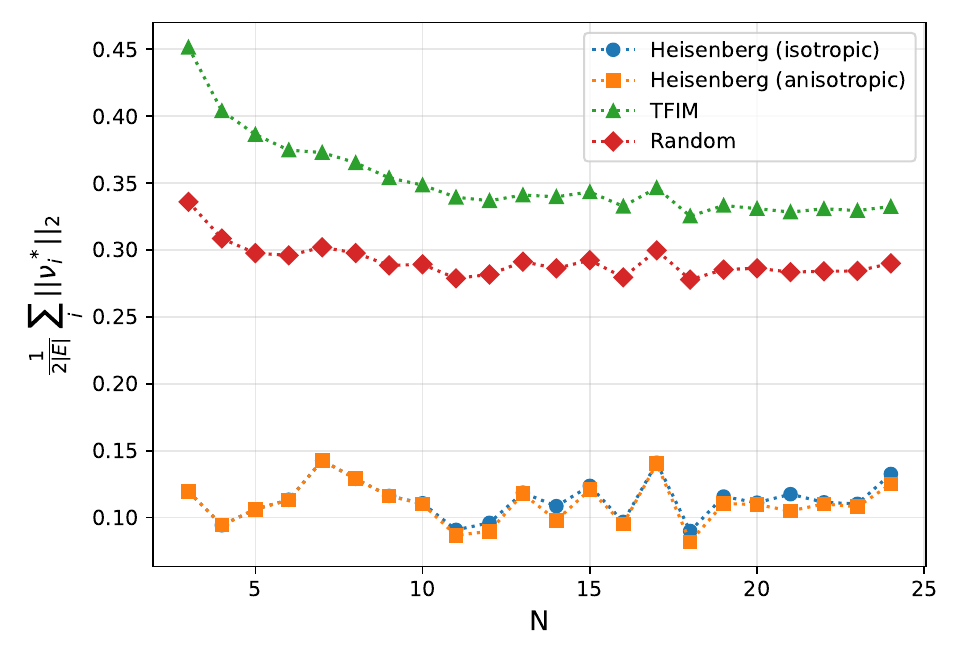}
    \caption{Average optimal Lagrange multiplier norm for a chain of various lengths $N$ for four Hamiltonians: isotropic Heisenberg $(J = 1)$, anisotropic Heisenberg  $(J_x,J_y, J_z) = (2,-3,0.5)$,  critical transverse field Ising model (TFIM) $(J =h =1)$, and a random Hamiltonian with independently sampled Hermitian terms on each edge.}
    \label{fig:multiplier_independence}
\end{figure}

\subsection{Comparison of the subgradient and smooth methods}
\label{section:numerics_comparisons}

In \cref{section:qmp}, we demonstrated that in the worst-case, the number of iterations needed to obtain $\epsilon$ convergence of the subgardient and Nesterov-type accelerated algorithms were, respectively, $\mathcal{O}(1/\epsilon^2)$ and $\mathcal{O}(1/\epsilon)$, in the case of bounded degree graphs and under the assumption that the optimal Lagrange multipliers were independent of system size. Furthermore, the cost per iteration was also similar. In theory then, it seems like the smooth algorithm prevails over the subgradient one. Here, we provide numerical evidence that this is also generally the case in practice, except for some adversarially selected Hamiltonians.

% First, we considered the iteration growth to achieve different convergence values of $\epsilon$

First, we considered the practical dependence of the number of iterations with the desired accuracy $\epsilon \in [0.02,0.002]$ for standard 1D hamiltonians, like isotropic and anisotropic Heisenberg models, the transverse field Ising model, as well as random Hamiltonians and free-fitted polynomials of $1/\epsilon$. 

In the subgradient method, the largest exponent observed was of $2$ for random Hamiltonians, while the exponents for the others ranged from $1.33$ to $1.8$. For the smooth method, we had similar observations, where the largest exponent of $0.8$ was achieved for random Hamiltonians with coefficients of $\approx 0.5$ for the other Hamiltonians. \cref{fig:eps_dependence} shows these curves for the upper and lower limits observed.

Our numerical experiments show that there are Hamiltonians that can meet the theoretical upper bounds, but this bound is also loose for naturally occurring Hamiltonians. Furthermore, they also confirm that the smooth method is also better in practice than the subgradient method, at least for standard Hamiltonians. It remains an open question whether there are specific cases where the subgradient algorithm outperforms the smooth one. We leave this as a direction for future work.

\begin{figure}[t]
    \center
    \begin{subfigure}{0.48\textwidth}
        \center
        \includegraphics[width=1\linewidth]{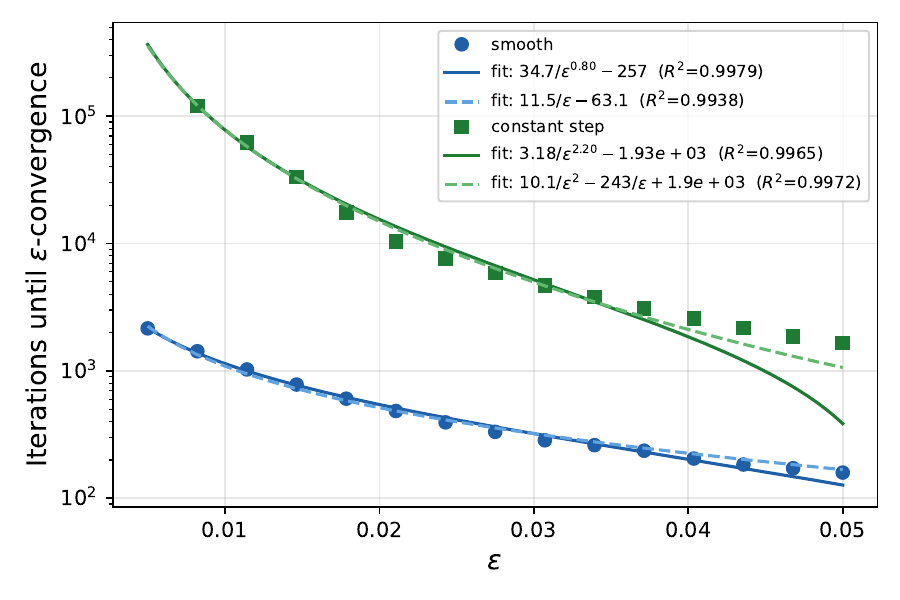}
        \caption{}
        \label{fig:eps_dependence_random}
    \end{subfigure}
    \hspace{-0.3cm}
    \begin{subfigure}{0.48\textwidth}
        \center
        \includegraphics[width=1\linewidth]{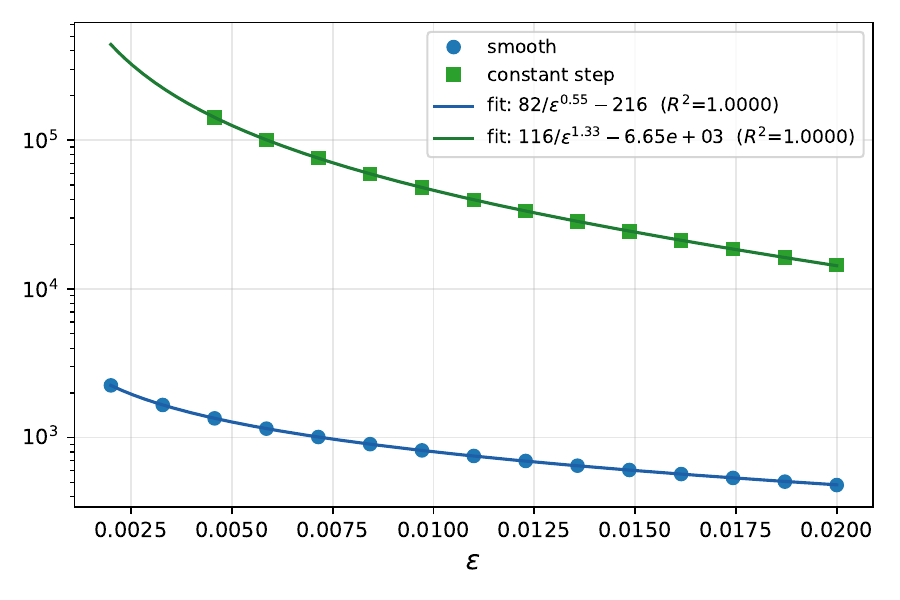}
        \caption{}
        \label{fig:eps_dependence_ising}
    \end{subfigure}
    \caption{Dependence of the number of iterations needed for $\epsilon$ convergence to the exact optimal value as a function of $\epsilon$ for (a) a random Hamiltonian on a chain, and (b) an isotropic Heisenberg Hamiltonian ($J = 1$) with an on site $X$ potential with strength $h = 0.5$. The solid lines represent the free-fits, and in (a) the dashed ones represent the constrained quadratic fits. We note that for the random Hamiltonian of the first panel, the free-fit indicates an exponent of $2.2$ that surpasses the theoretical bound. This is likely due to the noise associated with such Hamiltonians, or the need for more data points. The quadratic fit, however, gives a better $R^2$ value and hence confidence that the theoretical upper bound we have derived is not violated.}
    \label{fig:eps_dependence}
\end{figure}

\subsection{Benchmark against CVXPY}

We have yet to provide a meaningful performance comparison of our algorithms against CVXPY. We have found that CVXPY performs a number of under-the-hood optimisations---exploiting sparsity or solver-specific heuristics---on the problems we consider here, which currently allow it to outperform our current rustic Python implementation. Improving the performance of our algorithm is left to future work, which include allowing GPU parallelisation, just-in-time compilation, or perhaps even rewriting it in a low-level language instead. One particularly promising optimisation is to consider tensor network messages, which is discussed in more detail in the following section.

\section{Discussion and future work}

In this work, we have derived novel convergence guarantees for locally consistent marginal relaxations. We have also introduced highly parallelisable message passing algorithms to solve those relaxations, and derived their algorithmic complexity.

This work opens several future research directions. First, it is still an open question whether one can provide exactness of higher levels of the hierarchy for some specific classes of Hamiltonians, beyond commuting Hamiltonians on a tree.

Second, it would be interesting to extend the proof of exponential convergence of the interval hierarchy for gapped weakly-interacting chains to strongly interacting systems and higher dimensions. A promising direction is to use the higher dimensional generalisation of the block diagonalisation used in the proof of theorem \ref{thm:exp_conv_1d_weak} derived in \cite{Del_Vecchio_2022}.

A problem that occurs in the relaxations we have discussed, is that the size of the marginals grows exponentially with the level of the hierarchy, which in the algorithms, is reflected in the exponentially growing cost runtime and memory. A promising avenue for improvement is to consider the dual program of the optimisation program considered in \cite{kull2024lower} (specifically Section II) which parameterises the locally consistent reduce density matrices as matrix product states (if in 1D) or tree tensor networks (for systems in higher dimensional geometries). The advantage being that these can be described with a polynomial number of parameters. Hence, we can conceive a scenario where solving the dual of the their resulting optimization program and applying similar methods to what we have done here, it may result in a message passing algorithm where the messages are themselves tensor network states. 

Finally, we reiterate the need to further optimise the runtime of our algorithms for successful benchmark comparisons against CVXPY. The most promising route being a GPU implementation to exploit the inherently local, and hence distributed structure of our algorithms. \label{section:outro}

\section*{Acknowledgements}

We thank Daniel Nagaj, Andrej Gendiar, and Jordi Tura for helpful discussions. RRC was supported by the Slovak Research and Development Agency through the project APVV-22-0570 (DeQHOST) and project VEGA 2/0164/25 (Quantum Structures).
SKL and RB acknowledge support from the EPSRC Grant number EP/W032643/1. This work was supported in part by a grant of access to OpenAI models through the ChatGPT for Academic Researchers program. We acknowledge the use of Claude and ChatGPT to provide sanity checks, improve grammar and presentation, aid in the literature review, and help with figures and the proof of Theorem \ref{thm:exp_conv_1d_weak}. In the Python code, they were used to debug, improve structure and speed, and generate some tests. All the text has been written by the authors.

\newpage
\printbibliography

\newpage
\appendix

\section{Proofs of theorems} \label[appendix]{appendix:thmproofs}

\subsection{Proof of \cref{thm:lpi}}
An important part of the proof of \cref{thm:lpi} has to deal with the decomposition of local Hilbert spaces of quantum Markov states on a tree. First, let us state a lemma used in that proof about the decomposition of the space for star graphs.

\begin{lemma}[Star graph decomposition] \label{lemma:stardecomposition}
    For a quantum Markov state on a star graph where the centre node is $i$ and neighbours $s$ leafs labelled $j_1, \ldots, j_s $, the Hilbert space of site $i$ decomposes as
    \begin{equation} \label{eqn:BVstar}
        \mathcal{H}_i = \bigoplus_{\alpha_i} \bigotimes_{k=1}^{s} \mathcal{H}_{i \to j_k}^{\alpha_i}.
    \end{equation}
    \noindent
    Here, $\mathcal{H}_{i \to j_k}^{\alpha_i}$ can be thought of as the space in sector $\alpha_i$ (a multi-index with $s-1$ components) associated with edge $(i,j_k) \in E$. Furthermore, the quantum Markov state on this graph takes the form

    \begin{equation} \label{eqn:QMN_state_decomp}
        \mu_{i \cup J} = \bigoplus_{\alpha_i} p_{\alpha_i} \bigotimes_{k=1}^s \mu_{i \to j_k, j_k}^{\alpha_i}
    \end{equation}
    \noindent
    where $p_{\alpha_i}$ is a probability distribution, i.e.\ $\sum_{\alpha_i} p_{\alpha_i} = 1$.
\end{lemma}

\begin{proof}
    First, observe that for a quantum Markov network in such a graph, the conditional mutual information between any two disjoint subsets of the leaf nodes conditioned on site $i$ always vanishes. The proof idea is to recover \cref{eqn:BVstar} by repeatedly considering the space decompositions of \cite{QMN} for successively smaller subspaces. 

    We begin by considering the subsets $L = \{j_1\}$ and $R = \{j_2, \ldots, j_s\}$. As mentioned, $I(L:R \,|\, i) = 0$, and the \cite{QMN} decomposition implies that the space at site $i$ can be decomposed as

    \begin{equation}
        \mathcal{H}_i = \bigoplus_{\alpha \in S} \mathcal{H}_{i \to j_1}^{\alpha} \otimes \mathcal{H}_{i \to R}^{\alpha},
    \end{equation}
    \noindent
    where we have explicitly denoted $S$ to be the set of labels that $\alpha$ iterates over. We have successfully recovered the space $\mathcal{H}_{i \to j_1}$, yet the other space is still associated with several other sites, i.e.\ $j_2, \ldots,j_s$. We subsequently refine this latter subspace by again separating the nodes into two groups: $T = \{j_2\}$ (top) and $B = \{j_3, \ldots, j_s\}$ (bottom). Still, $I(T:B\,|\,i) = 0$, and taking into consideration the existing decomposition, the space at each sector $\alpha$ then decomposes further, resulting in

    \begin{equation}
        \mathcal{H}_{i \to R}^{\alpha} = \bigoplus_{\beta_{\alpha} \in S_\alpha} \mathcal{H}_{i \to j_2}^{\beta_{\alpha}} \otimes \mathcal{H}_{i \to B}^{\beta_\alpha},
    \end{equation}
    \noindent
    where $\beta_\alpha$ is an index iterating over the set $S_{\alpha}$.

    By repeatedly refining the larger subspace by peeling one node at a time until all decompositions apply to a single node (which requires $s-1$ steps), the decomposition for node $i$ can be expressed as 

    \begin{equation} \label{eqn:CMIstarresult}
        \mathcal{H}_i = \bigoplus_{t \, \in \, T}
        \Bigl( \bigotimes_{k=1}^{s-1} \mathcal{H}_{i \to j_k}^{\,t|_k} \Bigr)
        \otimes \mathcal{H}_{i \to j_s}^{\,t|_{s-1}},
    \end{equation}
    \noindent
    where we have adapted the notation to better suit the general case and express the progressively nested labels. These are expressed via the tree $T = T_{s-1}$ defined recursively for all $1 < k \leq s-1$ as

    \begin{equation}
        T_k = \bigl\{\, (q, \alpha) : q \in T_{k-1},\; \alpha \in S_k(q) \,\bigr\},
    \end{equation}
    \noindent
    where $(q,\alpha)$ appends $\alpha$ to the sequence $q$ and $T_1 = S_1$.\footnote{Technically, we have describe a rooted forest, where the roots are the values of $S_1$, however, we can add an artificial and inconsequential root node to see them all as part of one larger tree.} Thus, $T_k$ is the set of nodes at depth $k$: each is obtained from a depth-$(k-1)$ node $q$ by choosing a child $\alpha \in S_k(q)$, and $S_k(q)$ is precisely the index set produced by peeling the $k$-th node within sector $q$. For some path $t \in T$, $t|_k \in T_k$ is the truncation restricted to the first $k$ levels of the tree.

    This decomposition differs from \cref{eqn:BVstar} only in that its sector labels are nested rather than shared, however, they are in fact equivalent by redefining

    \begin{equation}
        \widehat{\mathcal{H}}_{i \to j_k}^{\,t} := \mathcal{H}_{i \to j_k}^{\,t|_k}
        \quad (1 \le k \le s-1), \qquad
        \widehat{\mathcal{H}}_{i \to j_s}^{\,t} := \mathcal{H}_{i \to j_s}^{\,t|_{s-1}},
    \end{equation}
    \noindent
    so that finally:

    \begin{equation}
        \mathcal{H}_i = \bigoplus_{t \in T} \bigotimes_{k=1}^{s}
        \widehat{\mathcal{H}}_{i \to j_k}^{\,t},
    \end{equation}
    \noindent
    which is of the form mentioned in \cref{eqn:BVstar}.

    As for the quantum Markov state, after the first decomposition, the state can be written as

    \begin{equation}
        \mu_{i \cup J} = \bigoplus_{\alpha \in S} p_{\alpha} \mu_{i\to j_1, j_1}^\alpha \otimes \mu_{i \to R, R}^\alpha
    \end{equation}
    \noindent
    where $p_\alpha$ is a probability distribution, and
    \begin{equation}
        \mu_{i \cup J} = \bigoplus_{\alpha \in S} \bigoplus_{\beta_{\alpha} \in S_{\alpha}} p_{\alpha} p_{\beta_{\alpha}} \mu_{i \to j_1, j_1}^{\alpha} \otimes \mu_{i \to j_2, j_2}^{\beta_\alpha} \otimes \mu_{i \to B, B}^{\beta_\alpha}
    \end{equation}
    after the second decomposition. Moreover, here, for any fixed $\alpha \in S$, $\sum_{\beta_\alpha \in S_{\alpha}} p_{\beta_{\alpha}} = 1$. Hence, after all decompositions, the state is 

    \begin{equation}
        \mu_{i \cup J} = \bigoplus_{t \in T} \Big( \prod_{\ell =1}^s p_{t|_\ell} \Big) \Big(\bigotimes_{k = 1}^{s-1} \mu_{i \to j_k}^{t|_k} \Big) \otimes \mu_{i \to j_s}^{t|_{s-1}}.
    \end{equation}

    Defining $\hat{p}_{t} := \prod_{\ell =1}^s p_{t|_\ell}$, $\hat{\mu}_{i \to j_k, j_k}^{t} := \mu_{i \to j_k, j_k}^{t|_k}$ for all $1 \leq k \leq s-1$, and $\hat{\mu}_{i \to j_s, j_s}^t := \mu_{i \to j_s, j_s}^{t|_{s-1}}$ we obtain a state of the form as in \cref{eqn:QMN_state_decomp}. Finally, observe that since $\hat{p}_t$ is a probability tree, $\sum_t \hat{p}_t = 1$ and is hence a probability distribution.
\end{proof}
    
\LequalsT*

\begin{proof}
    First, let us show that $\mqmn (T) \subseteq \lpi(T)$ and subsequently that $\lpi(T) \subseteq \mqmn (T)$. 

    \paragraph{($\mqmn (T) \subseteq \lpi(T)$):}
    This requires demonstrating that the marginals of quantum Markov networks on a tree are of the form shown in $\lpi$.
    
    Observe that for any site $i$ in the bulk of the graph, by definition of a quantum Markov network, we have that $I(j_1:j_2,\ldots,j_s | i) = 0$ where the $j_k$ are the different neighbours of site $i$. These nodes form a star subgraph, and the reduced density matrix supported on those nodes is clearly also a quantum Markov network. We can thus apply \cref{lemma:stardecomposition}, to obtain the reduced state on the star, expressed as a combination of two-site reduced density matrices. Now, focus on one particular neighbouring node, $j_x$ for $x \in \{1,\ldots, s\}$. If $j_x$ is another node in the bulk, then it admits a similar decomposition, resulting in two descriptions of the state sitting on the edge $(i,j_x)$:

    \begin{equation} \label{eqn:marginals_from_stars}
        \mu_{i,j_x} = \bigoplus_{\alpha_i} p_{\alpha_i} \mu_{\mathrm{comp}(i \to j_x)}^{\alpha_i} \otimes \mu_{i \to j_x, j_x}^{\alpha_i} \hspace{0.3cm} \text{and} \hspace{0.3cm} \mu'_{i,j_x} = \bigoplus_{\alpha_x} p_{\alpha_x} \mu_{i, j_x \to i}^{\alpha_x} \otimes \mu_{\mathrm{comp}(j_x \to i)}^{\alpha_x}
    \end{equation}
    \noindent
    where $\mu_{\mathrm{comp}(i \to j_x)}^{\alpha_i} = \bigotimes_{k \neq x} \mu_{i \to j_s}^{\alpha_i}$. Because the global state has consistent marginals, these must then be the same. To find a description of the state consistent with both, let $P_{\alpha_i}$ and $Q_{\alpha_x}$ be projectors unto the $\alpha_i$ and $\alpha_x$ sectors of site $i$ and $j_x$, respectively. As the projectors act on different sites, they commute, and the two decompositions can be simultaneously refined into sectors labelled by both $\alpha_i$ and $\alpha_x$. Hence, $\mu_{i,j_x} = \mu_{i,j_x}' = \bigoplus_{\alpha_i, \alpha_x} X_{\alpha_i, \alpha_x}$ where

    \begin{equation}
        X_{\alpha_i, \alpha_x} := (P_{\alpha_i} \otimes Q_{\alpha_x}) \mu_{i,j_x} (P_{\alpha_i} \otimes Q_{\alpha_x})
    \end{equation}
    is the matrix block in sector $(\alpha_i,\alpha_x)$.  Applying the projectors for fixed $\alpha_i$ and $\alpha_x$ to $\mu_{i,j_x}$, we see that

    \begin{equation}
        X_{\alpha_i, \alpha_x} = p_{\alpha_i} \mu_{\mathrm{comp}(i \to j_x)}^{\alpha_i} \otimes [(I \otimes Q_{\alpha_x})\mu_{i \to j_x, j_x}^{\alpha_i} (I \otimes Q_{\alpha_x}) ] =  \mu_{\mathrm{comp}(i \to j_x)}^{\alpha_i} \otimes W_{\alpha_i, \alpha_x} \otimes  \mu_{\mathrm{comp}(j_x \to i)}^{\alpha_x}
    \end{equation}
    \noindent
    where the first equality follows straightforwardly, and for the second, we noted that $(I \otimes Q_{\alpha_x})\mu_{i \to j_x, j_x}^{\alpha_i} (I \otimes Q_{\alpha_x})$ is an operator on the space $\mathcal{H}_{i \to j_x}^{\alpha_i} \otimes \mathcal{H}_{j_x \to i}^{\alpha_x} \otimes \mathcal{H}_{\mathrm{comp}(j_x \to i)}^{\alpha_x}$, and made use of the fact from $\mu_{i,j_x}'$, that the state over site $j_x$ also decomposes. Here, $W_{\alpha_i, \alpha_x} := p_{\alpha_i} \Tr_{\mathrm{comp}(j_x \to i)}[(I \otimes Q_{\alpha_x})\mu_{i \to j_x, j_x}^{\alpha_i} (I \otimes Q_{\alpha_x})]$ is an operator supported only on $\mathcal{H}_{i \to j_x}^{\alpha_i} \otimes \mathcal{H}_{j_x \to i}^{\alpha_x}$.

    The resulting state $\mu_{i, j_x} = \bigoplus_{\alpha_i, \alpha_j} X_{\alpha_i,\alpha_x}$ is not currently normalized,  and so we define $t_{\alpha_i, \alpha_x}:= \Tr(X_{\alpha_i, \alpha_j}) = \Tr(W_{\alpha_i, \alpha_x}) = p_{\alpha_i} \Tr[(I \otimes Q_{\alpha_x})\mu_{i \to j_x, j_x}^{\alpha_i}]$ to be the normalization factor. Letting $\mu_{i\to j_x, j_x \to i}^{\alpha_i, \alpha_x} := W_{\alpha_i, \alpha_x}/t_{\alpha_i,\alpha_x}$, we finally obtain the normalized state 
    
    \begin{equation}
        \mu_{i, j_x} = \bigoplus_{\alpha_i, \alpha_x} t_{\alpha_i,\alpha_x} \mu_{\mathrm{comp}(i \to j_x)}^{\alpha_i} \otimes \mu_{i\to j_x, j_x \to i}^{\alpha_i, \alpha_x} \otimes  \mu_{\mathrm{comp}(j_x \to i)}^{\alpha_x}
    \end{equation}
    \noindent
    as promised. Note that $\sum_{\alpha_i} t_{\alpha_i,\alpha_x} = p_{\alpha_x}$ and similarly $\sum_{\alpha_x} t_{\alpha_i,\alpha_x} = p_{\alpha_i}$.

    \paragraph{($\lpi (T) \subseteq \mqmn(T)$):}

    The proof of this statement simply consists in showing that the Petz map can be applied iteratively to the marginals of $\lpi(T)$ to construct a global state of the chain.

    Starting from a marginal
    
    \begin{equation}
    \mu_{1, 2} = \bigoplus_{\alpha_2} p_{\alpha_2} \mu_{1, 2 \to 1}^{\alpha_2} \otimes \bigotimes_{\substack{j \in \mathcal{N}(2) \\ j \neq 1}} \mu_{2 \to j}^{\alpha_2}
    \end{equation}
    \noindent
    supported on an edge with a leaf node (site $1$), and adjacent edge marginal

    \begin{equation}
        \begin{aligned}
            \mu_{2,3} &= \bigoplus_{\alpha_2, \alpha_3} t_{\alpha_2, \alpha_3} \mu_{\mathrm{comp}(2 \to 3)}^{\alpha_2} \otimes \mu_{2\to 3, 3 \to 2}^{\alpha_2, \alpha_3} \otimes \mu_{\mathrm{comp}(3 \to 2)}^{\alpha_3} \\
            &= \bigoplus_{\alpha_2, \alpha_3} t_{\alpha_2,\alpha_3} \Big( \mu_{2 \to 1}^{\alpha_2} \otimes \bigotimes_{\substack{j \in \mathcal{N}(2) \\ j \neq 1}} \mu_{2 \to j}^{\alpha_2} \Big) \otimes \mu_{2\to 3, 3 \to 2}^{\alpha_2, \alpha_3} \otimes \mu_{\mathrm{comp}(3 \to 2)}^{\alpha_3},
        \end{aligned}
    \end{equation}
    \noindent
    we construct $\mu_{1,2,3}$ by using the map $\mathcal{P}_{\mu_{2,3}}(\mu_{1,2}) = \mu_{2,3}^{\frac{1}{2}} (\mu_2^{-\frac{1}{2}} \mu_{1,2} \mu_2^{-\frac{1}{2}} \otimes I_3)\mu_{2,3}^{\frac{1}{2}}$. Here, $\mu_2 =  \bigoplus_{\alpha_2} p_{\alpha_2} \mu_{2 \to 1}^{\alpha_2} \otimes \bigotimes_{j \in \mathcal{N}(2)} \mu_{2 \to j}^{\alpha_2}$ is, by definition, consistent with $\mu_{1,2}$ and $\mu_{2,3}$. In a slight abuse of notation, we denoted $\mu_{2 \to 1}^{\alpha_2} := \Tr_1(\mu_{1, 2 \to 1}^{\alpha_2})$. Given the form of the marginals, the reconstruction is exact and $\mu_{1,2,3}$ will have $\mu_{1,2}$ and $\mu_{2,3}$ as its marginals. For the actual application of the map, note that since the marginals are block diagonal, we can apply the map in each of the sectors (and consider the inverse only on the support of the operators).

    For a choice of sector $\alpha_2 = k$ and $\alpha_3 = \ell$, for the middle term, we have 

    \begin{equation}
        \begin{aligned}
        \mu_2^{-\frac{1}{2}} \mu_{1,2} \mu_2^{-\frac{1}{2}}
            &= \Big[p_{k}^{-\frac{1}{2}} (\mu^k_{2\to 1})^{-\frac{1}{2}} \otimes \bigotimes_{\substack{j \in \mathcal{N}(2) \\ j \neq 1}} (\mu_{2 \to j}^k)^{-\frac{1}{2}} \Big] p_{k} \mu_{1, 2 \to 1}^{k} \otimes \bigotimes_{\substack{j \in \mathcal{N}(2) \\ j \neq 1}} \mu_{2 \to j}^{k}     \Big[p_{k}^{-\frac{1}{2}} (\mu^k_{2\to 1})^{-\frac{1}{2}} \otimes \bigotimes_{\substack{j \in \mathcal{N}(2) \\ j \neq 1}} (\mu_{2 \to j}^k)^{-\frac{1}{2}} \Big] \\
            &= (\mu^k_{2\to 1})^{-\frac{1}{2}} \mu_{1, 2 \to 1}^{k} (\mu^k_{2\to 1})^{-\frac{1}{2}} \otimes  \bigotimes_{\substack{j \in \mathcal{N}(2) \\ j \neq 1}} (\mu_{2 \to j}^k)^{-\frac{1}{2}} \mu_{2 \to j}^k (\mu_{2 \to j}^k)^{-\frac{1}{2}}.
        \end{aligned}
    \end{equation}
    \noindent
    Then, conjugation by $\mu_{2,3}^{\frac{1}{2}}$, yields

    \begin{equation}
        \begin{aligned}
            &= \Big[ t_{k,\ell} \big(\mu_{2 \to 1}^{k} \otimes \bigotimes_{\substack{j \in \mathcal{N}(2) \\ j \neq 1}} \mu_{2 \to j}^{k} \big) \otimes \mu_{2\to 3, 3 \to 2}^{k, \ell} \otimes \mu_{\mathrm{comp}(3 \to 2)}^{\ell} \Big]^{\frac{1}{2}}  \Big[ (\mu^k_{2\to 1})^{-\frac{1}{2}} \mu_{1, 2 \to 1}^{k} (\mu^k_{2\to 1})^{-\frac{1}{2}} \otimes  \bigotimes_{\substack{j \in \mathcal{N}(2) \\ j \neq 1}} (\mu_{2 \to j}^k)^{-\frac{1}{2}} \mu_{2 \to j}^k (\mu_{2 \to j}^k)^{-\frac{1}{2}}\Big] \\
            & \Big[ t_{k,\ell} \big(\mu_{2 \to 1}^{k} \otimes \bigotimes_{\substack{j \in \mathcal{N}(2) \\ j \neq 1}} \mu_{2 \to j}^{k} \big) \otimes \mu_{2\to 3, 3 \to 2}^{k, \ell} \otimes \mu_{\mathrm{comp}(3 \to 2)}^{\ell} \Big]^{\frac{1}{2}} \\
            &= t_{k,\ell} \Big(\mu_{1, 2 \to 1}^k \otimes \bigotimes_{\substack{j \in \mathcal{N}(2) \\ j\neq 1,3}} \mu_{2 \to j}^k \Big) \otimes  (\mu_{2\to 3, 3 \to 2}^{k, \ell})^{\frac{1}{2}} (\mu_{2 \to 3}^k)^{-\frac{1}{2}}\mu_{2 \to 3}^k(\mu_{2 \to 3}^k)^{-\frac{1}{2}} (\mu_{2\to 3, 3 \to 2}^{k, \ell})^{\frac{1}{2}} \otimes \mu_{\mathrm{comp}(3 \to 2)}^\ell \\
            &= t_{k,\ell} \Big(\mu_{1, 2 \to 1}^k \otimes \eta_{2 \to 3, 3 \to 2}^{k, \ell} \Big) \otimes \Big(\bigotimes_{\substack{j \in \mathcal{N}(2) \\ j\neq 1,3}} \mu_{2 \to j}^k  \Big) \otimes \Big( \bigotimes_{ \substack{m \in \mathcal{N}(3) \\ m \neq 2}} \mu_{3 \to m}^\ell \Big)\, ,
        \end{aligned}
    \end{equation}
    \noindent
    where we have let $\eta_{2\to3, 3 \to 2}^{k,\ell} := (\mu_{2\to 3, 3 \to 2}^{k, \ell})^{\frac{1}{2}} (\mu_{2 \to 3}^k)^{-\frac{1}{2}}\mu_{2 \to 3}^k(\mu_{2 \to 3}^k)^{-\frac{1}{2}} (\mu_{2\to 3, 3 \to 2}^{k, \ell})^{\frac{1}{2}}$, and rearranged the terms in the last line so it is more apparent that the first terms are fixed and the rest can still admit expansions.

    For general $\alpha_2$ and $\alpha_3$, then

    \begin{equation}
        \mu_{1,2,3} = \mathcal{P}_{\mu_{2,3}}(\mu_{1,2}) = \bigoplus_{\alpha_2, \alpha_3} t_{\alpha_2,\alpha_3} \Big(\mu_{1, 2 \to 1}^{\alpha_2} \otimes \eta_{2 \to 3, 3 \to 2}^{\alpha_2, \alpha_3} \Big) \otimes \Big(\bigotimes_{\substack{j \in \mathcal{N}(2) \\ j\neq 1,3}} \mu_{2 \to j}^{\alpha_2}  \Big) \otimes \Big( \bigotimes_{ \substack{m \in \mathcal{N}(3) \\ m \neq 2}} \mu_{3 \to m}^{\alpha_3} \Big).
    \end{equation}

    Next, consider adding another edge adjacent to $2$, e.g.\ $(2,4) \in E$. This yields

    \begin{equation}
        \mu_{1234} = \bigoplus_{\alpha_2, \alpha_3, \alpha_4} t_{\alpha_2,\alpha_3,\alpha_4} \Big(\mu_{1, 2 \to 1}^{\alpha_2} \otimes \eta_{2 \to 3, 3 \to 2}^{\alpha_2, \alpha_3} \otimes \eta_{2\to4,4\to2}^{\alpha_3, \alpha_4} \Big) \otimes \Big(\bigotimes_{\substack{j \in \mathcal{N}(2) \\ j\neq 1,3,4}} \mu_{2 \to j}^{\alpha_2}  \Big) \otimes \Big( \bigotimes_{ \substack{m \in \mathcal{N}(3) \\ m \neq 2}} \mu_{3 \to m}^{\alpha_3} \Big) \otimes \Big( \bigotimes_{ \substack{z \in \mathcal{N}(4) \\ m \neq 2}} \mu_{4 \to z}^{\alpha_4} \Big)
    \end{equation}
    \noindent
    for some joint probability $t_{\alpha_2,\alpha_3,\alpha_4}$ yet to be defined. From these examples, it is evident that the quantum part (the sector marginals) can be glued together on trees and in fact, in any geometry. However, it is the mixing of the classical sector labels that generally creates a problem. The difficulty is hence classical, and the junction-tree theorem \cite{JW} (the one responsible for showing equality between $\mathbb{M}_c(T) = \mathbb{L}_c(T)$) states that on trees, it is possible to create such joint probabilities. From this theorem, we see that for any $\alpha = (\alpha_1, \ldots, \alpha_m)$ with $m \in [n]$, 

    \begin{equation}    
    t_\alpha = \prod_{s \in V} p_{\alpha_s} \prod_{(s,t) \in E} \frac{t_{\alpha_s, \alpha_t}}{p_{\alpha_s} p_{\alpha_t}} .
    \end{equation}
\end{proof}

\subsection{Proof of \cref{thm:commuting}}

\commutingtree*

\begin{proof}
    To begin, consider a node in the bulk, which we denote $i$, with $s$ neighbours labelled  $j_1, \ldots, j_s$. The Hamiltonian terms in its neighbourhood are  $\{h_{i,j_k}\}_{k = 1}^s$ and trivially satisfy $[h_{i,j_x}, h_{i,j_y}] = 0$ for all $x,y \in [s]$. The Structure lemma for commuting hamiltonians from \cite{BV} then states that the local Hilbert space of particle $i$ decomposes as\footnote{In \cite{BV}, there is an extra Hilbert space $\mathcal{H}_{i \to i}^{\alpha_i}$ in the decomposition. However, this term can be safely neglected as its existence is due to a further refinement of the decomposition we discuss here; see  \cite{hamiltoniancomplexity}.} 

    \begin{equation}
        \mathcal{H}_i = \bigoplus_{\alpha_i} \bigotimes_{k = 1}^s \mathcal{H}_{i \to j_k}^{\alpha_i}.
    \end{equation}
    \noindent
    such that for a given $\alpha_i$, each $h_{i,j_k}$ is only supported on the space $\mathcal{H}_{i \to j_k}^{\alpha_i} \otimes \mathcal{H}_{j_k}$. 

    Now, assume $j_x$ with $x \in [s]$ is also a bulk node. Then, there are two seemingly distinct decompositions of the Hamiltonian term $h_{i,j_x}$, one supported on the space $\mathcal{H}_{i \to j_k}^{\alpha_i} \otimes \mathcal{H}_{j_k}$, and the other on the space $\mathcal{H}_{i} \otimes \mathcal{H}_{j_k \to i}^{\alpha_x}$. It is then clear that local term acts non-trivially only on the spaces $\mathcal{H}_{i \to j_x}^{\alpha_i} \otimes \mathcal{H}_{j_x \to i}^{\alpha_x}$. Hence, $h_{i,j_x}$ decomposes as

    \begin{equation}
        h_{i,j_x} = \bigoplus_{\alpha_i, \alpha_x} I_{\mathrm{comp}(i \to j_x)}^{\alpha_i} \otimes \tilde{h}^{\alpha_i, \alpha_x}_{i \to j_x,j_x \to i} \otimes I_{\mathrm{comp}(j \to i)}^{\alpha_x},
    \end{equation}
    where $\tilde{h}^{(\alpha_i, \alpha_j)}_{i \to j,j \to i}$ is an effective Hamiltonian supported only on $\mathcal{H}_{i \to j}^{\alpha_i} \otimes \mathcal{H}_{j \to i}^{\alpha_j}$, and $I_{\mathrm{comp}(i \to j_x)}  =  \bigotimes_{\ell \in\mathcal{N}(i)} I_{i \to \ell}$ (and $\ell \neq j_x$). 

    At an edge containing a leaf node $\ell$ and a bulk node $k$, the hamiltonian term $h_{\ell,k}$ can be expressed as

    \begin{equation}
        h_{\ell,k} = \bigoplus_{\alpha_k} h_{i, k \to i}^{\alpha_k} \otimes I^{\alpha_k}_{\mathrm{comp}(k \to i)}.
    \end{equation}

    The non-trivial interactions then take place in restricted subspaces of the composite Hilbert space, and so it is safe to restrict the search for the minimising reduced density matrices to be density matrices with these structures. These are the marginals of the set $\lpi$. Hence,

    \begin{equation}
        f_1^* = \min_{\mu \in \mathbb{L}(T)} \sum_{(i,j) \in E} \langle \mu_{i,j}, h_{i,j} \rangle = \min_{\mu \in \lpi(T)} \sum_{(i,j) \in E} \langle \mu_{i,j}, h_{i,j} \rangle .
    \end{equation}
    \noindent
    The first program is a lower bound to the exact ground state energy $E_0$, and by \cref{thm:lpi} (in particular the direction $\lpi(T) \subseteq \mqmn(T)$), the second program is an upper bound. Thus, $f_1^* = E_0$, and additionally, the optimal density matrices $\{\mu^*_{i,j}\}$ are marginals of a global quantum Markov state.
\end{proof}

\subsection{Proof of Theorem \ref{thm:exp_conv_1d_weak}}
\label{app:proof exp_conv_1d_weak}

We recall the theorem we want to prove, defining the notation used throughout the proof.
Consider the following Hamiltonian on a chain of 
$n$ qudits:
\begin{align}
    \label{eq:H0tV}
    H = H_0 + t V\,,
    \quad 
    H_0 = \sum_{i=1}^n h_i
    \,,\quad 
    V = \sum_{i=1}^{n-1}v_{i,i+1}
\end{align}
such that, denoted the unique ground state of $H_0$ by $\ket{\Omega}=\bigotimes_{i=1}^n \ket{\Omega}_i$ with projectors onto the ground space and its orthogonal complement
\begin{align}
    p_i = (\ket{\Omega}\bra{\Omega})_i
    \,,\quad q_i = \id - p_i
\end{align}
we have
\begin{align}
    h_i\ket{\Omega}_i = 0 
    \,,\quad h_i = p_i h_i p_i + q_ih_iq_i = q_ih_iq_i\succeq q_i
    \,.
\end{align}
We also assume
\begin{align}
    \|h_i\|\le h_* \,,\quad \|v_{i,i+1}\|\le 1\,,
\end{align}
and take $t>0$.

For $2\le \ell\le n$, 
we consider the level-$\ell$ hierarchy consisting of density matrices $\rho_I$ supported on intervals $I$ of size $|I|\le \ell$, satisfying the local consistency conditions
\begin{align}
    \mathbb{Y}_\ell
    &=
    \{
    \{\rho_I\}_{|I|\le \ell}\,|\, 
    \rho_I \succeq 0\,,\quad
    \Tr(\rho_I)=1\,,\quad 
    \Tr_{I\setminus J}(\rho_I) 
    =\rho_J \,,\text{ for all }J\subseteq I
    \}
\end{align}
We consider the functional $\omega_\ell(A)=\Tr(\rho_I A)$ for operators $A$ supported on intervals $I$, for $|I|\le \ell$. This is linearly extended to the vector space 
$\mathcal{V}_\ell$
of operators that are sum of operators supported on intervals, which includes the Hamiltonian under consideration.
Then we define the energy of the relaxation by
\begin{align}
    f^*_\ell
    &=
    \min_{\{\rho_I\}\in \mathbb{Y}_\ell} 
    \omega_\ell(H)    
\end{align}
We prove here the following result.

\begin{theorem}
Given the Hamiltonian \eqref{eq:H0tV}, satisfying the assumptions above, 
denote its ground state energy by $E_0=\lambda_{\text{min}}(H)$.
Then there exists a system-size independent $t_*>0$ such that for all $0<t<t_*$,  
\begin{align}
    0\le 
    \frac{E_0-f^*_\ell}{n}
    \le 
    c(t) e^{-\ell/8}
    \,,
\end{align}
where $c(t)$ is a positive function of $t$ which is independent of the system size $n$.
\end{theorem}

We next present the proof, which consists of two steps. First, the derivation of a quasi-local sum-of-squares decomposition and then the introduction of an extension of the local relaxation functional and the final proof of exponential convergence.

\subsubsection{Quasi-local sum-of-squares decomposition}

\cite{froehlich2019lieschwingerblockdiagonalizationgappedquantum} construct a unitary $U$ that block-diagonalises $H$ in the basis spanned by the ground state of $H_0$ and its orthogonal complement.
This section reviews their result and uses it to establish a quasi-local sum-of-square decomposition of $H-E_0\id$.

The unitary $U$ is a sequence of gates, each associated with an interval $I$ (i.e.~a connected set of sites in the chain).
To describe it, denote  the interval $\{q,q+1,\dots,q+k\}$
by $(k;q)$.
Then, with standard notation $[n]=\{1,\dots,n\}$, 
there is an interval for each $k\in [n-1]$ and $q\in [n-k]$.
The number of intervals is thus $n(n-1)/2=\mathcal{O}(n^2)$.
There is a strict total order on the set of intervals: $(k';q')\succ (k;q)$ if $k'>k$ or if $k'=k$ and $q'>q$.
$U$ is then a product of gates acting on intervals in the order given by this total ordering
\begin{align}
    U = U_{(n-1;1)}\cdots U_{(1;2)} U_{(1;1)}\,.
\end{align}
$U_{(k;q)}$ is supported on $(k;q)$. 
With $U_{I}=e^{S_{I}}$ 
for interval $I$, there exists a $t_0>0$, $n$-independent, such that
for $t<t_0$,
\cite[Lemma A.3]{froehlich2019lieschwingerblockdiagonalizationgappedquantum}
\begin{align}
    \label{eq:SI}
    \|S_I\|
    \le 
    \frac{8C_St}{m^2} a^{m-2}\,,\quad 
    a\equiv t^{1/3}\,,\quad 
    m\equiv |I|\,.
\end{align}
We define also for a set $I$, the orthogonal projectors
\begin{align}
    p_I = \bigotimes_{i\in I}(\ket{\Omega}\bra{\Omega})_i
    \,,\quad q_I = \id - p_I
    \,.
\end{align}
\cite[Thm 3.4]{froehlich2019lieschwingerblockdiagonalizationgappedquantum} then shows that for $t<t_0$ we  have
\begin{align}
    \tilde{H}=UHU^\dagger
    =
    \sum_{i=1}^n h_i + t \sum_{I,\, |I|\ge 2}
    V_I
\end{align}
where the sum is over intervals of length at least $2$.
$V_I$ is Hermitian and supported on $I$, and is block diagonal,
\begin{align}
    V_I = p_I V_I p_I + q_I V_I q_I\,,
\end{align}
with norm bounded as
\begin{align}
    \|V_I\| \le \frac{8}{m^2} a^{m-2}\,.
\end{align}
\cite[Thm 3.5]{froehlich2019lieschwingerblockdiagonalizationgappedquantum} shows that the ground state of $\tilde{H}$ is $\ket{\Omega}$, so that the ground state energy of $\tilde{H}$ (and thus of $H$) is
\begin{align}
    E_0=\bra{\Omega}\tilde{H}\ket{\Omega}
    =
    \sum_{I,\, |I|\ge 2}
    e_I
    \,,\quad 
    e_I = t\Tr(V_I p_I)
    \,.
\end{align}
Next, we show that this result can be cast into a sum-of-squares decomposition of $\tilde{H}-E_0\id$. We have
\begin{align}
    \tilde{H}-E_0\id 
    =
    \sum_{i=1}^n h_i +  \sum_{I,\, |I|\ge 2}
    W_I
    \,,\quad 
    W_I \equiv tV_I-e_I\id\,. 
\end{align}
Note that $W_I$ has zero component along $p_I$, because since $p_I$ is rank $1$, we have $p_I tV_I p_I = e_I p_I$, so
\begin{align}
    W_I = 
    e_I p_I + q_I tV_I q_I -e_I (p_I+q_I)
    =
    q_I W_I q_I\,.
\end{align}
$W_I$ has spectral norm
\begin{align}
    \|W_I\| \le t\|V_I\| + |e_I|
    \le 2 t\|V_I\| \le b_I \equiv 
    \frac{16 t}{m^2} a^{m-2}\,.
\end{align}
By definition of spectral norm, $W_I \succeq -\|W_I\| q_I$, and thus we have
\begin{align}
    W_I \succeq -b_I q_I \succeq 
    -b_I \sum_{i\in I}q_i
\end{align}
where we used the inequality
\begin{align}
    q_I \preceq \sum_{i\in I}q_i
\end{align}
since the l.h.s.~is one in a product state if all terms equal $\ket{\Omega}_i$ while the r.h.s.~is one if at least one of the terms equal $\ket{\Omega}_i$.
In a sum-of-square decomposition we want to write $\tilde{H}-E_0\id = \sum_\gamma T_\gamma$, with $T_\gamma \succeq 0$.
We thus define the positive operator for $I, |I|\ge 2$:
\begin{align}
    T_I = W_I + b_I \sum_{i\in I}q_i\succeq 0\,.
\end{align}
Summing and subtracting the second term in the current decomposition, we have by rearranging the sums (setting $b_I$ to be zero for sets that are not intervals of length $\ge 2$)
\begin{align}
    \tilde{H}-E_0\id 
    =
    \sum_{i=1}^n h_i +  
    \sum_{I,\, |I|\ge 2}
    T_I - \sum_{I}
    b_I \sum_{i\in I}q_i
    =
    \sum_{i=1}^nh_i +  
    \sum_{I,\, |I|\ge 2}
    T_I 
    - \sum_{i=1}^n \beta_i q_i
    \,, \quad \beta_i \equiv 
    \sum_{I\ni i}b_I\,.
\end{align}
Now we claim that
\begin{align}
    T_i = h_i - \beta_i q_i \succeq 0\,,\quad 
    \text{if } \beta_i \in [0,1] \,.
\end{align}
Indeed, by assumption $h_i \succeq q_i$, so as long as $\beta_i\le 1$, $T_i$ is positive.
Note also that $\|T_i\| = \|h_i\|-\beta_i \le \|h_i\|\le h_*$.
To summarise, we have achieved
\begin{align}
    \tilde{H}-E_0\id 
    &=
    \sum_{i=1}^n T_i +  
    \sum_{I,\, |I|\ge 2}
    T_I 
    \\
    T_i, T_I&\succeq 0\,,\quad
    \label{eq:T_gamma_bounds}
    \|T_i\|\le h_*\,,\quad 
    \|T_I\| \le (m+1)b_I\,.
\end{align}
Since $h_i \succeq q_i$, then
$T_i \succeq (1-\beta_i)q_i$, and
we also have the lower bound $\tilde{H}-E_0\id\succeq \sum_{i}(1-\beta_i)q_i$.
If we take $\beta_i\le \beta_*\equiv 1/2$, we thus have that the gap 
of $\tilde{H}$ (and thus of $H$)
is at least $1/2$.

Now we want to use this decomposition to infer a quasi-local sum-of-squares decomposition for $H-E_0\id$. Define $\alpha(A) = U^\dagger A U$.
Then 
\begin{align}
    \label{eq:HalphaT}
    H-E_0\id 
    =
    \sum_\gamma \alpha(T_\gamma)
    \,.
\end{align}
Clearly $\alpha(T_\gamma)\succeq 0$, but its support is the whole lattice.
However, the structure of $U$ is such that the resulting operator is quasi-local, namely the error we obtain in approximating it with its truncation to a ball of radius $r$ around the support of $T_\gamma$ is exponentially small in $r$.
This follows from applying the Lieb-Robinson bound \cite{hastings2010localityquantumsystems,Nachtergaele_2019} to a fictitious time-dependent evolution that generates $U$.
The details are presented in the following Lemma.
We will denote by $\Lambda$ the set of sites of the chain and by $d(I,J)=\min_{x\in I,y\in J}|x-y|$ the distance between two sets of sites.

\begin{lemma}
    \label{lemma:LRalpha}
    Let $C_S>0$ be the constant in \eqref{eq:SI}.
    Then there exists a $t_S>0$ such that for all $0\le t<t_S$ we have for an operator $A$ supported on interval $I$, $I=m$:
    \begin{align}
        \| \alpha(A) - \mathbb{E}_{B_r(I)}(\alpha(A)) \|
        \le D m \|A\| e^{-r} 
        \,.
    \end{align}
    Here $\mathbb{E}_{J}(A) = \tr_{J^c}(A)\otimes \id_{J^c}$ and $\tr$ the normalised trace,
    $B_r(I)=\{ x\in\Lambda \,|\, d(I,x)\le r\}$ is a ball of radius $r$ around $I$, and $D>0$ is a $n$-independent constant.
\end{lemma}

\begin{proof}
We first recast the statement to be proven in terms of commutators following a standard procedure \cite{hastings2010localityquantumsystems,Nachtergaele_2019}.
Let $\dd U$ be the Haar measure over the unitary group supported on $J^c$ for a set $J$. Then 
\begin{align}
    \mathbb{E}_J(A) = \int \dd U UAU^\dagger\,.
\end{align}
This is due to the invariance of Haar measure under conjugation, as for the partial trace.
Then we have
\begin{align}
    \label{eq:approx_Haar}
    \| A - \mathbb{E}_J(A) \|
    =
    \| \int \dd U  (A - UAU^\dagger) \|
    \le 
    \int \dd U 
    \| 
    A - UAU^\dagger \|
    =
    \int \dd U 
    \| 
    U[A, U^\dagger] \|
    =
    \int \dd U 
    \| 
    [A, U^\dagger] \|
    \,.
\end{align}
So in the following, we will look at the more general problem of bounding the commutator of $\alpha(A)$ 
for $A$ supported on $X$ with operators $B$ supported on $Y$.

For notation simplicity we denote the gates in $U$ as $U=e^{S_L}\cdots e^{S_2}e^{S_1}$, with $S_j^\dagger = -S_j$ and denote the support of $S_j$ by $Z_j$.
To make use of Lieb-Robinson bounds we define a fictitious time evolution $U(s)$ such that
$U(0)=\id$, $U(L)=U$. We define for $s\in (j-1,j]$:
\begin{align}
    U(s) = e^{(s-j+1) S_j} \cdots  
    e^{S_2}e^{S_1}\,
\end{align}
and $\alpha_s(A) = U(s)^\dagger A U(s)$.
We denote the time-dependent Hamiltonian generating this evolution by $\Phi(s)$: $U(s)'=-i\Phi(s)U(s)$, 
\begin{align}
    \Phi(s) 
    =
    \sum_Z \Phi_Z(s)
    =
    \begin{cases}
        iS_j & s\in (j-1,j]\\
        0 & \text{ otherwise }
    \end{cases}
\end{align}
where the sum over $Z$ is over sets.
We then want to bound the commutator
\begin{align}
    C_B(X,s)
    =
    \sup_{A\in \mathcal{A}_X}
    \frac{\|[\alpha_s(A),B] \|}{\|A\|}
\end{align}
where $\mathcal{A}_X$ is the set of (non-zero) operators supported on $X$.
In particular, we will derive an iterative expression that bounds $C_B(X,s)$ in terms of $C_B(Z,s)$ for $Z$ sets connecting $X$ to $Y$.
We start to derive and solve the ODE satisfied by $g_A(s) := [\alpha_s(A),B] $. First, the Heisenberg evolution is
\begin{align}
    \frac{\dd}{\dd s}\alpha_s(A) 
    &=
    +iU(s)^\dagger \Phi(s)AU(s)
    -iU(s)^\dagger A\Phi(s)U(s)
    =i\alpha_s([\Phi(s),A])
    =
    i[K_X(s), \alpha_s(A)]
    \\
    K_X(s) &= \sum_{Z\in S(X)}
    \alpha_s(\Phi_Z(s))
    \,,\quad 
    S(X) = \{ Z \,|\, Z\cap X \neq \emptyset \}\,.
\end{align}
Using Jacobi identity,
we have
\begin{align}
    g_A(s)'
    =i[[K_X(s),\alpha_s(A)],B]
    =
    i[K_X(s),g_A(s)]
    +
    R_A(s)
    \,,\quad 
    R_A(s)
    =
    -i[\alpha_s(A),[B,K_X(s)]]
    \,.
\end{align}
To solve this equation we use variation of constants.
Thus we define $G_A(s) = V(s)^\dagger g_A(s) V(s)$, where $V(s)'=iK_X(s)V(s)$ is the time evolution that corresponds to the homogenous equation with $R_A$ set to zero. It satisfies
$G_A(s)' = V^\dagger(s) R_A(s) V(s)$, so
its solution is
\begin{align}
    G_A(s)
    =
    G_A(0)
    +
    \int_0^s \dd u V(u)^\dagger
    R_A(u)V(u)
\end{align}
and so
\begin{align}
    g_A(s)
    =
    V(s)V(0)^\dagger g_A(0)
    V(0)V(s)^\dagger +
    \int_0^s \dd u 
    V(s)V(u)^\dagger
    R_A(u)V(u)V(s)^\dagger
    \,.
\end{align}
The norm of the remainder can be bounded as
\begin{align}
    \|R_A(s)\|
    \le 
    2 
    \|A\|
    \|[K_X(s), B]\|
    \le 
    2
    \|A\|
    \sum_{Z\in S(X)}
    \| [\Phi_Z(s), B] \|
    \le 2 \|A\|
    \sum_{Z\in S(X)}
    \|\Phi_Z(s)\|
    C_B(Z,s)\,,
\end{align}
which implies, dividing the expression for $g_A(s)$ by $\|A\|$:
\begin{align}
    C_B(X,s)
    \le 
    C_B(X,0)
    +
    \sum_{Z\in S(X)}
    \int_0^s \dd u
    f_Z(u) C_B(Z,u)
    \,,\quad 
    f_Z(u) = 2\| \Phi_Z(u) \| \,.
\end{align}
The next step is to iterate this bound:
\begin{align}
    C_B(X,s)
    &\le 
    C_B(X,0)
    +
    \sum_{Z_1\in S(X)}
    \int_{0\le u_1 \le s} \dd u_1
    f_{Z_1}(u_1) C_B(Z_1,0)\\
    &\quad +
    \sum_{Z_1\in S(X)}
    \sum_{Z_2\in S(Z_1)}
    \int_{0\le u_2 \le u_1\le s} 
    \dd u_1 \dd u_2
    f_{Z_1}(u_1)
    f_{Z_2}(u_2)
    C_B(Z_2,u_2)
    \\
    &\le 
    C_B(X,0)
    +
    \sum_{Z_1\in S(X)}
    \int_{0\le u_1 \le s} \dd u_1
    f_{Z_1}(u_1)
    C_B(Z_1,0)\\
    &\quad +
    \sum_{Z_1\in S(X)}
    \sum_{Z_2\in S(Z_1)}
    \int_{0\le u_2 \le u_1\le s} 
    \dd u_1 \dd u_2
    f_{Z_1}(u_1)
    f_{Z_2}(u_2)
    C_B(Z_2,0)
    \\&\quad +
    \sum_{Z_1\in S(X)}
    \sum_{Z_2\in S(Z_1)}
    \sum_{Z_3\in S(Z_2)}
    \int_{0\le u_3\le u_2 \le u_1\le s} 
    \dd u_1 \dd u_2 \dd u_3
    f_{Z_1}(u_1)
    f_{Z_2}(u_2)
    f_{Z_3}(u_3)
    C_B(Z_3,u_3)
    \,.
\end{align}
If we iterate $p$ steps, we have:
\begin{align}
    C_B(X,s)
    &\le 
    \sum_{q=0}^{p-1}a_q(X,s)
    +
    N_p(X,s)\,,\\
    a_q(X,s)&=
    2\|B\|
    \sum_{Z_1,\dots,Z_q\in C_q(X,Y)}
    \prod_{i=1}^q 
    \int_{0}^L 
    f_{Z_i}(u_i)
    \dd u_i
    \\
    N_p(X,s)&=2\|B\|
    \sum_{Z_1,\dots,Z_p\in C_p(X)}
    \prod_{i=1}^p 
    \int_{0}^L 
    f_{Z_i}(u_i)    \dd u_i
    \,.
\end{align}
We used 
that $C_B(Z,u)\le 2\|B\|$ and
that $C_B(Z,0)\le 2 \| B\| $, if
$Z\cap Y\neq \emptyset$,
and zero otherwise, and denoted by $C_q(X,Y)$ the chain of sets $(Z_1,\dots,Z_q)$ such that 
$Z_1\cap X\neq \emptyset$,
$Z_q\cap Y\neq \emptyset$
and
$Z_i\cap Z_{i-1}\neq \emptyset$
for $i=2,\dots, q-1$ 
and $C_q(X,Y)$ the same but without the constraint on $Y$.
We have also enlarged the integration to $[0,L]$ for each variable, which, since integrand is positive, 
produces an upper bound.

Next, define the matrix indexed by sites of the lattice
\begin{align}
    M_{xy} = \sum_{Z\ni xy} \int_0^L\dd u f_Z(u)
    =
    2
    \sum_{j|Z_j\ni xy} 
    \| S_j \|
    \,.
\end{align}
Note that $(M^q)_{xy}$ sums over all chains at least once. For example,
\begin{align}
    (M^3)_{xy}
    =
    \sum_{Z_1\ni xz_1}
    \sum_{Z_2\ni z_1z_2}
    \sum_{Z_3\ni z_2y}
    \prod_{i=1}^3
    \int_0^L\dd u_i f_{Z_i}(u_i)
\end{align}
and the summands contain at least once  all possible chains of size $3$ between $x$ and $y$. Therefore,
we can upper bound the sums over chains by these matrix powers:
\begin{align}
    a_q(X,s)\le 2\|B\|\sum_{x\in X}
    \sum_{y\in Y} (M^q)_{xy}
    \,,\quad 
    N_p(X,s)
    \le 
    2\|B\|\sum_{x\in X}
    \sum_{y\in \Lambda} (M^p)_{xy}
    \,.
\end{align}

Next we prove a bound on the column sums of $M$. Let
\begin{align}
    \kappa
    =
    2\sup_{x\in \Lambda}
    \sum_{j\,|\, x\in Z_j}
    |Z_j| e^{\text{diam}(Z_j)}
    \|S_j\|
    <1\,.
\end{align}
We will prove later that for $t$ sufficiently small this holds and just assume it for the moment.
Recall that for an interval, $\text{diam}(I) = |I|-1$.
Then
\begin{align}
    \sup_{x\in\Lambda} \sum_{y\in \Lambda}
    M_{xy}
    \le 
    \sup_{x\in\Lambda} \sum_{y\in \Lambda}
    e^{|x-y|}
    M_{xy}
    =
    2\sup_{x\in\Lambda} 
    \sum_{j|Z_j\in x}
    \sum_{y\in Z_j}
    e^{|x-y|}
    \|S_j\|
    \le 
    2\sup_{x\in\Lambda} 
    \sum_{j|Z_j\in x}
    |Z_j|
    e^{\text{diam}(Z_j)}
    \|S_j\|
    =\kappa\,,
\end{align}
so that $\kappa$ is an upper bound on the column sums of $M$.
From triangle inequality on $|x-y|$ we have
\begin{align}
    \sum_{y}
    e^{|x-y|}
    (M^{q+1})_{xy}
    \le 
    \sum_{z}
    e^{|x-z|}
    (M)_{xz}
    \sum_y
    e^{|z-y|}
    (M^{q})_{zy}
    \,.
\end{align}
This shows by induction that
\begin{align}
\sum_{y}(M^{q})_{xy}
    \le 
    \sum_{y}e^{|x-y|} (M^{q})_{xy}
    \le \kappa^q    \,.
\end{align}
Also, we have a tighter bound on the column sum of $M^q$ when $x\in X$:
\begin{align}
    \sum_{y\in Y}
    (M^q)_{xy}
    =
    \sum_{y\in Y}
    e^{-|x-y|}
    e^{|x-y|}
    (M^q)_{xy}
    \le 
    e^{-\delta}
    \sum_{y\in\Lambda}
    e^{|x-y|}
    (M^q)_{xy}
    \le 
    e^{-\delta}
    \kappa^q\,,
\end{align}
with $\delta = \text{dist}(X,Y)$.
This shows that the remainder term goes to zero for $p\to\infty$:
\begin{align}
    \lim_{p\to \infty}N_p(X,s)
    =
    2\|B\|\sum_{x\in X}
    \lim_{p\to \infty}
    \sum_{y\in \Lambda}
    (M^p)_{xy}
    \le 
    2\|B\||X|
    \lim_{p\to \infty}
    \kappa^p = 0\,.
\end{align}
Finally, this leads to the bound on the commutator
\begin{align}
    C_B(X,s)
    &\le 
    2\|B\|
    \sum_{q=1}^\infty 
    \sum_{x\in X}\sum_{y\in Y}
    (M^q)_{xy}
    \le 
    2\|B\||X|
    \sum_{q=1}^\infty 
    \sup_{x\in X} 
    \sum_{y\in \Lambda}
    (M^q)_{xy}
    \le 
    2\|B\||X|e^{-\delta}
    \sum_{q=1}^\infty 
    \kappa^q \\
    &=
    D\|B\||X|e^{-\delta}    
    \,,\quad 
    D=\frac{2\kappa}{1-\kappa}\,.
\end{align}
Here we started the sum from $1$ since we assume that $X\cap Y=\emptyset$ as is the case in our application, where $X=I$ and $Y=B_r(I)$.
Together  with eq.~\eqref{eq:approx_Haar} implies that, denoted by $\mathcal{U}_J$ the set of unitaries on $J$,
\begin{align}
    \| 
    \alpha(A)
    -
    \mathbb{E}_{B_r(I)}(\alpha(A))
    \|
    \le 
    \sup_{U\in \mathcal{U}_{B_r(I)^c}} \|[\alpha(A),U]\|
    \le 
    \|A\|
    \sup_{U\in \mathcal{U}_{B_r(I)^c}}
    C_{U}(I, L)
    \le 
    D\|A\| m e^{-r}
    \,.
\end{align}
Finally, we need to compute $\kappa$ and verify that it is smaller than $1$.
Using \eqref{eq:SI}, the fact that the intervals involved in $U$ have length
$m\ge 2$ and that the number of intervals containing $x$ is at most $m$,
\begin{align}
    \kappa
    &=
    2\sup_{x\in \Lambda}
    \sum_{j\,|\, x\in Z_j}
    |Z_j| e^{\text{diam}(Z_j)}
    \|S_j\|
    =
    2\sup_{x\in \Lambda}
    \sum_{m= 2}^n
    m e^{m-1}
    \sum_{\substack{I,
    |I|=m, x\ni I}}
    \|S_j\|
    \\
    &\le 
    \frac{16C_st}{ea^2}
    \sum_{m\ge 2}
    (ea)^{m}
    =
    \frac{16C_st}{ea^2}
    \frac{e^2a^2}{1-ea}
    =
    \frac{16C_ste}{1-ea}
    \equiv \kappa_*
    \,.
\end{align}
Thus by choosing $t$ such that $\kappa_*<1$, we guarantee quasi-locality of $\alpha(A)$, concluding the proof of the Lemma.

\end{proof}

\subsubsection{Extension of local relaxation and hierarchy bound}

The linear functional $\omega_\ell$ used in the level $\ell$ relaxation is defined on $\mathcal{V}_\ell$, the space of operators which are linear combinations of those whose support is contained in intervals of size $\ell$.
Thus we cannot apply $\omega_\ell$ to Eq.~\eqref{eq:HalphaT} directly, since the r.h.s.~involves operators outside this set.
To remedy this, we introduce an extension of $\omega_\ell$, denoted $\hat{\omega}_\ell$ that is defined on the set of all operators, and which agrees with $\omega_\ell$ on $\mathcal{V}_\ell$.
Its definition involves projectors $\Delta_{[a,b]}$ onto operators whose support is exactly contained in intervals $[a,b]$ and we first introduce those.

We choose a basis of the space of operators consisting of the identity and strings of traceless operators $P$ (for qubits, Pauli strings), so that we can write
\begin{align}
    A = \tr(A) \id 
    +
    \sum_{P\neq \id} c_P P
    \,,
\end{align}
where we recall that $\tr$ is the normalised trace.
Now, given $a\le b$, we group terms $P$ such that the support $S$ of $P$ is such that $\min S = a, \max S=b$:
\begin{align}
    A = \tr(A) \id 
    +
    \sum_{1\le a\le b\le n}
    \sum_{\substack{P\neq \id\\ 
    \min S = a, \max S=b}} c_P P
    \,.
\end{align}
Let us then define
\begin{align}
    \Delta_{[a,a]}
    &=
    \mathbb{E}_{\{a\}}
    -
    \mathbb{E}_{\emptyset}
    \\
    \Delta_{[a,b]}
    &=
    \mathbb{E}_{[a,b]}
    -
    \mathbb{E}_{[a+1,b]}
    -
    \mathbb{E}_{[a,b-1]}
    +
    \mathbb{E}_{[a+1,b-1]}
    \,,\quad 
    a<b
    \,.
\end{align}
Here $\mathbb{E}_I(A)
=\tr_{I^c}(A)\otimes \id_{I^c}$ is the conditional expectation used above, which satisfies for a traceless operator basis $P$ of support $S$:
\begin{align}
    \mathbb{E}_I(P)
    =
    \begin{cases}
        P &\text{ if }S\subseteq I
        \\
        0 &\text{otherwise}
    \end{cases}
\end{align}
so that $\mathbb{E}_I$ projects onto traceless strings $P$ whose support is in $I$.
Then $\Delta_{[a,a]}$
projects onto traceless operators whose support is exactly $\{a\}$ -- subtracting $\mathbb{E}_{\emptyset}(A)=\tr(A) \id$ removes the identity component. 
Similarly, $\Delta_{[a,b]}$
projects onto traceless operators $P$ whose support $S$ is such that $\min S=a,\max S=b$. For example, we can check explicitly that if $P=X_2$, we have
\begin{align}
    \Delta_{[2,2]}(X_2) 
    = X_2 - 0 = X_2
    \,,\quad 
    \Delta_{[2,5]}(X_2) 
    = X_2 - 0 - X_2 + 0 = 0
\end{align}
and if $P=X_2Y_3Z_5$, 
\begin{align}
    \Delta_{[2,5]}(X_2Y_3Z_5) 
    = X_2Y_3Z_5 - 0 - 0 + 0 = X_2Y_3Z_5
    \,.
\end{align}
With these definitions, note that
\begin{align}
    \Delta_{[a,b]}
    \Big( 
    \sum_{\substack{P\neq \id\\ 
    \min S = c, \max S=d}} c_P P
    \Big)
    =
    \delta_{a,c}\delta_{b,d}
    \sum_{\substack{P\neq \id\\ 
    \min S = c, \max S=d}} c_P P
    \,,
\end{align}
so we can write the decomposition of $A$ as
\begin{align}
    \label{eq:ADelta}
    A = \tr(A) \id 
    +
    \sum_{1\le a\le b\le n}
    \Delta_{[a,b]}(A)
    \,.
\end{align}

Now define the linear functional on the space of all operators
\begin{align}
    \label{eq:omegahat}
    \hat{\omega}_\ell(A)
    =
    \tr(A)
    +\sum_{
    \substack{
    1\le a\le b\le b\\
    b-a+1\le \ell
    }
    }
    \omega_\ell(\Delta_{[a,b]}(A))
    \,.
\end{align}
It is clear if $A\in \mathcal{V}_\ell$ then the decomposition \eqref{eq:ADelta} involves only terms $(a,b)$ with $b-a+1\le \ell$ so that in this case, $\hat{\omega}_\ell=\omega_\ell$.
Next, we prove a useful result that relates
$\hat{\omega}_\ell(A)$ to the norm of $A$.

\begin{lemma}
    \label{lemma:hatomega_norm}
    Suppose that $A$ is supported on an interval $I$ of size $m$.
    Then
    \begin{align}
        |\hat{\omega}_\ell(A)|
        \le 
        p(m) \|A\|
        \,,\quad 
        p(m) = 1+2m^2\,.
    \end{align}
\end{lemma}
\begin{proof}
    If $A$ has support on an interval of size $m$, then 
    in the sum \eqref{eq:omegahat} 
    at most $m$ components when $a=b$ contribute, and at most $m(m-1)/2$ components when $a<b$.
    Thus
    \begin{align}
        |\hat{\omega}_\ell(A)|
        \le 
        \|A\|
        +
        m 2 \|A\|
        +
        \frac{1}{2}m(m-1)
        4 \|A\|
        =
        p(m) \|A\|
        \,,
    \end{align}
    where we also used that
    $\mathbb{E}_{[a,b]}$
    is a contraction.
\end{proof}

Next we introduce the telescopic sum that represents $\alpha(A)$
for an operator $A$ supported on interval $I$
in terms of truncations
\begin{align}
    \alpha(A)
    =
    \alpha(A)_0
    +
    \sum_{s\ge 0}
    (\alpha(A)_{s+1}-\alpha(A)_s)
    \,,\quad 
    \alpha(A)_s
    =
    \mathbb{E}_{B_s(I)}(
    \alpha(A))\,.
\end{align}
Note that the sum is finite since the terms for $s\ge n$ are zero.
We have the following properties for $\hat{\omega}_\ell(\alpha(A))$.

\begin{lemma}
    Let $A$ be supported on an interval $I$ of size $m$. Then, there is a constant $G>0$ such that
    \begin{align}
        \label{eq:m_long}|\hat{\omega}_\ell(\alpha(A))|
        \le 
        G(1+m)^3\|A\|
        \,.
    \end{align}
    Further, if $A\succeq 0$, $m\le \ell/2$, 
    there is a constant $T>0$ such that
    \begin{align}
        \label{eq:m_short}
        \hat{\omega}_\ell(\alpha(A))\ge 
        -
        Te\|A\|m(\ell+2)^2
        e^{-\ell/4}
    \end{align}
\end{lemma}

\begin{proof}
    Using the telescopic decomposition
    and Lemma \ref{lemma:hatomega_norm} we have
    \begin{align}
        |\hat{\omega}_\ell(\alpha(A))|
        &\le 
        |\hat{\omega}_\ell(\alpha(A)_0)|
        +
        \sum_{s\ge 0}
        |\hat{\omega}_\ell(\alpha(A)_{s+1}-\alpha(A)_{s})|
        \\
        &\le 
        p(m)\|A\|
        +
        \sum_{s\ge 0}
        p(m+2(s+1))
        \|\alpha(A)_{s+1}-\alpha(A)_{s}\|
        \\
        &\le 
        p(m)\|A\|
        +
        2Dm\|A\|
        \sum_{s\ge 0}
        p(m+2(s+1))
        e^{-s}
        \,,
    \end{align}
    where in the last inequality we used Lemma
    \ref{lemma:LRalpha}.
    Next, to simplify the expression we use the bounds
    \begin{align}
        p(m)\le 3(1+m)^2
        \,,\quad 
        p(m+2(s+1))
        \le 
        12 (m+1)^2 (s+1)^2
    \end{align}
    so that 
    \begin{align}
        |\hat{\omega}_\ell(\alpha(A))|
        &\le 
        3(1+m)^2\|A\|
        +
        24Dm
        (m+1)^2
        \|A\|
        \sum_{s\ge 0}
        (s+1)^2
        e^{-s}
        \\
        &\le 
        (1+m)^3\|A\|
        G\,,\quad 
        G=
        3+24 DF\,,\quad 
        F=
        \sum_{s\ge 0}
        (s+1)^2
        e^{-s}   
        \,.
    \end{align}
    For the second statement, let us assume that $A\succeq 0$ and $m\le \ell/2$.
    Then $r=\lfloor (\ell-m)/2\rfloor$ is such that the support of $\alpha(A)_r$ is $2r+m\le \ell$, so that
    $\hat{\omega}_\ell(\alpha(A)_r)=
    \omega_\ell(\alpha(A)_r)\succeq 0$.
    Then, by expanding $\alpha(A)$ starting from $\alpha(A)_r$ rather than $\alpha(A)_0$, we have
    \begin{align}
        |\hat{\omega}_\ell(\alpha(A)-\alpha(A)_r)|
        &\le 
        2Dm\|A\|
        \sum_{s\ge r}
        p(m+2(s+1))
        e^{-s} 
        \\
        &=
        2Dm\|A\|
        e^{-r}
        \sum_{k\ge 0}
        p(m+2k+2r+2)
        e^{-k} 
        \,.
    \end{align}
    Then using the bound
    \begin{align}
        p(m+2k+2r+2)
        \le 
        3(m+2k+2r+2)^2
        \le 
        3(m+2r+2)^2(k+1)^2
        \,,
    \end{align}
    we have
    \begin{align}
        |\hat{\omega}_\ell(\alpha(A)-\alpha(A)_r)|
        &\le 
        2Dm\|A\|
        e^{-r}
        3(m+2r+2)^2
        \sum_{k\ge 0}
        (k+1)^2
        e^{-k} 
        =
        6DF\|A\|
        m(m+2r+2)^2
        e^{-r}
        \\
        &\le 
        Te
        \|A\|
        m(\ell+2)^2
        e^{-\ell/4}
        \,,\quad 
        T=6DF
        \,,
    \end{align}
    where we used that $m\le \ell/2$ so that 
    $r\ge \ell/4-1$.
    Thus,
    \begin{align}
        \hat{\omega}_\ell(\alpha(A))
        \ge 
        \omega_\ell(\alpha(A)_r)
        -
        Te\|A\|m(\ell+2)^2
        e^{-\ell/4}
        \,,
    \end{align}
    and the Lemma follows from the positivity assumption.
\end{proof}

Finally, we are going to use these properties to control the error of the relaxation.
We have
\begin{align}
    \omega_\ell(H)-E_0 
    =
    \sum_\gamma 
    \hat{\omega}_\ell(\alpha(T_\gamma))
    =
    \sum_{m_\gamma\le \ell/2}
    \hat{\omega}_\ell(\alpha(T_\gamma))
    +
    \sum_{m_\gamma> \ell/2}
    \hat{\omega}_\ell(\alpha(T_\gamma))
    \,.
\end{align}
We are going to use the estimate \eqref{eq:m_short}
for the first sum.
We have, recalling the bounds on $\|T_\gamma\|$ of Eq.~\eqref{eq:T_gamma_bounds},
\begin{align}
    \sum_{m_\gamma\le \ell/2}
    \hat{\omega}_\ell(\alpha(T_\gamma))
    &\ge 
    -Te(\ell+2)^2 e^{-\ell/4}
    \sum_{m_\gamma\le \ell/2}
    m_\gamma \|T_\gamma\|
    \\
    &\ge
    -Te(\ell+2)^2 e^{-\ell/4}
    n
    \Big(h_*
    +
    16 t
    \sum_{m\ge 2}
    \frac{(m+1)}{m}
    a^{m-2}
    \Big)
    \,.
\end{align}
We will bound the convergent sum later.
Next, we look at the second sum, where we use the estimate \eqref{eq:m_long}:
\begin{align}
    \sum_{m_\gamma> \ell/2}
    \hat{\omega}_\ell(\alpha(T_\gamma))
    &\ge 
    - \sum_{m_\gamma> \ell/2}
    |\hat{\omega}_\ell(\alpha(T_\gamma))|
    \ge 
    -G
    \sum_{m_\gamma> \ell/2}
    (m_\gamma+1)^3\|T_\gamma\|
    \\
    &=
    -G
    \sum_{m_\gamma> \ell/2}
    e^{-m_\gamma}e^{+m_\gamma}
    (m_\gamma+1)^3\|T_\gamma\|
    \\
    &\ge
    -G
    e^{-\ell/2}
    \sum_{m_\gamma\ge 1}
    e^{+m_\gamma}
    (m_\gamma+1)^3\|T_\gamma\|
    \ge -Gn
    e^{-\ell/2}\mathcal{M}(t)
    \\
    \mathcal{M}(t)
    &=
    8eh_*
    +
    16 t
    \sum_{m\ge 2}
    e^{+m}
    \frac{(m+1)^4}{m^2}a^{m-2}
    \,.
\end{align}
Now note that we can also bound the series appearing in the bound on the short operators ($m_\gamma\le \ell/2$) by $\mathcal{M}(t)$, so that:
\begin{align}
    \omega_\ell(H)-E_0 
    \ge
    -n\mathcal{M}(t)
    (Te(\ell+2)^2 e^{-\ell/4}
    +G
    e^{-\ell/2}
    )
    \,.
\end{align}
Finally, minimising over density matrices and
using that for $x\ge 0$, $(x+2)^2 e^{-x/8}\le 45$ (which can be shown by computing the maximum occurring at $x=14$), we have
\begin{align}
    0\le E_0 - f^*_\ell \le 
    n\mathcal{M}(t)
    (45Te
    +G
    )e^{-\ell/8}
    \,.
\end{align}
This proves theorem \ref{thm:exp_conv_1d_weak}.

\subsection{Proofs from Section \ref{section:qmp}} \label[appendix]{appendix:section_qmp}

\convergencesubgradient*

\begin{proof}

We follow the proof technique of \cite{boyd2003subgradient}. For a constant step size $\alpha_\ell = h > 0$, the error between $f_1^*$ and $f_{\mathrm{subgrad}} \coloneqq \max_{0 \le t < T} Q \left(\nu^{(t)} \right)$ 
satisfies 
\begin{align}
\frac{f_1^* - f_{\mathrm{subgrad}}}{n}
&\le \frac{ \mathrm{dist}^2 \left(\nu^{(0)}, \arg \max_\nu Q \left (\nu \right) \right) + \sum_{\ell=0}^{T-1} \alpha_\ell^2 \, \left \| g^{(\ell)} \right \|_F^2  }{2 n \sum_{\ell=0}^{T-1} \alpha_\ell}
% &\le \frac{ \mathrm{dist}^2 \left( \nu^{(0)}, \arg \max_\nu Q(\nu) \right) + \sum_{\ell=1}^k  h^2 \, \| g^{(\ell)} \|_F^2  }{2 kh N} \\ 
% &\le \frac{ R^2 + B_g^2 h^2 T}{2 n h T } 
\le \frac{R^2}{2 n h T} + \frac{ B_g^2 h}{2 n},
\end{align}
\noindent
where $R \coloneqq \mathrm{dist} \left(\nu^{(0)}, \arg \max_\nu Q \left (\nu \right) \right)$ and $B_g \ge \left \| g^{(\ell)} \right \|_F$ for all $\ell$ is an upper-bound on the Frobenius norm of the subgradient for all iterations.

We next bound $B_g$. Since 
$g^{(\ell)} \coloneqq \bigoplus_{i \in V} \bigoplus_{j \in N(i)} g_{j \to i}^{(\ell)}$, 
we have 
$\left \| g^{(\ell)} \right \|_F^2 = \sum_{i \in V} \sum_{j \in N(i)} \left \| g_{j \to i}^{(\ell)}\right \|_F^2$. For each directed edge $j \to i$, the corresponding subgradient is \begin{equation}
    g_{j \to i}^{(\ell)} = \mu_i^{(k)} - \Tr_j \mu_{ij}^{(k)},
\end{equation}
where $\mu_i^{(\ell)}$ and $\mu_{ij}^{(\ell)}$ are ground-state density matrices of the respective local effective Hamiltonians. For any two density matrices $\rho$, $\sigma$ on the same space, \begin{equation}
    \left \| \rho - \sigma \right \|_F^2 = \Tr \left [\rho^2 \right ] + \Tr \left [\sigma^2 \right ] - 2 \Tr \left [\rho \sigma \right ] \le 1 + 1 = 2.
\end{equation}
\noindent
Hence each subgradient component satisfies \begin{equation}
    \left \| g_{j \to i}^{(\ell)} \right \|_F^2 \le 2.
\end{equation}
\noindent
Summing over all directed edges gives \begin{equation}
    \left \| g^{(\ell)} \right \|_F^2 = \sum_{i \in V} \sum_{j \in N(i)}  \left \| g_{j \to i}^{(\ell)} \right \|_F^2 \le 2 n \Delta_G,
\end{equation}
\noindent
where $\Delta_G \coloneqq \max_{i \in  V} |N(i)|$. Therefore,  $B_g = \sqrt{2 n \Delta_G}$ uniformly bounds the Frobenius norm of the subgradients for all iterations. 

The above error bound thus becomes \begin{align}
\frac{f_1^* - f_{\mathrm{subgrad}}}{n}
\le \frac{R^2}{2 n h T} + h \Delta_G.
\end{align}
\noindent
Let $h = c \epsilon$, where $0 < c < 1/\Delta_G$ is a constant that is independent of both the system size $n$ and $\epsilon$. The error bound is then at most $\epsilon$ for any positive integer $T$ such that 
\begin{align}
    T \ge \left \lceil \frac{R^2}{2c \left( 1-c \Delta_G \right ) n \epsilon^2} \right \rceil.
\end{align}
\noindent
This proves the stated sufficient iteration bound.

Below we derive the per-iteration time complexity of \cref{alg:subgradient}. In each iteration, the algorithm samples a ground-state eigenvector of each of the $n$ node effective Hamiltonians and the $|E|$ edge effective Hamiltonians. These matrices have dimensions $d$ and $d^2$, respectively, where $d$ is the dimension of the local Hilbert space per site. Using a dense eigensolver, this amounts to a combined per-iteration cost of \begin{equation}
    \mathcal O \left (n d^3 + |E| (d^2)^3 \right ) = \mathcal O \left (n d^3 + |E| d^6 \right).
\end{equation}

\end{proof}

\hfill

\smootherror*

\begin{proof}
    We first bound the smoothing error introduced by replacing local ground state energy with its finite-temperature free energy approximation at inverse temperature $\beta > 0$. Let the eigenvalues of $K \in \mathrm{Herm}(\mathbb C^{d \times d})$ be $\lambda_1 \le \lambda_2 \le \cdots \le \lambda_d$. We have \begin{align}
        F_\beta(K) 
        &= - \frac{1}{\beta} \log \Tr \left [ e^{- \beta K} \right] % \\
        = - \frac{1}{\beta} \log \left( \sum_{k=1}^d e^{ - \beta \lambda_k} \right) % \\
        = - \frac{1}{\beta} \log \left( e^{ - \beta \lambda_1} \sum_{k=1}^d e^{ - \beta (\lambda_k - \lambda_1)} \right) \\
        &= \lambda_1 - \frac{1}{\beta} \log \left( \sum_{k=1}^d e^{ - \beta (\lambda_k - \lambda_1)} \right). 
    \end{align}
    
    Since $\lambda_k-\lambda_1\ge 0$ for all $k$, each term in the final sum is at most one, while the $k=1$ term equals to one. The sum therefore lies between $1$ and $d$. As $\beta>0$, it follows that \begin{equation}
        F_\beta(K) \le \lambda_{\min} (K) \le F_\beta(K) + \frac{1}{\beta} \log d.
    \end{equation}

    % Since $\lambda_k  - \lambda_1 \ge 0$ for all $k$, 
    % $1 \le  \sum_{k=1}^d e^{ - \beta (\lambda_k - \lambda_1)} \le d$ and therefore
    % $0 \le \frac{1}{\beta} \log \left( \sum_{k=1}^d e^{ - \beta (\lambda_k - \lambda_1)} \right ) \le \frac{1}{\beta} \log d$. By introducing $1/\beta$ and minimum eigenvalue $\lambda_1$, we get 
    % $\lambda_1 - \frac{1}{\beta} \log d \le F_\beta(K) \le \lambda_1$, which gives us the bound on the smoothing error between local ground state energy and its free energy approximation,
    % \begin{equation}
    %     F_\beta(K) \le \lambda_{\min} (K) \le F_\beta(K) + \frac{1}{\beta} \log d.
    % \end{equation}

    % Since $\lambda_k  - \lambda_1 \ge 0$ for all $k$, \begin{align}
    %     1 \le  \sum_{k=1}^d e^{ - \beta (\lambda_k - \lambda_1)} \le d,
    % \end{align} and therefore \begin{align}
    %     0 \le \frac{1}{\beta} \log \left( \sum_{k=1}^d e^{ - \beta (\lambda_k - \lambda_1)} \right ) \le \frac{1}{\beta} \log d.
    % \end{align} 
    % \noindent
    % By introducing $1/\beta$ and minimum eigenvalue, we get \begin{equation}
    %     \lambda_1 - \frac{1}{\beta} \log d \le F_\beta(K) \le \lambda_1,
    % \end{equation} which gives us the bound on the smoothing error between local ground state energy and its free energy approximation, \begin{equation}
    %     F_\beta(K) \le \lambda_{\min} (K) \le F_\beta(K) + \frac{1}{\beta} \log d.
    % \end{equation}

    Hence, for the full smoothed dual function, \begin{align}
        Q_\beta(\nu) \le Q(\nu) \le Q_\beta(\nu) + \frac{1}{\beta} \Gamma_G, 
    \end{align} where $\Gamma_G = \sum_{(ij) \in E} \log (d_i d_j) + \sum_{i \in V} \log (d_i)$ with $d_i = \dim (\mathcal H_i)$ the dimension of the local Hilbert space at site $i \in V$.
\end{proof}

\hfill

In order to prove \cref{thm:convergence_smoothed_gradient_complete}, we must first find the Lipschitz constant of the gradients of finite-temperature free energy function $F_\beta(H)$ (\cref{lem:Lipschitz_smoothness_of_the_local_free_energy_approximation}), before finding the Lipschitz constant of the gradients of the smooth dual function $Q_\beta(\nu)$ (\cref{lem:upper_bound_on_the_Lipschitz_constant}), which we sketch out below.

%\begin{lemma}[Lipschitz smoothness of the local free energy approximation]
\begin{restatable}[Lipschitz smoothness of the local free energy approximation]{lemma}{Lipschitzsmoothnessofthelocalfreeenergyapproximation}
\label{lem:Lipschitz_smoothness_of_the_local_free_energy_approximation}
    For all $K, K' \in \mathrm{Herm}(\mathbb C^{d \times d})$ and $\beta > 0$, \begin{equation}
        \| \nabla F_\beta(K) - \nabla F_\beta(K') \|_F \le \frac{\beta}{m_F} \| K - K' \|_F,
    \end{equation} 
    \noindent
    where $m_F = 2$ is the strong concavity parameter of $S(\rho)$ with respect to the Frobenius norm.
\end{restatable}

    \begin{proof}
        Let \begin{equation}
           J_K(\rho) = \Tr [\rho K] - \frac{1}{\beta} S(\rho)
        \end{equation} be the non-equilibrium free energy of the state $\rho$ at inverse temperature $\beta$.  Since $S(\rho)$ is $m_F$-strongly concave with respect to the Frobenius norm, $- S(\rho) / \beta$ is $(m_F/\beta)$-strongly convex. And since $\Tr [\rho K]$ is linear in $\rho$, $J_K(\rho)$ is also $(m_F/\beta)$-strongly convex with respect to the Frobenius norm.
    
        This means that, for any density matrices $\sigma$ and $\rho_\beta$ the minimiser of the non-equilibrium free energy, we have 
        \begin{align}
            J_K(\sigma) 
            &\ge J_K(\rho_\beta) + \langle \nabla J_K(\rho_\beta) , \sigma - \rho_\beta \rangle + \frac{m_F}{2 \beta} \| \sigma - \rho_\beta \|_F^2 \\
            &\ge J_K(\rho_\beta) + \frac{m_F}{2 \beta} \| \sigma - \rho_\beta \|_F^2,
        \end{align} where the second inequality comes from the fact that $\langle \nabla J_K(\rho_\beta) , \sigma - \rho_\beta \rangle \ge 0$ as $\rho_\beta$ is the minimiser.
    
        Thus, for two Gibbs states $\rho_\beta(K)$ and $\rho_\beta(K')$, we have the following relations: 
        \begin{align}
            J_K(\rho_\beta(K')) &\ge J_K(\rho_\beta(K)) + \frac{m_F}{2 \beta} \| \rho_\beta(K') - \rho_\beta(K) \|_F^2, \\
            J_{K'}(\rho_\beta(K)) &\ge J_{K'}(\rho_\beta(K')) + \frac{m_F}{2 \beta} \| \rho_\beta(K) - \rho_\beta(K') \|_F^2 .
        \end{align} 
        Combining the two, we have 
        \begin{align}
            J_K(\rho_\beta(K')) + J_{K'}(\rho_\beta(K)) - J_K(\rho_\beta(K)) - J_{K'}(\rho_\beta(K')) \ge  \frac{m_F}{\beta} \| \rho_\beta(K) - \rho_\beta(K') \|_F^2.
        \end{align} 
        Since the entropy terms in the l.h.s. cancel out, this becomes
        \begin{align}
            \left \langle \rho_\beta(K) - \rho_\beta(K'), K' - K \right \rangle \ge  \frac{m_F}{\beta} \| \rho_\beta(K) - \rho_\beta(K') \|_F^2.
        \end{align} 
        
        By the Cauchy--Schwarz inequality $|\langle A, B \rangle| \le \| A \|_F \| B \|_F$, we have \begin{align}
            \| \rho_\beta(K) - \rho_\beta(K') \|_F \| K - K' \|_F &\ge  \frac{m_F}{\beta} \| \rho_\beta(K) - \rho_\beta(K') \|_F^2. % \\
            % \Longrightarrow \| K - K' \|_F &\ge  \frac{m_F}{\beta} \| \rho_\beta(K) - \rho_\beta(K') \|_F.
        \end{align}
        If $\rho_\beta(K) = \rho_\beta(K')$, the desired inequality is immediate; otherwise, we divide both side by $\| \rho_\beta(K) - \rho_\beta(K') \|_F$ and obtain \begin{equation}
            \| \rho_\beta(K) - \rho_\beta(K') \|_F \le \frac{\beta}{m_F} \| K - K' \|_F.
        \end{equation} 
        And as $\rho_\beta(K) = \nabla F_\beta(K)$, we finally get 
        \begin{align}
            \| \nabla F_\beta(K) - \nabla F_\beta(K') \|_F \le \frac{\beta}{m_F} \| K - K' \|_F.
        \end{align}
    \end{proof}

\hfill

With \cref{lem:Lipschitz_smoothness_of_the_local_free_energy_approximation} at hand, we can prove the following lemma.

\begin{restatable}[Lipschitz constant of the smoothed dual function]{lemma}{upperboundonthelipchitzconstant}
    \label{lem:upper_bound_on_the_Lipschitz_constant}

    Let $\mathcal V$ denote the space of traceless Lagrange multipliers. Then, for any $\nu, \nu' \in \mathcal V$,
    \begin{equation}
        \left \| \nabla Q_\beta(\nu) - \nabla Q_\beta(\nu')  \right \|_F \le \widehat L_\beta \left \| \nu - \nu' \right \|_F,
    \end{equation}
    where
    \begin{equation}
        \widehat L_\beta = \frac{\beta}{2} \left( \Delta_G + d \right)
    \end{equation}
    with $\Delta_G \coloneqq \max_{i \in V} | N(i) |$ and 
    $d$ the uniform on-site dimension.
\end{restatable}

\begin{proof}

    Define the linear map $\mathcal B$ from the direct sum of messages to
the direct sum of local effective Hamiltonians by
\begin{equation}
    (\mathcal B X)_i
    =
    \sum_{j \in  N(i)} X_{j\to i},
    \quad
    (\mathcal B X)_{ij}
    =
    -X_{j\to i} \otimes \id_j
    -
    \id_i \otimes X_{i\to j}.
\end{equation}
Thus, $\mathcal B X$ describes the changes in the effective node and edge Hamiltonians induced by a message perturbation
$X$.

Let $\boldsymbol{\rho}_\beta(\nu)$ denote the direct sum of all
node and edge Gibbs states. \cref{lem:Lipschitz_smoothness_of_the_local_free_energy_approximation} states that 
\begin{equation}
    \left\|
        \boldsymbol{\rho}_\beta(\nu)
        -
        \boldsymbol{\rho}_\beta(\nu')
    \right\|_F
    \leq
    \frac{\beta}{2}
    \left\|
        \mathcal B(\nu-\nu')
    \right\|_F.
    \label{eq:block_gibbs_lipschitz}
\end{equation}
Moreover, the chain rule gives
\begin{equation}
    \nabla Q_\beta(\nu)
    =
    \mathcal B^*\boldsymbol{\rho}_\beta(\nu).
    \label{eq:smoothed_dual_chain_rule}
\end{equation}

Hence, it remains to bound the operator norm of $\mathcal B$. For any
traceless collection of messages $X$, the node-wise contributions
satisfy
\begin{equation}
    \sum_{i\in V}
    \left\|
        \sum_{j\in  N(i)} X_{j\to i}
    \right\|_F^2
    \leq
    \Delta_G\|X\|_F^2
\end{equation}
by the Cauchy--Schwarz inequality. For each edge $(i,j)$, tracelessness
gives
\begin{equation}
    \left\langle
        X_{j\to i}\otimes \id_j,\,
        \id_i\otimes X_{i\to j}
    \right\rangle
    =
    \Tr(X_{j\to i})\Tr(X_{i\to j})
    =
    0.
\end{equation}
Together with $\|X\otimes \id \|_F^2=d \|X\|_F^2$, this shows that the
edge-wise contributions sum to $d\|X\|_F^2$. Therefore,
\begin{equation}
    \|\mathcal B X\|_F^2
    \leq
    (\Delta_G+d)\|X\|_F^2,
\end{equation}
and hence
\begin{equation}
    \|\mathcal B\|^2
    \leq
    \Delta_G+d.
    \label{eq:B_norm_bound}
\end{equation}

Combining \cref{eq:block_gibbs_lipschitz,eq:smoothed_dual_chain_rule,eq:B_norm_bound}
and using the fact that $\|\mathcal B^*\|=\|\mathcal B\|$, we thus obtain
\begin{align}
    \|\nabla Q_\beta(\nu)-\nabla Q_\beta(\nu')\|_F
    &\leq
    \frac{\beta}{2}\|\mathcal B\|^2
    \|\nu-\nu'\|_F
    \\
    &\leq
    \frac{\beta}{2}(\Delta_G+d)
    \|\nu-\nu'\|_F.
\end{align}
Hence,
\begin{equation}
    \widehat L_\beta
    =
    \frac{\beta}{2}(\Delta_G+d).
\end{equation}

Finally, we remark on the tracelessness of our Lagrange multipliers $\nu$. As specified in \cref{alg:accelerated-smoothed-quantum-dual}, all Lagrange multipliers are initialised as zero, and every component of the gradient satisfies
\begin{align}
        \Tr \left[
        {\rho_{\beta}}_i
        -
        \Tr_j \left ({\rho_{\beta}}_{ij} \right)
    \right]
    =
    1-1
    =
    0.
\end{align}
And since all subsequent gradient and extrapolation steps in \cref{alg:accelerated-smoothed-quantum-dual} are just linear combinations of traceless matrices, the Lagrange multipliers $\nu$ always stay traceless.

\end{proof}

\hfill

With \cref{lem:bound_on_smoothing_error_2_local} and \cref{lem:upper_bound_on_the_Lipschitz_constant} now at hand, we are finally ready to prove \cref{thm:convergence_smoothed_gradient_complete}.

\convergencesmoothedgradient*

\begin{proof}

Define $f_1^* = Q^* \coloneq \max_{\nu \in \mathcal V} Q (\nu )$ and $Q_\beta^* \coloneq \max_{\nu \in \mathcal V} Q_\beta(\nu)$, where $\mathcal V$ is the traceless Lagrange multiplier space. Define also $f_{\mathrm{grad}} \coloneq Q_\beta \left ( \nu^{(T)} \right)$. The total error bound of \cref{alg:accelerated-smoothed-quantum-dual} is 
\begin{align}
    f_1^* - f_{\mathrm{grad}} = \left(  Q^*  - Q_\beta^* \right) + \left( Q_\beta^* - Q_\beta \left ( \nu^{(T)} \right) \right),
\end{align} where the first term represents our smoothing error and the second term represents our approximation error resulting from Nesterov's accelerated gradient descent. We now bound each term separately below.

Define $\nu^* = \arg \max_{\nu \in \mathcal V} Q(\nu)$ and $\nu_\beta^* = \arg \max_{\nu \in \mathcal V} Q_\beta(\nu)$. For the first-term smoothing error, \cref{lem:bound_on_smoothing_error_2_local} implies \begin{align}
    Q_\beta^* = Q_\beta \left (\nu_\beta^* \right ) \le Q \left (\nu_\beta^* \right) \le Q \left(\nu^* \right) = Q^* 
\end{align} and \begin{align}
    Q^* = Q \left (\nu^* \right) \le Q_\beta \left (\nu^* \right) + \frac{\Gamma_G}{\beta} \le Q_\beta \left (\nu_\beta^* \right) + \frac{\Gamma_G}{\beta} = Q_\beta^* + \frac{\Gamma_G}{\beta},
\end{align}
\noindent
where $\Gamma_G = \sum_{(i, j) \in E} \log (d^2) + \sum_{i \in V} \log (d)$ and $d$ is the uniform on-site dimension. Combining the two inequalities thus gives us \begin{align}
    0 \le Q^* - Q_\beta^* \le \frac{\Gamma_G}{\beta}.
\end{align}

For the second-term Nesterov's accelerated optimisation error, Corollary 1 of \cite{tseng2008accelerated} states the following convergence guarantee 
\begin{align}
    Q_\beta^* - Q_\beta \left ( \nu^{(T)} \right) \le 
    \frac{2 \widehat L_\beta R_\beta^2}{\left( T+1 \right)^2},
\end{align} where $\widehat L_\beta$ is the Lipschitz constant of the gradients of $Q_\beta$, $R_\beta = \mathrm{dist} \left( \nu^{(0)}, \arg \max_{\nu \in \mathcal V} Q_\beta \left (\nu \right) \right)$, and $T$ is the number of iteration ran. Therefore, the total error-density bound is 
\begin{align}
    \frac{f_1^* - f_{\mathrm{grad}}}{n} \le \frac{\Gamma_G}{n \beta} + \frac{2 \widehat L_\beta R_\beta^2}{n \left( T+1 \right)^2}.
\end{align}

In order to ensure that this total error-density does not exceed some target error-density $\epsilon$, we split $\epsilon$ equally between the smoothing error and the approximation error. We choose 
\begin{equation}
    \frac{\Gamma_G}{n \beta} = \frac{\epsilon}{2}
\end{equation} 
for the smoothing error, which determines the value for the $\beta$ input to \cref{alg:accelerated-smoothed-quantum-dual},
\begin{equation}
    \beta = \frac{2 \Gamma_G}{ n \epsilon}.    
\end{equation}
Letting the optimisation error to be
\begin{equation}
    \frac{2 \widehat L_\beta R_\beta^2}{n \left( T+1 \right)^2} \le \frac{\epsilon}{2}
\end{equation} 
and using $\widehat L_\beta = \beta \left(\Delta_G + d \right) / 2$ from \cref{lem:upper_bound_on_the_Lipschitz_constant},
we obtain the following lower bound on the sufficient iteration count $T \in \mathbb N$, \begin{align}
    T \ge \left \lceil \frac{2 R_\beta \sqrt{\Gamma_G \left(\Delta_G + d \right) }}{n \epsilon} \right \rceil - 1.
\end{align} 

Below we derive the per-iteration time complexity of \cref{alg:accelerated-smoothed-quantum-dual}, which follows directly our argument made for the subgradient method listed in the proof for \cref{thm:convergence_subgradient_complete}. Instead of sampling ground-state eigenvectors, \cref{alg:accelerated-smoothed-quantum-dual} computes the Gibbs states of local effective edge and node Hamiltonians, which has the same arithmetic cost of $\mathcal O(d^3)$ for $d$-dimensional matrices. Therefore, the combined per-iteration cost is again \begin{equation}
    \mathcal O \left (n d^3 + |E| (d^2)^3 \right) = \mathcal O \left (n d^3 + |E| d^6 \right).
\end{equation}
\end{proof}

\end{document}